\documentclass[12pt]{article}

\usepackage{hyperref}%
\hypersetup{%
   breaklinks,%
   colorlinks=true,%
   linkcolor=[rgb]{0.45,0.0,0.0},%
   citecolor=[rgb]{0,0,0.45}
}

\usepackage{amsmath}%
\usepackage[amsmath,thmmarks]{ntheorem}
\usepackage{amssymb}%
\theoremseparator{.}
\usepackage{paralist}%
\usepackage{mleftright}%
\usepackage{xspace}%
\usepackage{euscript}%
\usepackage[cm]{fullpage}%
\usepackage{xcolor}%
\usepackage{graphicx}%
\usepackage{framed}%

\numberwithin{figure}{section}%
\numberwithin{table}{section}%
\numberwithin{equation}{section}%

\newcommand{\HLinkShort}[2]{\hyperref[#2]{#1\ref*{#2}}}
\newcommand{\HLink}[2]{\hyperref[#2]{#1~\ref*{#2}}}
\newcommand{\HLinkPage}[2]{\hyperref[#2]{#1~\ref*{#2}%
      $_\text{p\pageref{#2}}$}}
\newcommand{\HLinkPageOnly}[1]{\hyperref[#1]{Page~\refpage*{#1}%
      $_\text{p\pageref{#1}}$}}

\newcommand{\HLinkSuffix}[3]{\hyperref[#2]{#1\ref*{#2}{#3}}}
\newcommand{\HLinkPageSuffix}[3]{\hyperref[#2]{#1\ref*{#2}%
      #3$_\text{p\pageref{#2}}$}}

\newcommand{\figlab}[1]{\label{fig:#1}}
\newcommand{\figref}[1]{\HLink{Figure}{fig:#1}}

\newcommand{\seclab}[1]{\label{sec:#1}}
\newcommand{\secref}[1]{\HLink{Section}{sec:#1}}

\newcommand{\corlab}[1]{\label{cor:#1}}
\newcommand{\corref}[1]{\HLink{Corollary}{cor:#1}}%

\providecommand{\deflab}[1]{\label{def:#1}}

\newcommand{\itemlab}[1]{\label{item:#1}}
\newcommand{\itemref}[1]{\HLinkSuffix{(}{item:#1}{)}}

\newcommand{\remlab}[1]{\label{rem:#1}}
\newcommand{\remref}[1]{\HLink{Remark}{rem:#1}}%

\newcommand{\lemlab}[1]{\label{lemma:#1}}
\newcommand{\lemref}[1]{\HLink{Lemma}{lemma:#1}}%

\newcommand{\thmlab}[1]{{\label{theo:#1}}}
\newcommand{\thmref}[1]{\HLink{Theorem}{theo:#1}}

\providecommand{\eqlab}[1]{}%
\renewcommand{\eqlab}[1]{\label{equation:#1}}
\newcommand{\Eqref}[1]{\HLinkSuffix{Eq.~(}{equation:#1}{)}}
\newcommand{\eqqref}[1]{(\ref{equation:#1})}

\newtheorem{theorem}{Theorem}[section] 

\newtheorem{defn}[theorem]{Definition}

\newtheorem{lemma}[theorem]{Lemma}
\newtheorem*{lemma:u}[theorem]{Lemma}%
\newtheorem{corollary}[theorem]{Corollary}
\theoremstyle{remark}%
\theorembodyfont{\upshape}%
\newtheorem{remark}[theorem]{Remark}%
\newtheorem{example}[theorem]{Example}

\theoremheaderfont{\em}%
\theorembodyfont{\upshape}%
\theoremstyle{nonumberplain}
\theoremseparator{}
\theoremsymbol{\myqedsymbol}
\newtheorem{proof}{Proof:}
\newcommand{\eps}{{\varepsilon}}%

\newcommand{\cardin}[1]{\left| {#1} \right|}%

\newcommand{\pth}[1]{\mleft({#1}\mright)}
\newcommand{\Set}[2]{\left\{ #1 \;\middle\vert\; #2 \right\}}
\newcommand{\brc}[1]{\left\{ {#1} \right\}}

\newcommand{\Ex}[2][\!]{\mathop{\mathbf{E}}#1\pbrcx{#2}}

\renewcommand{\th}{th\xspace}

\newcommand{\Prob}[1]{\mathop{\mathbf{Pr}}\!\pbrcx{#1}}
\newcommand{\pbrcx}[1]{\left[ {#1} \right]}

\newcommand{\myqedsymbol}{\rule{2mm}{2mm}}

\newcommand{\ds}{\displaystyle}%

\newcommand{\range}{\mathbf{r}}%

\newcommand{\CDChar}{\Xi}%
\newcommand{\CD}[1]{\CDChar #1}%

\newcommand{\CDH}[2]{\CDChar_{\geq #2}\pth{ #1}}%

\newcommand{\Family}{\EuScript{F}}%

\newcommand{\cell}{\sigma}
\newcommand{\cellA}{\psi}

\newcommand{\cDim}{b}%
\newcommand{\Dim}{\delta}%

\newcommand{\BVT}{{\sf BVT}}

\newcommand{\DefSet}[1]{\mathsf{D}\pth{#1}}
\newcommand{\KillSet}[1]{\mathsf{K}\pth{#1}}
\newcommand{\clX}[1]{\KillSet{#1}}%

\newcommand{\BadProb}{\varphi}

\newcommand{\RangeSet}{{\mathcal{R}}}

\newcommand{\Growth}[2]{g^{}_{#1}\pth{#2}}

\newcommand{\VCProj}[2]{#1_{|#2}}

\renewcommand{\Re}{\mathbb{R}}%

\def\eps{{\varepsilon}}
\def\bd{{\partial}}

\def\A{{\cal A}}
\def\Arr{{\cal A}}

\def\H{{\cal H}}
\def\nn{{\bf n}}
\def\xx{{\bf x}}
\def\R{{\cal R}}
\def\reals{{\mathbb R}}
\def\xx{{\bf x}}

\def\V{{\sf VD}}

\def\VD{{\sf VD}}
\def\BT{{\sf BVT}}

\definecolor{blue25}{rgb}{0, 0, 11}

\newcommand{\M}{{\cal M}}%
\renewcommand{\O}{{\cal O}}%

\newcommand{\MCY}[2]{%
   \begin{minipage}{#1\linewidth}
       \begin{center}
           \smallskip%
           #2
           \smallskip%
       \end{center}
   \end{minipage}%
}%

\newcommand{\Bezout}{B\'ezout\xspace}
\newcommand{\Gartner}{G\"artner\xspace}
\newcommand{\etal}{\textit{et~al.}\xspace}

\newenvironment{SarielMod}[1][\hsize]{%
   \MakeFramed{\hsize#1\advance\hsize-\width\FrameRestore}%
}%
{\endMakeFramed}

\def\C{{\cal C}}

\begin{document}

\title{Decomposing arrangements of hyperplanes: \\
   VC-dimension, combinatorial dimension, \\
   and point location%
   \thanks{%
      Work by Esther Ezra was partially supported by NSF CAREER under
      grant CCF:AF-1553354 and by Grant 824/17 from the Israel Science Foundation.
      Work by Sariel Har-Peled was partially supported by %
      NSF AF awards CCF-1421231 and CCF-1217462.  %
      Work by Haim Kaplan was supported by Grant 1841/14 from the
      Israel Science Fund and by Grant 1161/2011 from the German
      Israeli Science Fund (GIF).
      Work by Micha Sharir has been supported by Grant 2012/229 from
      the U.S.-Israel Binational Science Foundation, and by Grant
      892/13 from the Israel Science Foundation.
      Work by Haim Kaplan and Micha Sharir was also supported by
      the Israeli Centers for Research Excellence (I-CORE) program
      (center no.~4/11), by the Blavatnik Computer Science Research
      Fund at Tel Aviv University, and by the Hermann
      Minkowski--MINERVA Center for Geometry at Tel Aviv University.
      %
   }%
}

\author{Esther Ezra\thanks{%
    Dept.~of Computer Science, Bar-Ilan University, Ramat Gan, Israel and
    School of Mathematics, Georgia Institute of Technology, Atlanta,
    Georgia 30332, USA.  \texttt{eezra3@math.gatech.edu}.  } \and
  Sariel Har-Peled\thanks{%
    Department of Computer Science, University of Illinois, 201
    N.~Goodwin Avenue, Urbana, IL, 61801, USA. %
    \texttt{sariel@illinois.edu}. } \and Haim Kaplan\thanks{%
    Blavatnik School of Computer Science, Tel Aviv University, Tel
    Aviv 69978 Israel.  \texttt{haimk@post.tau.ac.il}.  } \and Micha
  Sharir\thanks{%
    Blavatnik School of Computer Science, Tel Aviv University, Tel
    Aviv 69978 Israel.  \texttt{michas@post.tau.ac.il}.  } }

\maketitle


\begin{abstract}
    This work is motivated by several basic problems and techniques
    that rely on space decomposition of arrangements of hyperplanes in
    high-dimensional spaces, most notably Meiser's 1993 algorithm for
    point location in such arrangements. A standard approach to these
    problems is via random sampling, in which one draws a random
    sample of the hyperplanes, constructs a suitable decomposition of
    its arrangement, and recurses within each cell of the decomposition
    with the subset of hyperplanes that cross the cell. The efficiency
    of the resulting algorithm depends on the quality of the sample,
    which is controlled by various parameters.

    One of these parameters is the classical \emph{VC-dimension}, and
    its associated \emph{primal shatter dimension}, of a suitably
    defined corresponding range space.  Another parameter, which we
    refer to here as the \emph{combinatorial dimension}, is the maximum number
    of hyperplanes that are needed to define a cell that can arise in
    the decomposition of some sample of the input hyperplanes; this
    parameter arises in Clarkson's (and later Clarkson and Shor's)
    random sampling technique.%

    We re-examine these parameters for the two main space
    decomposition techniques---\emph{bottom-vertex triangulation}, and
    \emph{vertical decomposition}, including their explicit dependence on
    the dimension $d$, and discover several unexpected
    phenomena, which show that, in both techniques, there are large
    gaps between the VC-dimension (and primal shatter dimension), and
    the combinatorial dimension.

    For vertical decomposition, the combinatorial dimension is only
    $2d$, the primal shatter dimension is at most $d(d+1)$,
    and the VC-dimension is at least $1 + d(d+1)/2$ and at most $O(d^3)$.  For
    bottom-vertex triangulation, both the primal shatter dimension and
    the combinatorial dimension are $\Theta(d^2)$, but there seems to be
    a significant gap between them, as the combinatorial dimension is
    $\frac12d(d+3)$, whereas the primal shatter dimension is at most $d(d+1)$,
    and the VC-dimension is between $d(d+1)$ and $5d^2 \log{d}$ (for $d\ge 9$).

    Our main application is to point location in an arrangement of $n$
    hyperplanes is $\Re^d$, in which we show that the query cost in
    Meiser's algorithm can be improved if one uses vertical
    decomposition instead of bottom-vertex triangulation, at the cost
    of some increase in the preprocessing cost and storage. The best query
    time that we can obtain is $O(d^3\log n)$, instead of
    $O(d^4\log d\log n)$ in Meiser's algorithm.  For these bounds to
    hold, the preprocessing and storage are rather large
    (super-exponential in $d$). We discuss the tradeoff between query cost
    and storage (in both approaches, the one using bottom-vertex trinagulation
    and the one using vertical decomposition).

    Our improved bounds rely on establishing several new structural
    properties and improved complexity bounds for vertical decomposition,
    which are of independent interest, and which we expect to find additional applications.

    The point-location methodology presented in this paper can be adapted to
    work in the \emph{linear decision-tree model}, where we are only concerned
    about the cost of a query, and measure it by the number of point-hyperplane sign tests
    that it performs. This adaptation is presented in the companion paper~\cite{ES17},
    where it yields an improved bound for the linear decision-tree complexity of $k$-SUM.
    A very recent breakthrough by Kane et al.~\cite{KLM17} further improves the bound,
    but only for special, ``low-complexity'' classes of hyperplanes. We show here
    how their approach can be extended to yield an efficient and improved point-location
    structure in the RAM model for such collections of hyperplanes.
\end{abstract}

\section{Introduction}

\paragraph{Point location.}
This work is motivated by several basic problems and techniques that
rely on space decomposition of arrangements of hyperplanes in
high-dimensional spaces, most notably (i) Meiser's 1993 algorithm for
point location in such arrangements~\cite{m-plah-93}, and (ii) the linear
decision tree complexity of $k$-SUM, {\sc SubsetSum}, {\sc Knapsack}, and related problems.

Let $H$ be a set of $n$ hyperplanes in $\Re^d$, and let $\Arr(H)$ denote the arrangement of $H$.
The \emph{point-location problem} in $\Arr(H)$ is to preprocess $H$ into a
data structure that supports efficient point-location queries, each of
which specifies a point $q\in\Re^d$ and asks for the (not necessarily
full-dimensional) cell of $\Arr(H)$ that contains $q$.
We represent the output cell $C$ by its sign pattern with respect to the hyperplanes in $H$, where the sign of $C$ with respect to
a hyperplanes $h\in H$ is $0$ if $h$ contains $C$, $+1$ if $C$ is in the positive side of $H$ and $-1$ if $C$ is in the negative side
of $h$.
A simpler
variant, known as \emph{vertical ray-shooting}, which does not seem to
make the problem substantially simpler, is to ask for the hyperplane
of $H$ that lies directly above $q$, namely, the first hyperplane that
an upward $x_d$-vertical ray emanating from $q$ will hit (it is
possible that the ray terminates right away, when $q$ lies on one or several such hyperplanes).

\paragraph{Linear decision tree complexity of $k$-SUM and related problems.}
The $k$-SUM problem is to determine, for a given set $A$ of $n$ real numbers and
a parameter $k$, whether there exist $k$ (say, distinct) elements of $A$ that sum up to $0$.
A simple transformation turns this into a point location problem: Let $H$
be the set of all hyperplanes in $\Re^n$ of the form $x_{i_1}+x_{i_2} + \cdots + x_{i_k}=0$,
over all the $\binom{n}{k}$ $k$-tuples $1\le i_1 < i_2 < \cdots < i_k\le n$.
Map the given set $A=\{a_1,a_2,\ldots,a_n\}$ to the point $q_A = (a_1,a_2,\ldots,a_n)\in\Re^n$.
The $k$-SUM problem is then reduced to the problem of determining whether $q_A$ lies
on a hyperplane of $H$, an instance of point location amid hyperplanes in a high-simensional space.

We analyze the complexity of this problem
in the \emph{linear decision tree} model, where we only count linear sign tests involving the
elements of $A$ (and no other operation is allowed to access the actual values of these elements).
This problem is studied in the companion paper~\cite{ES17}, where a full treatment and analysis are provided.

\paragraph{The challenge.}
In the context studied in this paper (as well as in Meiser's work), $d$ is not assumed to be a
constant, so the dependence of the bounds that we will derive on $d$
is crucial for calibrating the quality of our solutions.  This is a
departure from classical approaches to this problem in computational geometry,
in which $d$ is assumed to be a small constant, and constants that depends solely on $d$
(sometimes exponentially or worse) are ignored (i.e., hidden in the
$O$-notation). In particular, we cannot use off-the-shelf results,
in which the dependency on $d$ is not made explicit, as a black box.

\paragraph{Time-space tradeoff.}
As is typical with problems of this kind, there is a trade-off between
the cost of a query and the storage (and preprocessing time) required
by the structure. We are interested here in solutions in which the
query time is polynomial in $d$ and (poly)logarithmic in $n$. The only
known algorithm that achieves this goal is an old algorithm from 1993,
due to Meiser~\cite{m-plah-93} (see also the related work~\cite{Meyer84} under the decision-tree model).
We review Meiser's algorithm in \secref{meiser-bvt}, add a few called-for enhancements,
and show that it can achieve query time of $O(d^4 \log n)$, and that the storage,
for this query cost,\footnote{%
   This large dependence of the storage on $d$ was missed in the
   analysis in \cite{m-plah-93}; it can be reduced by increasing the
   query time---see \secref{meiser-bvt} for full details.}
is $O(n^{2d\log d + O(d)})$. The storage has been tightened (in terms of its dependence on $n$, but not
on $d$) to $O(d^{O(d^3)}n^d)$ in a follow-up work by Liu~\cite{l-nplah-04}.

\paragraph{Canonical decompositions.}
A standard approach to these problems is via random sampling, in which one draws a random
sample of the hyperplanes, constructs a suitable decomposition of its arrangement, and
recurses within each cell of the decomposition. The efficiency of the resulting algorithm
depends on the quality of the sample, which is controlled by various parameters, discussed shortly.

Before discussing the random sampling methodologies that turn this
intuition into rigorous theory, we briefly review the two basic
techniques for decomposing an arrangement of hyperplanes into subcells of
``constant'' description complexity (here the qualifier ``constant''
is misleading, because it depends on $d$, which is not assumed to be
constant, but this is how the standard terminology goes). These are
\emph{bottom-vertex triangulation} and \emph{vertical decomposition}.

In the bottom-vertex triangulation approach (\BVT{} for short; see
\cite{m-ldg-02}), we triangulate (recursively) each facet in the
arrangement (within the hyperplane containing it), and then
triangulate $d$-space, by picking the lowest vertex $w$ in each full-dimensional
cell $C$, and by connecting $w$ to each of the $(d-1)$-simplices in the
triangulations of the facets of $C$ that are not adjacent to $w$ (some
care is needed when triangulating the unbounded cells of $\Arr(H)$);
see \secref{b:v:t} for more details.

The other scheme, \emph{vertical decomposition} (\VD{} for short; see
\cite{cegs-sessr-91,sa-dsstg-95}), is an extension to higher
dimensions of the standard two-dimensional vertical decomposition of
arrangements of lines or curves (see, e.g.,~\cite{as-aa-00, bcko-cgaa-08,sa-dsstg-95}).
It is constructed by a careful and somewhat involved recursion on the
dimension, and decomposes $\Arr(H)$ into box-like vertical prisms of
some special structure; see \cite{sa-dsstg-95} and \secref{v:de} for
full details.

\paragraph{Cuttings and random sampling.}
The key tool for fast point location is to construct a canonical decomposition
of $\Re^d$ into cells, such that, for a given parameter $1<r<n$, the following properties hold.
\begin{compactenum}[\quad (a)]
    \item Each cell in the decomposition is ``simple'', and, in particular,  it is
    relatively easy (and efficient, when $r$ is small) to find the cell containing the query point.
    \item The total number of cells in the decomposition is small (in
    particular, it only depends on $r$).
    \item For each cell $\cell$, the set of hyperplanes of $H$ that cross $\cell$, called the
 {\em conflict list} of $\cell$, and denoted $\clX{\cell}$, is of size at most $n/r$.
\end{compactenum}
Such a decomposition is called a \emph{$(1/r)$-cutting}.  See
\cite{c-c-05,cf-dvrsi-90, bs-ca-95, h-cctp-00} and references therein for
more information about cuttings.

The point location algorithm then recursively
constructs the same data structure for each cell $\cell$, on the set $\clX{\cell}$ of hyperplanes,
and the overall tree-like structure allows us to locate a query point $q$ in $\Arr(H)$ efficiently.
Meiser's point location algorithm, as well as many other point location algorithms, follow this standard approach.

The main general-purpose tool for constructing a cutting is to take a random sample $R$ of $H$
of a suitable size (that depends on $r$), and argue that a canonical decomposition $\Xi$ of $\Arr(R)$
(namely, either the \BVT{} or the \VD{} of $R$) is a $(1/r)$-cutting of $\Arr(H)$, with sufficiently large probability.

There are two basic techniques for estimating the size of the sample
that is required to guarantee this property. They are based, respectively, on
the $\eps$-net theorem~\cite{hw-ensrq-87}, and on the Clarkson-Shor random sampling theory~\cite{cs-arscg-89}.
We devote \secref{sect2} for the analysis of these sampling methods.
We first briefly review the first approach, and then spend most of the section on a
careful analysis of various aspects of the Clarkson-Shor theory. Although we present it in more generality,
we are mainly interested in its specialization to the case of hyperplanes in high-dimensional spaces.

\paragraph{Range spaces and growth functions.}
A key concept in the analysis around the $\eps$-net theorem is the notion of a \emph{range space}.
In general, this is a pair $(H,\R)$, where $H$ is some ground set (in
our case, a set of hyperplanes in $\Re^d$), and $\R$ is a family of
subsets of $H$, called \emph{ranges}. In our settings, a range is the
subset of all hyperplanes of $H$ that cross some region that ``looks
like'' a cell of the decomposition; that is, an arbitrary simplex for
bottom-vertex triangulation, and an arbitrary vertical prism (with
some constraints on its structure, discussed in \secref{v:de})
for vertical decomposition.

These range spaces are hereditary in nature, in the sense that each
subset $H' \subseteq H$ defines an induced (or projected) range space
$(H', \R')$, where $\R' = \VCProj{\R}{H'} = \Set{\range \cap H'}{\range \in \R}$.
Of key interest is the dependence of (an upper bound on) the number of ranges in this
range space on the size of $H'$.  Formally, we define the
(global) \emph{growth function} of $(H,\R)$ to be
$\Growth{\R}{m} = \ds \max_{H' \subseteq H,\; \cardin{H'}= m} \cardin{\VCProj{\R}{H'}}$.

We say that a growth function is \emph{$(\alpha,\beta)$-growing}, for
a real parameter $\alpha$ and an integer parameter $\beta\ge 1$, if
$\Growth{\R}{m} \leq 2^\alpha m^{\beta}$ (either for all $m$, or for all
$m$ larger than some constant threshold). As will be reviewed below,
the growth function of a range space is a key tool in obtaining
estimates on the size of a random sample $R$ that guarantees that the
corresponding decomposition of $\Arr(R)$ is a $(1/r)$-cutting.  We note that in
most previous studies, the parameter $\alpha$ is ignored, since it is considered to be
a constant, and as such has little effect on the analysis.
Here, in contrast, $\alpha$ will depend on $d$, which is not a
constant, and this will lead to some ``nonstandard'' estimates on the
performance of random sampling.

\paragraph{VC-dimension and primal shatter dimension.}
The VC-dimension of a range space $(H,\R)$ is the size of the largest
subset $H'$ of $H$ that can be \emph{shattered} by $\R$; that is, each
of the $2^{|H'|}$ subsets of $H'$ is a range in the restriction of
$\R$ to $H'$. In general, by the Sauer--Shelah lemma (see, e.g.,
\cite{pa-cg-95,h-gaa-11}), if $(H,\R)$ has finite VC-dimension
$\delta$ (namely, a value independent of $|H|$, which is the case
in the special instances considered here), then the growth function satisfies
\begin{equation} \eqlab{gr:ss}
\Growth{\R}{m} \le \sum_{j=0}^\delta \binom{m}{j} \le 2 \pth{\frac{me}{\delta}}^\delta ,
\end{equation}
where the right inequality holds for $m\ge\delta$.
The \emph{primal shatter dimension} $\delta_0$ of $(H,\R)$ is the smallest integer $\beta$
for which there exists an $\alpha\ge 0$ such that $g_\R$ is $(\alpha,\beta)$-growing.
Informally, $\delta_0$ is obtained (a) by obtaining the best upper bound on $g_{\R}(m)$,
in terms of $m$, and then (b) by stripping away the dependence of $g_\R$ on the corresponding parameter $\alpha$.
In our contexts, $\alpha$ may be large (but it does not depend on $m$), so this simplified
notation makes more sense as $m$ gets larger.
Clearly, \eqqref{gr:ss} implies that $\delta_0 \leq \delta$,
and in general they need not be equal.  For example, for the range space in which the
ground set consists of points in the plane and the ranges are (subsets
obtained by intersecting the ground set with) unit disks, the
VC-dimension is $3$ but the primal shatter dimension is only $2$.
We also note that reducing $\delta_0$ (down from $\delta$, when possible) might
incur an increase in $\alpha$.

\paragraph{$\eps$-nets.}
These notions of dimension are used to control the quality of the decomposition,
via the theory of \emph{$\eps$-nets}. For a general range space
$(H,\R)$ and a prespecified parameter $\eps\in (0,1)$, a subset
$N\subseteq H$ is an \emph{$\eps$-net} if every ``heavy'' range $R\in\R$,
of size at least $\eps|H|$, contains an element of $N$.
Equivalently, every range that is disjoint from $N$ is of size smaller
than $\eps |H|$.  The celebrated theorem of Haussler and Welzl~\cite{hw-ensrq-87},
or rather its enhancement by Blumer et al.~\cite{Blumer}
(stated as \thmref{epsilon:net} below) asserts that, if $(H,\R)$ is a range space of
finite VC-dimension $\delta$, then a random sample of
$\frac{c\delta}{\eps} \log \frac{1}{\eps}$ elements of $H$, for some
suitable absolute constant $c$, is an $\eps$-net with constant
probability.  As follows from the proof of this theorem, a similar
assertion holds (more or less) when we replace $\delta$ by the
(potentially smaller) primal shatter dimension $\delta_0$, except that the
required sample size is now $\frac{c\delta_0}{\eps} \log \frac{\delta_0}{\eps}$, so
the dimension also appears inside the logarithm.

In our context, if the sample $R$ is of size $c\delta r \log r$ (resp.,
$c\delta_0 r \log (\delta_0 r)$), where $\delta$ (resp., $\delta_0$) is the
VC-dimension (resp., primal shatter dimension) of the corresponding range space
(where ranges are defined by simplices or by vertical prisms), then,
with constant probability, (the interior of) each cell in the decomposition of $\Arr(R)$,
which, by construction, is not crossed by any hyperplane of $R$, is crossed by
fewer than $n/r$ hyperplanes of $H$, so the decomposition is a $(1/r)$-cutting of $\Arr(H)$.
In particular, choosing $r=2$, say, we guarantee, with constant probability, that if we choose $R$ to be of
size $2c\delta$ (resp., $2c\delta_0\log(2\delta_0)$), then (with constant probability)
the conflict list of each of the cells of the decomposition is of size at most $n/2$.

In the point location application, presented in \secref{sec:meiser}, this choice of $|R|$,
namely, the one that guarantees (with constant probability), that the size of each conflict
list goes down by a factor of at least $2$,
leads to the best bound for the cost of a query that this method
yields, but it results in a rather large bound on the storage size.
To decrease storage, we need to use significantly larger values of $r$,
that is, larger samples. Loosely speaking, a larger sample
size better controls the size of the recursive subproblems, and leads to smaller storage size,
but one pays for these properties in having to spend more time to locate the query point
in the canonical decomposition of the sample. See \secref{sec:meiser} for more details.

\paragraph{The Clarkson--Shor theory and combinatorial dimension.}
In contrast with the general approach based on range spaces,
Clarkson~\cite{c-narsc-87} (and later Clarkson and
Shor~\cite{cs-arscg-89}) developed an alternative theory, which, in
the context considered here, focuses only on simplices or vertical
prisms (referred to by the common term `cell' in what follows, not
to be confused with the undecomposed cells of the arrangement) that
can arise in the actual canonical space decomposition of the
arrangement of some sample $R$ of $H$. Each such cell $\cell$ has, in
addition to its conflict list $\clX{\cell}$, defined as above, also a
\emph{defining set} $\DefSet{\cell}$, which is the smallest subset
$H'$ of $H$ (or a smallest subset, in case of degeneracies),
for which $\cell$ is a cell in the decomposition of
$\Arr(H')$. We refer to the cells that arise in this manner as
\emph{definable cells}.  Clearly, each definable cell determines a
range in the corresponding range space $(H,\R)$, but, as will follow
from our analysis, not necessarily vice versa.

Define the \emph{combinatorial dimension}, associated with the
decomposition technique, to be the maximum size of the defining set of
a cell. (This too is a much more widely applicable notion, within the
Clarkson-Shor sampling theory, but,
again, we only consider it here in the context of decomposing
arrangements of hyperplanes.)

When we analyze the behavior of random sampling under the Clarkson-Shor theory,
we replace the global growth function, as introduced earlier, by a \emph{local}
growth function $u$, so that $u(m)$ is the maximum number of cells (simplices
or vertical prisms) in the corresponding canonical decomposition of $\Arr(R)$,
for a sample $R$ of size $m$.

\paragraph{Exponential decay.}
The $\eps$-net theory, as well as the alternative approach of Clarkson
and Shor, are concerned with the ``worst-case'' behavior of a random
sample, in the sense that, for a sample $R$, the quality of the
decomposition associated with $R$ is measured by the \emph{largest}
size of a conflict list of a cell in the decomposition, whereas the average
size of such a list is in general smaller (typically, by a logarithmic
factor).  The analysis of Clarkson and Shor \cite{cs-arscg-89} shows
that the average size of a conflict list for a sample of size $O(br)$
(where $b$ is the combinatorial dimension),
raised to any fixed power $c$
 is $O((n/r)^c)$, suggesting that the
number of cells whose conflict list is significantly larger must be
very small. This has been substantiated by Chazelle and Friedman
\cite{cf-dvrsi-90}, who showed that the number of cells in a
canonical decomposition of the sample, having a conflict list whose size is
at least $t$ times larger than $n/r$, is exponentially decreasing as a function
of $t$. This is known as the \emph{exponential decay lemma}. Various
extensions and variants of the lemma have been considered in the literature;
see, e.g., \cite{AMS98,CMS93}.

\paragraph{Back to our context.}
Our main motivation for this study  was to (analyze more precisely, and) improve
Meiser's data structure for point location in high-dimensional arrangements of hyperplanes.
 A first (and crucial) step towards this goal
is to gain better understanding of the exact values of the above
parameters (VC-dimension, primal shatter dimension, and combinatorial
dimension), in particular of the way they depend on the dimension $d$.

\paragraph{Our results: The sampling parameters.}
We first re-examine the sampling parameters, namely, the VC-dimension,
growth function (and primal shatter dimension), and combinatorial
dimension, for the two canonical space decomposition techniques,
bottom-vertex triangulation and vertical decomposition. We discover
several unexpected phenomena, which show that (i) the parameters of the
two techniques are quite different, and (ii) within each technique, there
are large gaps between the VC-dimension, the primal shatter dimension, and the
combinatorial dimension. Some of these gaps reflect our present inability to
tighten the bounds on the corresponding parameters, but most of them are ``genuine'',
and we provide lower bounds that establish the gaps.

For vertical decomposition, the combinatorial dimension is only $2d$,
as easily follows from the definition. We prove that the global growth function is
at most $2^{O(d^3)}n^{d(d+1)}$, so the primal shatter dimension is at most $d(d+1)$,
and that the VC-dimension is at least $1+\frac12 d(d+1)$ and
at most $O(d^3)$ (the constants of proportionality are absolute and rather small).
Although we do not believe that the gap between the lower and upper bounds on the
VC-dimension is really so large, we were unable to obtain a tighter upper
bound on the VC-dimension, and do not know what is the true
asymptotic growth of this parameter (in terms of its dependence on $d$).
The local growth function is only $O(4^dn^{2d}/d^{7/2})$.

For bottom-vertex triangulation, both the primal shatter dimension and
the combinatorial dimension are $\Theta(d^2)$, but there is still a
significant gap between them: The combinatorial dimension is
$\frac12 d(d+3)$, whereas (i) the primal shatter dimension is at most
$d(d+1)$, and (ii) the VC-dimension is at least $d(d+1)$; we do not
know whether these two quantities are equal, and in fact do not even
know whether the VC-dimension is $O(d^2)$ (here standard arguments
imply that it is $O(d^2\log d)$). We also bound the local growth function by $n^d$.

The bound on the local growth function for 
vertical decomposition is new.
The bound for vertical decomposition is a drastic improvement from the earlier
bounds (such as in \cite{cegs-sessr-91}), in terms of its dependence on $d$,
and we regard it as one of the significant contributions of this work. It is
obtained using several fairly simple yet crucial observations concerning the
structure of the vertical decomposition in an arrangement of hyperplanes. We regard
this part of the paper as being of independent interest, and believe that it
will find applications in further studies of the structure and complexity of
vertical decompositions and their applications.

\paragraph{Our results: Point location.}
These findings have implications for the quality of random sampling
for decompositions of arrangements of hyperplanes. That is, in view of
our current state of knowledge, the best approach (with some caveats)
is to use vertical decomposition, and analyze it via the Clarkson--Shor
technique, in the sense that we can then ensure the desired sampling
quality while choosing a much smaller random sample.

This will be significant in our main application, to point location in
an arrangement of $n$ hyperplanes in $\Re^d$.
We show that the (fastest) query cost in Meiser's algorithm can be
improved, from $O(d^4 \log n)$ to $O(d^3\log n)$, if one uses vertical
decomposition instead of bottom-vertex triangulation (which is the one used in \cite{m-plah-93}),
and only wants to ensure that the problem size goes down by (at least) a factor of $2$ in
each recursive step (which is the best choice for obtaining fast query time).

In addition, we give a detailed and rigorous analysis of the tradeoff
between the query time and the storage size (and preprocessing cost),
both in the context of Meiser's algorithm (which, as noted, uses bottom-vertex
triangulation), and ours (which uses vertical decomposition). Meiser's
study does not provide such an analysis, and, as follows from our
analysis in \secref{sec:meiser}, the storage bound asserted in~\cite{m-plah-93}
(for the query time mentioned above) appears to be incorrect.
The storage bound has been improved, in a follow-up study of Liu~\cite{l-nplah-04},
to $O(n^d)$, but the constant of proportionality is $d^{O(d^3)}$.

We do not discuss here the second application of our analysis, namely to the decision complexity of $k$-SUM,  as it is presented in a separate
companion paper~\cite{ES17},
 and in the more recent work by Kane et al.~\cite{KLM17} (where the latter work is applicable only in certain restricted scenarios).

As far as we can tell, the existing literature lacks precise details of
the preprocessing stage, including a sharp analysis of the storage and
preprocessing costs. (For example, Meiser's work~\cite{m-plah-93} completely
skips this analysis, and consequently misses the correct trade-off
between query and storage costs.) One feature of our analysis is a careful
and detailed description of (one possible simple version of) this preprocessing stage.
We regard this too as an independent (and in our opinion, significant)
contribution of this work.

The improvement that we obtain by using vertical decomposition
(instead of bottom vertex triangulation) is a
consequence of the ability to use a sample of smaller size (and still
guarantee the desired quality).  It comes with a cost, though, since
(the currently best known upper bound on) the complexity of vertical
decomposition is larger than the bound on the complexity of
bottom-vertex triangulation. Specifically, focusing only on the dependence
on $n$, the complexity of bottom-vertex triangulation is $O(n^d)$, but the
best known upper bound on the complexity of vertical decomposition is $O(n^{2d-4})$~\cite{koltun}.

\paragraph{Degeneracies.}
Another feature of our work is that it handles arrangements that are not in general position.
This requires some care in certain parts of the analysis, as detailed throughout the paper.
One motivation for handling degeneracies is the companion study~\cite{ES17} of the $k$-SUM problem,
in which the hyperplanes that arise are very much not in general position.

\paragraph{Optimistic sampling.}
We also consider situations in which we do not insist that the entire decomposition of the arrangement
of a random sample be a $(1/r)$-cutting, but only require that the cell of the decomposition that contain
a prespecified point $q$ have conflict list of size at most $n/r$. This allows us to use a slightly
smaller sample size (smaller by a logarithmic factor), which leads to a small improvement in the cost
of a point-location query. It comes with a cost, of some increase in the storage of the structure.
We provide full analysis of this approach in \secref{sec:meiser}.

\paragraph{The recent work of Kane \etal}
Very recently, in a dramatic breakthrough, Kane et al.~\cite{KLM17} presented a new approach
to point location in arrangements of hyperplanes in higher dimensions, based on tools from active learning.
Although not stated explicitly in their work, they essentially propose a new cell decomposition technique
for such arrangements, which is significantly simpler than the two techniques studied here (bottom-vertex
triangulation and vertical decomposition), and which, when applied to a suitable random sample of
the input hyperplanes, has (with constant probability) the cutting property (having conflict lists of small size).
However, their approach only works for ``low-complexity'' hyperplanes, namely hyperplanes that have integer
coefficients whose tuple has small $L_1$-norm.\footnote{Kane et al.~use a more general notion of ``low inference dimension'', and show that low-complexity hyperplanes do have low inference dimension.} Moreover, the analysis in \cite{KLM17} only caters to the
\emph{linear decision tree model} that we have already mentioned earlier. Nevertheless, when the input
hyperplanes do have low complexity, the machinery in \cite{KLM17} can be adapted to yield a complete solution
in the RAM model, for point location in such arrangements.
We present this extension in \secref{sec:meiser}, using a suitable adaptation of our general point-location machinery.

\paragraph{Paper organization.}
We begin in \secref{sect2} with the analysis of random sampling, focusing mainly on the Clarkson-Shor
approach. We then study bottom-vertex triangulation in \secref{b:v:t}, including its representation and
complexity, and the various associated sampling parameters. An analogous treatment of vertical
decomposition is given in \secref{v:de}. Finally, in \secref{sec:meiser}, we present the point location
algorithms, for both decomposition methods considered here, including the improvement that cen be obtained
from optimistic sampling, and the variant for low-complexity hyperplanes.

\section{Random sampling}
\seclab{sect2}

In this section we review various basic techniques for analyzing the
quality of a random sample, focusing (albeit not exclusively) on the
context of decompositions of arrangements of hyperplanes.  We also add
a few partially novel ingredients to the cauldron.

As mentioned in the introduction, there are two main analysis techniques for random sampling,
one based on the VC-dimension or the primal shatter dimension of a suitably defined range space,
and one that focuses only on definable ranges, namely those that arise in a canonical decomposition
associated with some subset of the input. In this section we only briefly review the first approach,
and then spend most of the section on a careful analysis of the second one.
This latter technique, known as the Clarkson-Shor sampling technique (originally introduced in
Clarkson \cite{c-narsc-87} and further developed in Clarkson and Shor \cite{cs-arscg-89}), has been thoroughly
studied in the past (it is about 30 years old by now), but the analysis presented here still has some
novel features:
\begin{compactenum}[(A)]
\item It aims to calibrate in a precise manner the dependence of the various
  performance bounds on the dimension $d$ of the ambient space (in most of the earlier works,
  some aspects of this dependence were swept under the rug, and were hidden in the $O(\cdot)$ notation).

\item It offers an alternative proof technique that is based on
  \emph{double sampling}. This method has been used in the work
  of Vapnik and Chervonenkis~\cite{vc-ucrfe-71, vc-ucfoe-13} and
  the proof of Haussler and Welzl~\cite{hw-ensrq-87} of the
  $\eps$-net theorem. It is also (implicitly) used in the work of
  Chazelle and Friedman \cite{cf-dvrsi-90} in the context of the
  Clarkson-Shor framework.%

\item We also study the
  scenario where we do not care about the quality of the sample in terms of how small are the
  conflict lists of all the cells of the decomposition, but only of the cell that contains some
  fixed pre-specified point. This allows us to use smaller sample size, and thereby further
  improve the performance bounds. As in the introduction, we refer to this approach as
  \emph{optimistic sampling}.
\end{compactenum}

\subsection{Sampling, VC-dimension and $\eps$-nets}

We use the following sampling model, which is the one used in the
classical works on random sampling, e.g., in Haussler and Welzl~\cite{hw-ensrq-87}.

\begin{defn}
    \deflab{m:sample}%
    Let $H$ be a finite set of objects.  For a target size $\rho$, a
    \emph{$\rho$-sample} is a random sample $R \subseteq H$, obtained by $\rho$
    independent draws (with repetition) of elements from $H$.
\end{defn}

In this model, since we allow repetitions, the size of $R$ is not fixed (as it was in
the original analysis in \cite{c-narsc-87,cs-arscg-89}); it is a random variable whose
value is at most $\rho$.

The oldest approach for guaranteeing high-quality sampling is based on
the following celebrated result of Haussler and Welzl. In its original formulation,
the bound on the sample size is slightly larger; the improved dependence on the VC-dimension
comes from Blumer et al.~\cite{Blumer} (see \cite{Komlos:1992} for a matching lower bound, and see also \cite{pa-cg-95}).

\begin{theorem}[$\eps$-net theorem, Haussler and Welzl~\cite{hw-ensrq-87}]
    \thmlab{epsilon:net}%
    Let $(H,\RangeSet)$ be a range space of VC-dimension $\Dim$,
    and let $0 < \eps \leq 1$ and $0 < \BadProb < 1$ be given parameters. Let
    $N$ be an $m$-sample, where
    $$
        m \geq \max\pth{\; \frac{4}{\eps} \log \frac{4}{\BadProb} \,,\,
           \frac{8\Dim}{\eps}\log \frac{16}{\eps} \;}.
    $$
    Then $N$ is an $\eps$-net for $H$ with probability at least $1-\BadProb$.
\end{theorem}

\begin{remark}
    \remlab{epsilon:net:s}%
    A variant of the $\eps$-net theorem also holds for spaces with primal
    shattering dimension $\Dim_0$, except that the sample size
    has to be slightly larger.  Specifically, for the assertion of
    \thmref{epsilon:net} to hold, the required sample size is
    $\Theta\bigl( \frac{1}{\eps} \log \frac{1}{\BadProb} +
    \frac{\Dim_0}{\eps}\log \frac{\Dim_0}{\eps} \bigr)$. That is, $\Dim_0$
    also appears in the second logarithmic factor. See \cite{h-gaa-11} for details.
\end{remark}

\subsection{The Clarkson-Shor framework}
\seclab{e:d:l:setup}

We present the framework in a more general and abstract setting than what we really need,
in the hope that some of the partially novel features of the analysis will find
applications in other related problems.

Let $H$ be a set of $m$ objects, which are embedded in some space $E$.
We consider here a general setup where each subset $I \subseteq H$ defines a
decomposition of $E$ into a set $\CD{I}$ of canonical \emph{cells}, such that the following
properties hold.

\medskip

\begin{compactenum}[\;(a)]
        \item There exists an integer $b>0$ such that, for each $I \subseteq H$ and
        for each cell $\cell \in \CD{I}$, there exists a subset $J \subseteq I$ of
        size at most $\cDim$ such that $\cell \in \CD{J}$. A smallest such set $J$
        is called a \emph{defining set} of $\cell$, and is denoted as $\DefSet{\cell}$.
        Note that in degenerate situations $\DefSet{\cell}$ may not be uniquely defined,
        but it is uniquely defined if we assume general position. In this paper we allow
        degeneracies, and when $\DefSet{\cell}$ is not unique, any of the possible
        defining sets can be used.
        The minimum value of $\cDim$ satisfying this property is called the
        \emph{combinatorial dimension} of the decomposition.

        \smallskip%
        \item For any $I\subseteq H$ and any cell $\cell\in \CD{I}$, an object $f \in H$
        \emph{conflicts} with $\cell$, if $\cell$ is not a cell of
        $\CD{\DefSet{\cell} \cup \brc{f}}$; that is, the presence of $f$
        prevents the creation of $\cell$ in the decomposition corresponding to
        the augmented set. The set of objects in $H$ that conflict with $\cell$
        is called the \emph{conflict list} of $\cell$, and is denoted by $\KillSet{\cell}$.
        We always have $\DefSet{\cell}\cap \KillSet{\cell} = \emptyset$ for every defining set $\DefSet{\cell}$ of $\cell$.
    \end{compactenum}

\begin{figure}[b]
    \centerline{%
       \begin{tabular}{c|c|c|c}
         {\includegraphics[page=1,width=0.22\linewidth]%
         {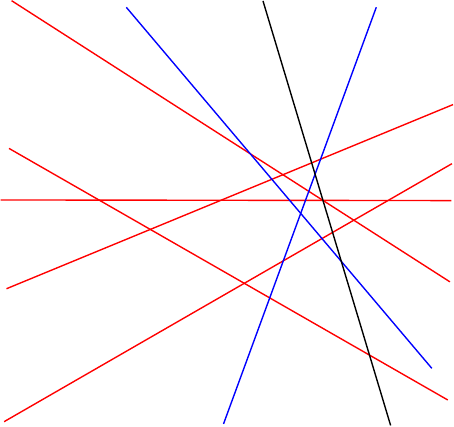}}%
         &
           {\includegraphics[page=2,width=0.22\linewidth]{figs/vertical}}%
         &
           {\includegraphics[page=3,width=0.22\linewidth]{figs/vertical}}%
         &
           {\includegraphics[page=4,width=0.22\linewidth]{figs/vertical}}%
         \\
         (A)
         &
           (B)
         &
           (C)
         &
           (D)
       \end{tabular}%
    }
    \caption{(A) Lines. (B) Vertical decomposition of the red
       lines. (C) Defining set for a vertical trapezoid $\sigma$, and its
       conflict list (the blue lines). (D) An alternative defining set
       for $\sigma$.}%
    \figlab{2dtrap}%
\end{figure}

\begin{example}
   Consider the case where $H$ is a set of $m$ lines in the plane,
    and the plane decomposition that a subset $I\subseteq H$ induces
    is the standard \emph{vertical decomposition} of $\Arr(I)$, as
    discussed in the introduction. Each cell in $\CD{H}$ is a
    \emph{vertical trapezoid}, defined by a subset of at most $b = 4$
    input lines. See \figref{2dtrap}.  The conflict list
    $\KillSet{\cell}$ of a vertical trapezoid $\cell$ is the set of
    lines in $H$ that intersect the interior of $\cell$. In degenerate
    situations, we may have additional lines that pass through
    vertices of $\cell$ without crossing its interior. Such lines are
    not included in $\KillSet{\cell}$, but may be part of an
    alternative defining set of $\cell$, see \figref{2dtrap} (D).
    %
\end{example}

We assume that the decomposition satisfies the following two conditions, sometimes
referred to as the \emph{axioms} of the framework.

\smallskip
\begin{compactenum}[\;(i)]
    \item 
    For any $R\subseteq H$ and for any $\cell \in \CD{R}$, we have $\DefSet{\cell} \subseteq R$,
    for some defining set $\DefSet{\cell}$ of $\cell$, and $\KillSet{\cell} \cap R = \emptyset$.

    \item 
    For any $R\subseteq H$, if $\DefSet{\cell} \subseteq R$, for some defining set
$\DefSet{\cell}$ of $\cell$, and
    $\KillSet{ \cell }\cap R=\emptyset$, then $\cell \in \CD{R}$.
\end{compactenum}

\smallskip%
When these conditions hold, the decomposition scheme \emph{complies} with the
analysis technique (or the framework) of Clarkson and Shor
\cite{cs-arscg-89} (see also \cite[Chapter 8]{h-gaa-11}).
These conditions do hold for the two decomposition schemes considered in this paper,
where the objects of $H$ are hyperplanes in $E=\Re^d$, and the decomposition is
either bottom-vertex triangulation or vertical decomposition.

\subsection{Double sampling, exponential decay, and cuttings}

\subsubsection{Double sampling}

One of the key ideas in our analysis is \emph{double sampling}.
This technique has already been used in the
original analysis of Vapnik and Chervonenkis from 1968 (see the
translated and republished version \cite{vc-ucfoe-13}) in the
context of VC-dimension.  In the context of the Clarkson-Shor
framework it was used implicitly by Chazelle and Friedman~\cite{cf-dvrsi-90}.
The technique is based on the intuition that two
independent random samples from the same universe (and of the same target size)
should look similar, in the specific sense that the sizes of their intersections
with any subset of the universe, of a suitable size, should not differ much from one another.

We need the following simple technical lemma.
\begin{lemma}
    \lemlab{double:sample}%
    Let $R_1$ and $R_2$ be two independent $\rho$-samples from the same universe
    $H$, and consider the merged sample $R = R_1 \cup R_2$. Let $B \subseteq H$ be
    a subset of $m$ elements. If $\rho\ge 2m$ then
    \begin{align*}
      \Prob{ \Bigl. B \subseteq R_1 \bigm\vert{ B \subseteq R}} \geq \frac{1}{3^m} .
    \end{align*}
\end{lemma}
\begin{proof}
    Let $A$ denote the event $B\subseteq R$.
    Each point $\pi$ in our probability space is  an ordered sequence of $2\rho$ draws:
    the $\rho$ draws of $R_1$ followed by the $\rho$ draws of $R_2$.
    The point $\pi$ is in $A$ if   every member of $B$ appears at least once in $\pi$.
    We define two sequences of $2\rho$ draws to be {\em equivalent} if one is a permutation of the other, and partition
    $A$ into the equivalence classes of this relation, so that each class consists of all the $(2\rho)!$ shuffles
    of (the indices of) some fixed sequence. In each class, each member of $B$ appears in $R$ some fixed number
    ($\ge 1$) of times. Fix one such class $A_1$. We estimate the conditional probability
    $\Prob{ \Bigl. B \subseteq R_1 \bigm\vert{ A_1}}$ from below, by deriving a lower bound for the number of shuffles
    in which every member of $B$ appears at least once in $R_1$. To do so, pick an arbitrary sequence $\sigma$
    in $A_1$, and choose $m$ (distinct) indices $i_1,i_2,\ldots,i_m$ such that $\sigma_{i_k}$ is the $k$-th
    element of $B$. (In general, this choice is not necessarily unique, in which case we arbitrarily pick
    one of these choices.) If the shuffle brings all these indices into $R_1$ (the first $\rho$ indices) then the event $B\subseteq R_1$ holds.
    The number of such shuffles is $\binom{\rho}{m} m! (2\rho-m)!$. That is,
\begin{align*}
    \Prob{ \Bigl. B \subseteq R_1 \bigm\vert{ A_1}} & \ge
    \frac{\binom{\rho}{m} m! (2\rho-m)! }{(2\rho)!} \\
    & = \frac{\rho!}{(\rho-m)!} \cdot \frac{(2\rho-m)!}{(2\rho)!} \\
    & = \frac{\rho}{2\rho}\cdot \frac{\rho-1}{2\rho-1}\cdots \frac{\rho-m+1}{2\rho-m+1} \\
    & \ge \left( \frac{\rho-m+1}{2\rho-m+1} \right)^m \ge \frac{1}{3^m} ,
\end{align*}
    where the last inequality follows from our assumption that $\rho\ge 2m$.
    Since this inequality holds for every class $A_1$ within $A$, it also holds for $A$; that is,
    $$
    \Prob{ \Bigl. B \subseteq R_1 \bigm\vert{ A}} \ge \frac{1}{3^m} ,
    $$
    as claimed.
\end{proof}

\begin{remark} \remlab{rem26}
The lower bound in \lemref{double:sample} holds also if instead of conditioning on
the event $A= B\subseteq R$ we condition on an event $A'\subseteq A$, where $A'$ is
closed under shuffles of the draws. That is, if some sequence $\pi$ of $2\rho$ draws
is in $A'$ then every permutation of $\pi$ should also be in $A'$.
We use this observation in the proof of \lemref{exponential:decay} below.
\end{remark}

\subsubsection{The exponential decay lemma}

For an input set $H$ of size $n$, and for parameters $\rho$ and $t \ge 1$,
we say that a cell $\cell$, in the decomposition of a $\rho$-sample from $H$, is
\emph{$\ell$-heavy} if the size of its conflict list is at least $\ell n/\rho$. (Note that $\ell$-heaviness
depends on $\rho$ too; we will apply this notation in contexts where $\rho$ has been fixed.)

\lemref{exponential:decay} below is a version of the exponential decay lemma, which gives
an upper bound on the expected number of $\ell$-heavy cells in the decomposition of a
$\rho$-sample from $H$. The original proof of the lemma was given by Chazelle and Friedman
\cite{cf-dvrsi-90}; our proof follows the one in Har-Peled \cite{h-sppss-16}.
Other instances and variants of the exponential decay lemma can be found, e.g., in \cite{AMS98,CMS93}.

Before proceeding, recall the definition of the global growth function, which
bounds the number of (definable or non-definable) ranges in any subset of $H$ of a given size.
In the rest of this section, we replace it by a ``local'' version, which provides, for each
integer $m$, an upper bound on the number of cells in the decomposition of any subset $I$ of
at most $m$ elements of $H$. We denote this function as $u(m)$; it is simply
$$
u(m) = \max_{I\subseteq H,\;|I|=m} |\CD{I}| .
$$

\begin{lemma}[Exponential decay lemma]%
 \lemlab{exponential:decay}
    Let $H$ be a set of $n$ objects with an associated decomposition scheme of
    combinatorial dimension $\cDim$ that complies with the Clarkson-Shor framework,
    and with local growth function $u(m) \le 2^\alpha m^\beta$, for some parameters $\alpha$ and
    $\beta\ge 1$. Let $R$ be a $\rho$-sample, for some $\rho \ge 4b$, let $\ell>1$ be arbitrary,
    and let $\CDH{R}{\ell}$ be the set of all $\ell$-heavy cells of $\CD{R}$. Then
    $$
        \Ex{\bigl. \cardin{ \CDH{R}{\ell} } } \le 3^b 2^{\alpha-\beta}\rho^\beta e^{-\ell/2} .
    $$
\end{lemma}
\begin{proof}
Consider $R$ as the union of two independent $(\rho/2)$-samples $R_1$ and $R_2$.
Let $\cell$ be a cell in $\CD{R}$.
By the  axioms of the Clarkson-Shor framework, $R$ contains at least one defining set $D = \DefSet{\cell}$,
and $K \cap R = \emptyset$ where $K = \KillSet{\cell}$.
Let $D_1,\ldots,D_y$ be all the possible defining sets of $\cell$, and let $A_i$, for $i=1,\ldots,y$, be the event that $D_i\subseteq R$ but
$D_j \not\subseteq R$ for $j <i$. Then we have
\begin{align*}
    \Prob{\bigl. \cell \in \CD{R_1} \mid  \cell \in \CD{R}}
   & =  \frac{\Prob{\bigl. \cell \in \CD{R_1} \wedge \cell \in \CD{R}} }{ {\Prob{\cell \in \CD{R}}}}  \\
 & =  \frac{\sum_{i=1}^y\Prob{\bigl. \cell \in \CD{R_1} \wedge A_i \wedge (K \cap R = \emptyset) } }{ {\Prob{\cell \in \CD{R}}}}  \\
& =  \frac{\sum_{i=1}^y\Prob{\bigl. \cell \in \CD{R_1} \mid A_i \wedge (K \cap R = \emptyset) } \Prob{A_i \wedge (K \cap R = \emptyset)} }{ {\Prob{\cell \in \CD{R}}}}  \\
&=%
\sum_{i=1}^y \Prob{\bigl. \cell \in \CD{R_1} \mid{ A_i \wedge (K \cap R = \emptyset) }} \Prob{A_i \mid \cell \in \CD{R} } \\
&\ge %
   \sum_{i=1}^y\Prob{\bigl. D_i \subseteq R_1 \mid{ A_i \wedge (K \cap R = \emptyset) }}\Prob{A_i \mid \cell \in \CD{R} } .
\end{align*}
We can now apply \lemref{double:sample} and the following \remref{rem26};
they apply since (i) each conditioning event
$(D_i\subseteq R) \wedge (\forall_{j<i} (D_j \not\subseteq R)) \wedge (K \cap R = \emptyset)$
is contained in the corresponding event $D_i\subseteq R$,
and is invariant under permutations of the draws, and (ii) $|R_1| = \rho/2 \ge 2b \ge 2|D|$. They imply that
$$
\Prob{\bigl. D_i \subseteq R_1 \mid{ A_i \wedge (K \cap R = \emptyset) }} \ge 1/3^b ,
$$
for $i=1,\ldots,y$. Furthermore, $\sum_{i=1}^y \Prob{A_i \mid \cell \in \CD{R} } = 1$,
as is easily checked, so we have shown that
$$
\Prob{\bigl. \cell \in \CD{R_1} \mid  \cell \in \CD{R}} \ge 1/3^b .
$$
Expanding this conditional probability and rearranging, we get
    \begin{align*}
      \Prob{ \big. \cell \in \CD{R}} \leq
      3^b \Prob{ \big. \cell \in \CD{R_1}\, \wedge\,
      \cell \in \CD{R}  } .%
      \eqlab{one}
    \end{align*}
Now, if $\cell$ is $\ell$-heavy with respect to $R$ (i.e., $\kappa:=\cardin{ \KillSet{\cell}} \geq \ell n/\rho$), then
\begin{equation}
    \Prob{ \bigl. \cell \in \CD{R} \mid{ \cell \in \CD{R_1}}}%
    = \Prob{ \KillSet{\cell} \cap R_2 =\emptyset}
    = (1-\kappa/n)^{\rho/2} \le e^{-\kappa\rho/(2n)} \le e^{-\ell/2} .  \eqlab{eq:expdecay}
\end{equation}
Let $\Family$ denote the set of all possible cells that arise in the decomposition
of some sample of $H$, and let $\Family_{\geq \ell}\subseteq \Family$ be the set of all
cells with $|K(\sigma)| \ge \ell n/\rho$. We have
    \begin{align*}
      \Ex{\bigl. \cardin{ \CDH{R}{\ell} } }%
      &=%
        \sum_{\cell \in \Family_{\geq \ell}} \Prob{ \bigl. \cell \in \CD{R}}
        \leq%
        \sum_{\cell \in \Family_{\geq \ell}} 3^b
        \Prob{ \bigl.  \cell \in
        \CD{R_1} \,\wedge\,  \cell \in \CD{R}}%
      \\ &%
          =%
          3^b \sum_{\cell \in \Family_{\geq \ell}} \underbrace{ \Prob{ \bigl.  \cell
          \in \CD{R} \mid{ \cell \in \CD{R_1}}}}_{\leq
          e^{-\ell/2}} \Prob{\bigl. \cell \in \CD{R_1}}
           \leq%
           3^b e^{-\ell/2} \sum_{\cell \in \Family_{\geq \ell}} \Prob{ \bigl. \cell \in
           \CD{R_1}}%
      \\ &%
        \leq%
        3^b e^{-\ell/2} \sum_{\cell \in \Family} \Prob{ \bigl. \cell \in
        \CD{R_1}}%
        = %
        3^b e^{-\ell/2} \Ex{\bigl. \cardin{\CD{R_1}}}%
        \leq%
        3^b e^{-\ell/2} u\pth{ \bigl. \cardin{{R_1}}} \\
        & \le %
        3^{b} e^{-\ell/2} 2^{\alpha} (\rho/2)^\beta
        \le 3^b 2^{\alpha-\beta} \rho^\beta e^{-\ell/2} ,
    \end{align*}
as asserted.
\end{proof}

\subsubsection{Cuttings from exponential decay}

The following central lemma provides a lower bound on the (expected) size of a sample $R$ that
guarantes, with controllable probability, that the associated space decomposition is a \emph{$(1/r)$-cutting}.
In the general setting in which we analyze the Clarkson-Shor framework, this simply means that each cell
of $\CD{R}$ has conflict list of size at most $n/r$.

\begin{lemma}%
    \lemlab{weak:cutting}%
    Let $H$ be a set of $n$ objects with an associated decomposition scheme of combinatorial dimension $b$
    that complies with the Clarkson--Shor framework, with local growth function
    $u\pth{m} \le 2^\alpha m^\beta$, for some parameters $\alpha$ and $\beta\ge 1$.
    Let $\BadProb \in (0,1)$ be some confidence parameter. Then, for
    \begin{align*}
      \rho \geq
      \max \pth{%
      4r\pth{ b\ln 3 + \pth{ \alpha - \beta }\ln 2 + \ln \frac{1}{\BadProb}} \, , \,
      8r\beta \ln(4 r \beta) } ,
    \end{align*}
    a $\rho$-sample induces a decomposition $\CD{R}$ that is a $(1/r)$-cutting with probability $\geq 1-\BadProb$.
\end{lemma}

\begin{proof}
    By \lemref{exponential:decay} (which we can apply since
    $\rho \ge 4r\pth{ b\ln 3 + \pth{ \alpha - \beta }\ln 2 + \ln \frac{1}{\BadProb}} \ge 4b$),
    the expected number of cells in $\CD{R}$ that are $\ell$-heavy is
    $\Ex{\bigl. \cardin{ \CDH{R}{\ell} } } \le 3^b 2^{\alpha-\beta}\rho^\beta e^{-\ell/2}$.
    We now require that $\ell n/ \rho \leq n /r$, and that the probability to obtain at least one heavy cell
    be smaller than $\BadProb$. We satisfy the first constraint by taking $\ell = \rho /r$, and we
    satisfy the second by making the expected number of heavy cells be smaller than $\BadProb$.
    Solving for this latter constraint, we get
    \begin{align*}
        3^b 2^{\alpha-\beta}\rho^\beta e^{-\rho /(2r)} \leq \BadProb & \iff%
        b\ln 3 + (\alpha-\beta) \ln 2 + \beta \ln \rho -\frac{\rho}{2r} \leq \ln \BadProb \\%
        & \iff \rho \geq 2r \pth{  b\ln 3 +(\alpha-\beta) \ln 2 + \beta \ln \rho +\ln \frac{1}{\BadProb}}.
    \end{align*}
    We assume that $\rho$ satisfies the conditions in the lemma. If the term $\beta \ln \rho$ is smaller than
    the sum of the other terms (in the parenthesis), then the inequality holds since
    \begin{math}
        \rho \geq 4r \pth{ b\ln 3 + (\alpha-\beta)\ln 2 +\ln \frac{1}{\BadProb}}.
    \end{math}
    Otherwise, the inequality holds since $\rho \geq 8r \beta \ln( 4r \beta)$, and for such values of $\rho$ we have
    \begin{math}
        \rho \geq 4r \beta \ln \rho,
    \end{math}
    as is easily checked.\footnote{%
      The general statement is that if $\rho \ge 2x\ln x$ then $\rho \ge x\ln \rho$. We will use this
      property again, in the proof of \lemref{v:c:dim}.}

    It follows that, with this choice of $\rho$, the probability of any heavy cell to exist
    in $\CD{R}$ is at most $\BadProb$ (by Markov's inequality), and the claim follows.
\end{proof}

\subsubsection{Optimistic sampling}%
\seclab{sec:optim}

In this final variant of random sampling, we focus only on the cell in the decomposition that contains
some prespecified point, and show that we can ensure that this cell is ``light'' (with high probability)
by using a smaller sample size. Note that here, as can be expected, the local growth function plays no role.

\begin{lemma}
    \lemlab{optimistic}%
    Let $H$ be a set of $n$ objects with an associated decomposition scheme of combinatorial
    dimension $\cDim$ that complies with the Clarkson-Shor framework,
    and let $R$ be a $\rho$-sample from $H$. Then, for any fixed point $q$ and any $\ell>1$, we have
    \begin{align*}
      \Prob{ \cardin{K(\cell(q, R))} \geq {\ell n}/{\rho}}
      \leq 3^\cDim e^{-\ell/2},
    \end{align*}
    where $\cell = \cell(q, R)$ is the cell in $\CD{R}$ that contains $q$.
\end{lemma}
\begin{proof}
We regard $R$ as the union of two independent random $(\rho/2)$-samples $R_1$, $R_2$,
as in the proof of \lemref{exponential:decay}, where it is shown that, for cells $\cell$ that
are $\ell$-heavy, we have
$$
\Prob{ \big. \cell \in \CD{R}} \leq
      3^b \Prob{ \big. \cell \in \CD{R_1}\, \wedge\, \cell \in \CD{R}  } \leq
      3^b e^{-\ell/2} \Prob{ \big. \cell \in \CD{R_1} } .
$$
Let $\Family_{\geq \ell}(q)$ be the set of all
possible canonical cells that are defined by subsets of $H$, contain
$q$, and are $\ell$-heavy (for the fixed choice of $\rho$). We have
\begin{align*}
  \Prob{ \cardin{\bigl. \KillSet{\cell(q, R)}} \geq {\ell n}/{\rho}}
  & = \sum_{\cell \in \Family_{\geq \ell}(q)} \Prob{ \bigl. \cell \in \CD{R}} \\
  & \leq 3^\cDim e^{-\ell/2}   \sum_{\cell \in \Family_{\geq \ell }(q)} \Prob{ \cell \in \CD{R_1}}
      \leq 3^\cDim e^{-\ell/2},
\end{align*}
since the events in the last summation are pairwise disjoint (only one cell in $ \Family_{\geq \ell}(q)$
can show up in $\CD{R_1}$).
\end{proof}

\begin{corollary}
    \corlab{optimistic}%
    Let $H$ be a set of $n$ objects with an associated decomposition scheme of combinatorial
    dimension $\cDim$ that complies with the Clarkson-Shor framework, and let
    $\BadProb \in (0,1)$ be a confidence parameter. Let $R$ be a $\rho$-sample from $H$, for
    $\rho \ge 2r( b\ln 3 + \ln \frac{1}{\BadProb})$. For any prespecified query point $q$, we have
    \begin{math}
        \Prob{ \cardin{\bigl.K(\cell(q, R))} > n/r } \leq \BadProb,
    \end{math}
    where $\cell = \cell(q, R)$ is the canonical cell in $\CD{R}$ that
    contains $q$.
\end{corollary}
\begin{proof}
    Put $\ell = \rho/r \ge 2(b\ln 3 + \frac{1}{\BadProb})$. By \lemref{optimistic}, we have
    $\Prob{ \cardin{\bigl.K(\cell(q, R))} > n/r} \le 3^be^{-\ell/2} \le \BadProb$, as claimed.
\end{proof}

\begin{remark}
  The significance of \corref{optimistic} is that it requires a sample size of $O(br)$,
  instead of $\Theta(br\log(br))$, if we are willing to consider only the single cell $\cell$
  that contains $q$, rather than \emph{all} cells in $\CD{R}$. We will use this result in
  \secref{sec:meiser} to obtain a slightly improved point-location algorithm.
\end{remark}

The results of this section are summarized in the table in \figref{var:dims}.

\begin{figure}
    \centerline{
       \begin{tabular}{|l|c|c||}
         \hline
         Technique & Sample size  & parameters \\
         \hline
         \hline
         $\Bigl.\eps$-net & $\ds O\pth{ \delta r \log r}$  & $\delta$: VC-dim.  \\
         \hline
         Shatter dim. & $\Bigl.\ds O\pth{ \delta_0 r  \log (\delta_0 r) }$ &  $\delta_0$: Shatter dim.  \\
         \hline
         \hline
         Combinatorial dim. &
           \begin{minipage}{0.25\linewidth}
               \begin{center}
                   $\Bigl. O\pth{ \bigl. r\pth{ b + \alpha + \beta \log( \beta r )} }$\\
                   \lemref{weak:cutting}
               \end{center}
           \end{minipage}
         &
           \begin{minipage}{0.25\linewidth}
           $b$: combinatorial dim. \\
               Number of cells in the\\
               decomposition is \\
               $(\alpha,\beta)$-growing.
           \end{minipage}
         \\
         \hline
         Optimistic sampling &
           \begin{minipage}{0.25\linewidth}
              \begin{center}
                 $O\bigl(rb \bigr)$\\
                 \corref{optimistic}%
              \end{center}
           \end{minipage}
         &
           \begin{minipage}{0.25\linewidth}
               \smallskip%
               Relevant cell has conflict list of size $\leq n/r$ w.h.p.
               \smallskip%
           \end{minipage}
         \\
         \hline
       \end{tabular}
    }%
    \caption{The different sizes of a sample that are needed to guarantee (with some fixed probability)
       that the associated decomposition is a $(1/r)$-cutting. The top two results and the third
       result are in different settings. In particular, the guarantee
       of the first two results is much stronger---any canonical (definable or non-definable)
       cell that avoids the sample is ``light'', while the third result only guarantees
       that the cells in the generated canonical decomposition of the sample are light.
       The fourth result caters only to the single cell that contains a specified point. \figlab{var:dims}}
\end{figure}

\section{Bottom-vertex triangulation}
\seclab{b:v:t}

We now apply the machinery developed in the preceding section to the
first of the two canonical decomposition schemes, namely, to bottom-vertex triangulations,
or \BVT, for short.

\subsection{Construction and number of simplices}
\seclab{b:v:t:construction}%

Let $H$ be a set of $n$ hyperplanes in $\Re^d$ with a suitably defined generic coordinate frame.
For simplicity, we describe the decomposition for the entire set $H$.
Let $C$ be a bounded cell of $\Arr(H)$ (unbounded cells are discussed later on). We decompose $C$ into
pairwise openly disjoint simplices in the following recursive manner.
Let $w$ be the bottom vertex of $C$, that is, the point of $C$ with the smallest
$x_d$-coordinate (we assume that there are no ties as we use a generic coordinate frame).
We take each facet $C'$ of $C$ that is not incident to $w$, and recursively construct its
bottom-vertex triangulation. To do so, we regard the hyperplane $h'$ containing $C'$ as a copy of $\Re^{d-1}$,
and recursively triangulate $\Arr(H')$ within $h'$, where
$H' = \{h\cap h' \mid h\in H\setminus \{h'\}\}$.  We complete the
construction by taking each $(d-1)$-simplex $\sigma'$ in the resulting
triangulation of $C'$ (which is part of the triangulation of $\Arr(H')$ within $h'$),
and by connecting it to $w$ (i.e., forming the
$d$-simplex ${\rm conv}(\sigma'\cup\{w\})$). Repeating this step over
all facets of $C$ not incident to $w$, and then repeating the whole
process within each (bounded) cell of $\Arr(H)$, we obtain the
bottom-vertex triangulation of $\Arr(H)$. The recursion terminates
when $d=1$, in which case $\Arr(H)$ is just a partition of the line
into intervals and points, which serves, as is, as the desired
bottom-vertex decomposition. We denote the overall resulting
triangulation as $\BT(H)$. See \cite{c-racpq-88,m-ldg-02} for more
details.  The resulting $\BVT$ is a simplicial complex~\cite{Spanier-66}.

To handle unbounded cells, we first add two hyperplanes $\pi_d^-$, $\pi_d^+$ orthogonal
to the $x_d$-axis, so that all  vertices of $\Arr(H)$ lie below $\pi_d^+$ and above $\pi_d^-$.
For each original hyperplane $h'$, we recursively triangulate $\Arr(H')$ within $h'$, where
$H' = \{h\cap h' \mid h\in H  \setminus \{h'\}\}$, and we also recursively triangulate
$\Arr(H^+)$ within $\pi_d^+$, where $H^+ = \{h\cap \pi_d^+ \mid h\in H \}$.
This triangulation is done in exactly the same manner, by adding artificial $(d-2)$-flats, and
all its simplices are either bounded or ``artificially-bounded''(i.e., contained in an artificial bounding flat).
We complete the construction by triangulating the $d$-dimensional cells in the arrangement of $H\cup \{\pi_d^-, \pi_d^+\}$.
The addition of $\pi_d^-$ guarantees that each such cell $C$ has a bottom vertex $w$.
(Technically, since we want to ensure that each cell has a unique bottom vertex, we slightly tilt
$\pi_d^-$, to make it non-horizontal, thereby ensuring this property, and causing no loss of generality.)
We triangulate $C$ by generating a $d$-simplex ${\rm conv}(\sigma'\cup\{w\})$) for each
$(d-1)$-simplex $\sigma'$ in the recursive triangulation of the facets on $\bd C$ which are
not incident to $w$ (none of these facets lies on $\pi_d^-$, but some might lie on $\pi_d^+$).
In the arrangement defined by the original hyperplanes of $H$ only, the unbounded simplices in
this triangulation correspond to simplices that contain features of the artificial hyperplanes
added at various stages of the recursive construction.

\paragraph{Number of simplices.}
We count the number of $d$-simplices in the \BVT{} in a recursive
manner, by charging each $d$-simplex $\cell$ to the unique
$(d-1)$-simplex on its boundary that is not adjacent to the bottom
vertex of $\cell$.

\begin{lemma}
    For any hyperplane $h\in H\cup \{\pi_d^+ \}$, each $(d-1)$-simplex $\cellA$ that is
    recursively constructed within $h$ gets charged at most once by the above scheme.
\end{lemma}
\begin{proof}
    Consider the cell $C'$ of $\Arr(H \cup \{\pi_d^+ \} \setminus \brc{h})$ that
    contains $\cellA$ in its interior, and let $p$ be the lowest
    vertex of $C'$. Assume first that $h$ does not pass through $p$.
    The hyperplane $h$ splits $C'$ into two cells $C_1, C_2$ of $\Arr(H)$.
    One of these cells, say $C_1$, has $p$ as its bottom vertex. The bottom vertex of the
    other cell $C_2$ must lie on $h$. In other words, $\cellA$ gets
    charged by a full-dimensional simplex within $C_1$ but not within $C_2$.
    If $h$ passes through $p$,  $\cellA$ cannot be charged at all.
\end{proof}

Note that the proof also holds in degenerate situations, where $p$ is incident to more than $d$ hyperplanes of $H$.

The $(d-1)$-simplices that are being charged lie on original hyperplanes or
on $\pi_d^+$. This implies the following lemma.


\begin{lemma}
    \lemlab{simplices:grow}%
(a) The number of $d$-simplices in $\BT(R)$, for any set $R$ of $\rho$ hyperplanes in $d$ dimensions,
is at most $\rho^d$.

\smallskip
\noindent
(b) The total number of simplices of all dimensions in $\BT(R)$ is at most $e\rho^d$.
\end{lemma}
\begin{proof}
    (a) The analysis given above implies that the number of
    $d$-simplices in $\BT(R)$ is at most the overall number of
    $(d-1)$-simplices in the bottom-vertex triangulations within the
    hyperplanes of $R$, plus the number of $(d-1)$-simplices in the
    bottom-vertex triangulations within $\pi_d^+$.  Hence, if we
    denote by $BV_d(\rho)$ the maximum number of $d$-simplices in a
    bottom-vertex triangulation of a set of $\rho$ hyperplanes in
    $\Re^d$, we get the recurrence
    \begin{math}
        B_d(\rho) \le (\rho+1) B_{d-1}(\rho-1),
    \end{math}
    and $B_1(\rho) = \rho+1$. The solution of this recurrence is
    \begin{align*}
      B_d(\rho) &\leq (\rho+1)\rho(\rho-1)\cdots (\rho-d+3)B_1(\rho-d+1) \\
      & = (\rho+1)\rho(\rho-1)\cdots (\rho-d+3)(\rho-d+2) \\
      & = \rho^d \cdot \left( 1+\frac{1}{\rho}\right) \cdot 1 \left( 1-\frac{1}{\rho}\right) \cdots
      \left(1-\frac{d-2}{\rho} \right) \leq \rho^d .
    \end{align*}
\smallskip
\noindent
(b) The number of $j$-flats of $\Arr(R)$ is $\binom{\rho}{d-j}$, and each $j$-simplex
belongs to the triangulation within some $j$-flat of $\Arr(R)$. It follows from (a)
that the number of $j$-simplices in $\BT(R)$ is at most $\binom{\rho}{d-j} \rho^j$.
Summing over all $j$ we get that the total number of simplices in $\BT(R)$ is bounded by
\begin{align*}
  \sum_{j=0}^d \binom{\rho}{d-j} \rho ^j \le
  \sum_{j=0}^d \frac{\rho^d}{(d-j)!} \le e \rho^d \ .
\end{align*}
\end{proof}

In particular, the local growth function for bottom vertex triangulation satisfies
$$
u(m) \le e m^d = 2^{\log e} m^d \ ;
$$
that is, it is $(\alpha,\beta)$-growing, where the corresponding parameters are $\alpha = \log e$ and $\beta = d$.

\begin{remark} \remlab{simp:comb:dim}
    The {\bf combinatorial dimension} of bottom-vertex triangulation is
    $\frac{1}{2} d(d+3)$ \cite{m-ldg-02}.  Indeed, if $C(d)$ denotes
    the combinatorial dimension of bottom-vertex triangulation  in $d$ dimensions, then
    $C(d) = d + 1 + C(d-1)$, where the non-recursive term $d+1$ comes
    from the $d$ hyperplanes needed to specify the bottom vertex of the
    simplex, plus the hyperplane $h'$ on which the recursive
    $(d-1)$-dimensional construction takes place. This recurrence,
    combined with $C(1)=2$, yields the asserted value. Simplices in
    unbounded cells require fewer hyperplanes from $H$ to define them.
\end{remark}

Plugging the parameters $b = \frac12 d(d+3)$, $\alpha = O(1)$, $\beta = d$ into \lemref{weak:cutting}, we obtain:

\begin{corollary}
Let $H$ be a set of $n$ hyperplanes in $\Re^d$
and let
  $\BadProb \in (0,1)$ be some confidence parameter. Then, for
$\rho \geq cr(d^2+d\log r+\ln \frac{1}{\BadProb})$, for a suitable absolute constant $c$,
the bottom-vertex triangulation of
 a $\rho$-sample  is a $(1/r)$-cutting with probability $\geq 1-\BadProb$.
\end{corollary}

\subsection{The primal shatter dimension and VC-dimension}
\seclab{vc-bvt}

Let $H$ be a set of $n$ hyperplanes in $\Re^d$.
The range space associated with the bottom-vertex triangulation
technique is $(H,\Sigma)$, where each range in $\Sigma$ is associated
with some (arbitrary, open) $d$-simplex $\sigma$, and is the subset of $H$
consisting of the hyperplanes that cross $\sigma$.  Clearly, the same
range will arise for infinitely many simplices, as long as they have
the same set of hyperplanes of $H$ that cross them.

\begin{lemma}[See \cite{as-aa-00}]%
    \lemlab{num:cells}%
    (a) The number of full dimensional cells in an arrangement of $n$
    hyperplanes in $\Re^d$ is at most
    \begin{math}
        \sum_{i=0}^d \binom{n}{i} \leq 2 \pth{\frac{ne}{d}}^d ,
    \end{math}
where the inequality holds for $n\ge d$.

\smallskip
\noindent
(b) The number of cells of all dimensions is at most
$\sum_{i=0}^d \binom{n}{i}2^i \leq 2 \pth{\frac{2ne}{d}}^d$;
again, the inequality holds for $n\ge d$.
\end{lemma}

\begin{proof}
See \cite{as-aa-00} for (a). For (b), as in the proof of \lemref{simplices:grow}(b), we apply the
bound in (a) within each flat formed by the intersection of some subset of the hyperplanes.
Since there are at most $\binom{n}{d-j}$ $j$-dimensional such flats, and at most $n-d+j$
hyperplanes form the corresponding arrangement within a $j$-flat, we get a total of at most
\begin{align*}
\sum_{j=0}^d \binom{n}{d-j} \sum_{i=0}^j \binom{n-d+j}{i}  & = \sum_{j=0}^d \sum_{i=0}^j \frac{n!}{(d-j)! i! (n-d+j-i)!} \\
& = \sum_{j=0}^d \sum_{k=0}^j \frac{n!}{(d-j)! (j-k)! (n-d+k)!} \\
& = \sum_{j=0}^d \sum_{k=0}^j \frac{n! (d-k)!}{(d-j)! (j-k)! (n-d+k)!(d-k)!} \\
&= \sum_{j=0}^d \sum_{k=0}^j \binom{n}{d-k} \binom{d-k}{d-j} \\
&= \sum_{k=0}^d  \binom{n}{d-k} \sum_{j=k}^d \binom{d-k}{d-j} \\
&= \sum_{k=0}^d  \binom{n}{d-k} 2^{d-k} = \sum_{i=0}^d \binom{n}{i}2^i
\end{align*}
cells of all dimensions. The asserted upper bound on this sum is an immediate consequence of (a).
\end{proof}

\begin{lemma}
    \lemlab{h:s:growth}%
    The global growth function $g_d(n)$ of the range space of hyperplanes and simplices in $\Re^d$ satisfies
    \begin{math}
        g_d(n) \leq 2^{d+2} \pth{\frac{e}{d}}^{(d+1)^2} n^{d(d+1)}.
    \end{math}
    Consequently, the primal shatter dimension of this range space is
    at most $d(d+1)$.
\end{lemma}

We can write this bound as $g_d(n) \le 2^\alpha n^\beta$, for $\beta =d(d+1)$ and $\alpha = d+2 - (d+1)^2 \log(d/e)$ ($\alpha$ is negative when $d$ is large).

\begin{proof}
    Let $\sigma$ be a $d$-simplex, and let $C_1,\ldots,C_{d+1}$ denote
    the cells of $\Arr(H)$ that contain the $d+1$ vertices of
    $\sigma$; in general, some of these cells may be
    lower-dimensional, and they need not be distinct. Moreover, the
    order of the cells is immaterial for our analysis.  As is easily
    seen, the range associated with $\sigma$ does not change when we
    vary each vertex of $\sigma$ within its containing cell. Moreover,
    since crossing a simplex means intersecting its interior, we may
    assume, without loss of generality, that all the cells $C_i$ are
    full-dimensional.  It follows that the number of possible ranges
    is bounded by the number of (unordered) 
    tuples of at most $d+1$ full-dimensional cells of
    $\Arr(H)$.  The number of such cells in $\Arr(H)$ is
    \begin{math}
        T \leq 2 \pth{\frac{ne}{d}}^d,
    \end{math}
    by \lemref{num:cells}(a).  Hence, using the inequality again, in the
    present, different context, the number of distinct ranges is
    \begin{align*}
      \sum_{i=0}^{d+1} \binom{T}{i}%
      & \leq%
      2\pth{\frac{ T e }{d+1}}^{d+1}
      \leq%
      2\pth{\frac{ 2 \pth{\frac{ne}{d}}^d e }{d+1}}^{d+1} \\
      & =%
      2 \pth{\frac{2e}{d+1}}^{d+1}\pth{\frac{e}{d}}^{d(d+1)} n^{d(d+1)}
      \leq%
      2^{d+2}
      \pth{\frac{e}{d}}^{(d+1)^2}
      n^{d(d+1)}.
    \end{align*}
\end{proof}

\subsubsection{Bounds on the VC-dimension}

\begin{lemma}
    \lemlab{v:c:dim}%
    The VC-dimension of the range space of hyperplanes and simplices in $\Re^d$
    is at least $d(d+1)$ and at most $5 d^2 \log d$, for $d \geq 9$.
\end{lemma}

\begin{proof}
    Using the bound of \lemref{h:s:growth} on the growth function, we
    have that if a set of size $k$ is shattered by simplicial ranges then
    \begin{math}
        2^k \leq g_d(k)%
        \leq%
        2^{d+2} \pth{\frac{e}{d}}^{(d+1)^2} k^{d(d+1)} \leq
        k^{d(d+1)},
    \end{math}
    for $d \geq 2e$.  This in turn is equivalent to
    \begin{math}
        \frac{k}{\ln k} \leq \frac{d(d+1)}{\ln 2}.
    \end{math}
    Using the easily verified property, already used in the proof of \lemref{weak:cutting}, that
    $x/ \ln x \leq u \implies x \leq 2 u \ln u$, we conclude that
    \begin{math}
        k \leq 2 \frac{d (d+1)}{\ln 2} \ln \pth{ \frac{d (d+1)}{\ln 2}
        } \leq %
        5 d^2 \log_2 d,
    \end{math}
    for $d \geq 9$, as can easily (albeit tediously) be verified.

    \begin{figure}[t]
        \centerline{%
           \begin{tabular}{ccc}
             \includegraphics[page=1]{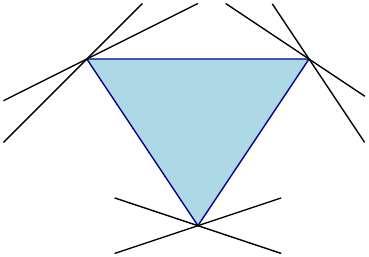}
             &
               \qquad\qquad
             &%
               \includegraphics[page=2]{figs/shatter}\\
             (A)
             &
             &
               (B)
           \end{tabular}%
        }
        \caption{Illustration of the proof of \lemref{v:c:dim}.
          (A) The construction. (B) Turning a subset into a simplex range.}
        \figlab{shatter:h:s}
    \end{figure}
    For the lower bound, let $\sigma_0$ denote some fixed regular
    simplex, and denote its vertices as $v_1,v_2,\ldots,v_{d+1}$.  For
    each $v_i$, choose a set $H_i$ of $d$ hyperplanes that are
    incident to $v_i$, so that (i) none of these hyperplanes crosses
    $\sigma_0$, (ii) they are all nearly parallel to some `ground
    hyperplane' $h^{(i)}$ that supports $\sigma_0$ at $v_i$, and the
    angles that $h^{(i)}$ forms with the vectors $\vec{v_iv_j}$, for
    $j\ne i$, are all at least some positive constant $\alpha$ (that
    depends on $d$ but is bounded away from $0$), and (iii) the
    hyperplanes of $H_i$ are otherwise in generic position (except
    for being concurrent at $v_i$). The set $H:=\bigcup_j H_j$ consists
    of $d(d+1)$ hyperplanes, and we claim that $H$ can be shattered by simplices.

    Indeed, let $c_0$ denote the center of mass of $\sigma_0$. In a
    suitable small neighborhood of $v_i$, the hyperplanes of $H_i$
    partition space into $2^d$ ``orthants'' (all of which are tiny,
    flattened slivers, except for the one containing $\sigma_0$ and
    its antipodal orthant), and for each subset $H'_i$ of $H_i$ there
    is a unique orthant $W(H'_i)$, such that for any point
    $q\in W(H'_i)$, the segment $c_0q$ crosses all the hyperplanes of
    $H'_i$, and no hyperplane of $H_i\setminus H'_i$, see \figref{shatter:h:s}.
    It is easily seen that the same property also holds in the following
    more general context.  Perturb $\sigma_0$ by moving each of its
    vertices $v_i$ to a point $q_i$, sufficiently close to $v_i$, in
    any one of the orthants incident to $v_i$. Then, assuming that the
    hyperplanes of each $H_i$ are sufficiently `flattened' near the
    ground hyperplane $h^{(i)}$, the perturbed simplex
    $\sigma = {\rm conv}(\{q_1,\ldots,q_{d+1}\})$ is such that, for
    each $i$, a hyperplane of $H_i$ crosses $\sigma$ if and only if it
    crosses the segment $c_0q_i$.

    All these considerations easily imply the desired property. That
    is, let $H'$ be any subset of $H$, and put $H'_i := H'\cap
    H_i$. For each $i$, choose a point $q_i$ near $v_i$ in the
    incident orthant that corresponds to $H'_i$. Then $H'$ is
    precisely the subset of the hyperplanes of $H$ that are crossed by
    the simplex whose vertices are $q_1,\ldots,q_{d+1}$, showing that
    $H$ is indeed shattered by simplices.
\end{proof}

In particular, we have established a fairly large gap between the
combinatorial dimension $\frac12d(d+3)$ and the VC-dimension (which is
at least $d(d+1)$).  The gap exists for any $d\ge 2$, and increases as
$d$ grows---for large $d$ the combinatorial dimension is about  half
the VC-dimension.


%



\section{Vertical decomposition}
\seclab{v:de}

In this section we study the combinatorial dimension, primal shatter
dimension, and VC-dimension of vertical decomposition. In doing so, we
gain new insights into the structure of vertical decomposition, and
refine the combinatorial bounds on its complexity; these insights are
needed for our analysis, but they are of independent interest, and
seem promising for further applications.

\subsection{The construction}
\seclab{v:d}

Let $H$ be a set of $n$ hyperplanes in $\Re^d$.  The construction of
the vertical decomposition (\VD{} for short) of $\Arr(H)$ works by
recursion on the dimension, and is somewhat more involved than the
bottom-vertex triangulation. Its main merit is that it also applies to
decomposing arrangements of more general (constant-degree algebraic)
surfaces, and is in fact the only known general-purpose technique for
this generalized task (although we consider it here only in the
context of decomposing arrangements of hyperplanes).  Moreover, as the
analysis in this paper shows, this technique is the ``winner'' in
obtaining faster point-location query time (with a few caveats, noted
later).

Let $H$ be a collection of $n$ hyperplanes in $\Re^d$ with a suitably
defined generic coordinate frame. We can therefore assume, in
particular, that there are no vertical hyperplanes in $H$, and, more
generally, that no flat of intersection of any subset of hyperplanes
is parallel to any coordinate axis.  The \emph{vertical decomposition}
$\V(H)$ of the arrangement $\Arr(H)$ is defined in the following
recursive manner (see~\cite{cegs-sessr-91,sa-dsstg-95} for the general
setup, and \cite{GHMS-95,koltun} for the case of hyperplanes in four
dimensions, as well as the companion paper~\cite{ES17}).  Let the coordinate system be $x_1, x_2, \ldots, x_d$,
and let $C$ be a $d$-dimensional cell in $\Arr(H)$.  For each
$(d-2)$-face $g$ on $\bd{C}$, we erect a $(d-1)$-dimensional
\emph{vertical wall} passing through $g$ and confined to $C$; this is
the union of all the maximal $x_d$-vertical line-segments that have
one endpoint on $g$ and are contained in $C$.  The walls extend
downwards (resp., upwards) from faces $g$ on the top boundary (resp.,
bottom boundary) of $C$ (faces on the ``equator'' of $C$, i.e., faces
that have a vertical supporting hyperplane, have no wall erected from
them within $C$).
This collection of walls subdivides $C$ into convex vertical prisms,
each of which is bounded by (potentially many) vertical walls, and by
two hyperplanes of $H$, one appearing on the bottom portion and one on
the top portion of $\bd{C}$, referred to as the \emph{floor} and the
\emph{ceiling} of the prism, respectively; in case $C$ is unbounded, a
prism may be bounded by just a single (floor or ceiling) hyperplane of
$H$.
More formally (or, rather, in an alternative, equivalent formulation),
this step is accomplished by projecting the bottom and the top
portions of $\bd{C}$ onto the hyperplane $x_d=0$, and by constructing
the \emph{overlay} of these two convex subdivisions (of the projection
of $C$).  Each full-dimensional (i.e., $(d-1)$-dimensional) cell in
the overlay, when lifted back to $\Re^d$ and intersected with $C$,
becomes one of the above prisms.

Note that after this step, the two bases (or the single base, in case
the prism is unbounded) of a prism may have arbitrarily large
complexity, or, more precisely, be bounded by arbitrarily many
hyperplanes (of dimension $d-2$).  The common projection of the two
bases is a convex polyhedron in $\Re^{d-1}$, bounded by at most $2n-1$
hyperplanes,\footnote{We will shortly argue that the actual number of
   hyperplanes is only at most $n-1$.} where each such hyperplane is
the vertical projection of either an intersection of the floor
hyperplane $h^-$ with another original hyperplane $h$, or an
intersection of the ceiling hyperplane $h^+$ with some other $h$; this
collection might also include $h^-\cap h^+$.

In what follows we refer to these prisms as \emph{first-stage} prisms.
Our goal is to decimate the dependence of the complexity of the prisms
on $n$, and to construct a decomposition so that each of its prisms is
bounded by no more than $2d$ hyperplanes. To do so, we recurse with
the construction at the common projection of the bases onto $x_d=0$.
Each recursive subproblem is now $(d-1)$-dimensional.

Specifically, after the first decomposition step described above, we
project each prism just obtained onto the hyperplane $x_d = 0$,
obtaining a $(d-1)$-dimensional convex polyhedron $C'$, which we
vertically decompose using the same procedure described above, only in
one lower dimension.  That is, we now erect vertical walls within $C'$
from each $(d-3)$-face of $\bd{C'}$, in the $x_{d-1}$-direction.
These walls subdivide $C'$ into $x_{d-1}$-vertical prisms, each of
which is bounded by (at most) two facets of $C'$, which form its floor
and ceiling (in the $x_{d-1}$-direction), and by some of the vertical
walls. We keep projecting these prisms onto coordinate hyperplanes of
lower dimensions, and produce the appropriate vertical walls.  We stop
the recursion as soon as we reach a one-dimensional instance, in which
case all prisms projected from previous steps become line-segments,
requiring no further decomposition.
We now backtrack, and lift the vertical walls (constructed in lower
dimensions, over all iterations), one dimension at a time, ending up
with $(d-1)$-dimensional walls within the original cell $C$; that is,
a $(d-i)$-dimensional wall is ``stretched'' in directions
$x_{d-i+2}, \ldots, x_{d}$ (applied in that order), for every
$i=2, \ldots, d-1$. See \figref{vd-in-3d} for an illustration of the construction.

\begin{figure}[b]
  \centerline{%
    \begin{tabular}{c|c|c}
      {\includegraphics[page=1,width=0.22\linewidth]%
        {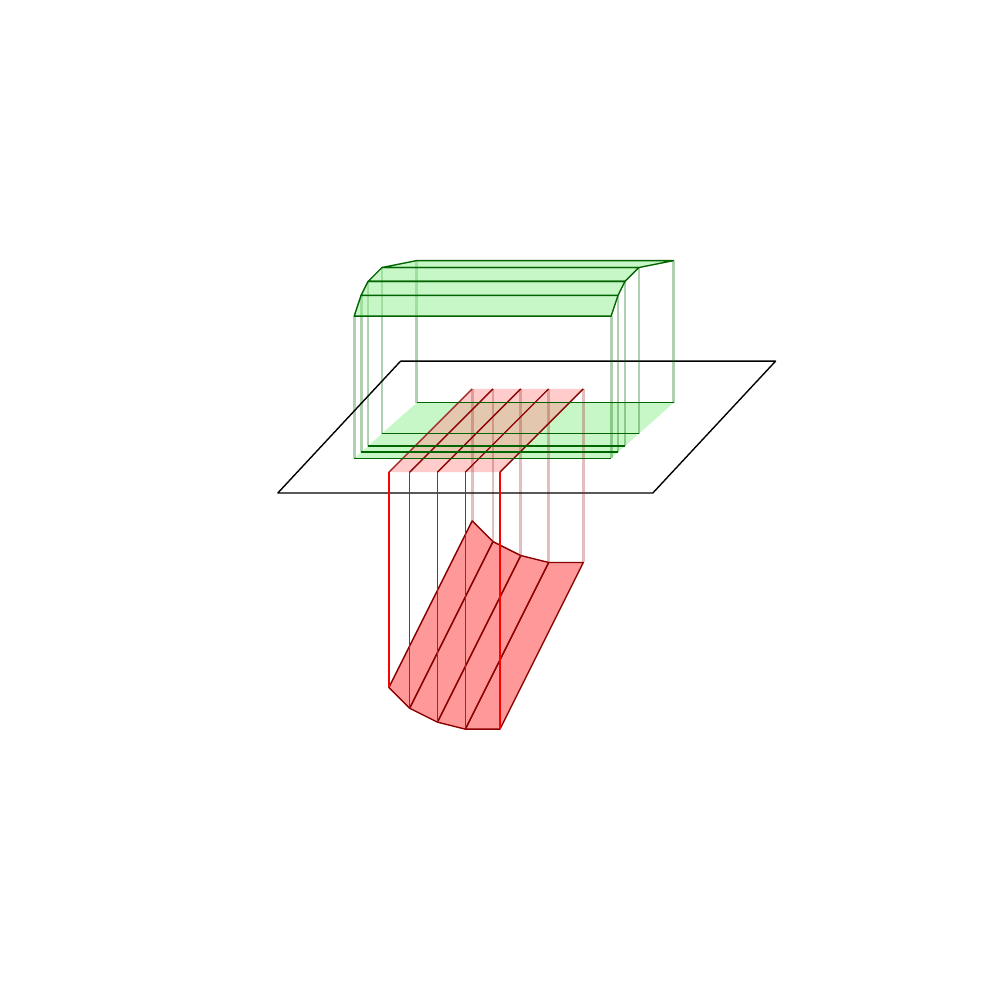}}%
      &
      {\includegraphics[page=2,width=0.22\linewidth]{figs/3d_prism}}%
      &
      {\includegraphics[page=4,width=0.22\linewidth]{figs/3d_prism}}%
      \\
      (A)
      &
      (B)
      &
      (C)
       \end{tabular}%
  }
  \caption{}%
  \figlab{vd-in-3d}%
\end{figure}

%
%

Each of the final cells is a ``box-like'' prism, bounded by at most
$2d$ hyperplanes. Of these, two are original hyperplanes (its floor
and ceiling in $\Re^d$), two are hyperplanes supporting two
$x_d$-vertical walls erected from some $(d-2)$-faces, two are
hyperplanes supporting two $x_{d-1}x_d$-vertical walls erected from
some $(d-3)$-faces (within the appropriate lower-dimensional
subspaces), and so on.

We apply this recursive decomposition for each $d$-dimensional cell
$C$ of $\Arr(H)$, and we apply an analogous decomposition also to each
lower dimensional cell $C'$ of $\Arr(H)$, where the appropriate
$k$-flat that supports $C'$ is treated as $\reals^k$, with a 
suitable generic choice of coordinates. The union of the resulting
decompositions is the vertical decomposition $\VD(H)$.

\begin{lemma} \lemlab{def:hyp} Consider the prisms in the vertical
    decomposition of a $d$-dimensional cell $C$ of $\Arr(H)$. Then (i)
    Each final prism $\cell$ is defined by at most $2d$ hyperplanes of
    $H$. That is, there exists a subset $\DefSet{\cell} \subseteq H$
    of size at most $2d$, such that $\cell$ is a prism of
    $\VD(\DefSet{\cell})$.

    \smallskip
    \noindent
    (ii) $\DefSet{\cell}$ can be enumerated as
    $(h_1^-,h_1^+,h_2^-,h_2^+,\ldots,h_d^-,h_d^+)$ (where some of
    these hyperplanes may be absent), such that, for each
    $j=d,d-1,\ldots,1$, the floor (resp., ceiling) of the projection
    of $\cell$ processed at dimension $j$ is defined by a sequence of
    intersections and projections that involve only the hyperplanes
    $(h_1^-,h_1^+,h_2^-,h_2^+,\ldots,h_{d-j}^-,h_{d-j}^+)$ and
    $h_{d-j+1}^-$ (resp., $h_{d-j+1}^+$).
\end{lemma}


\begin{proof}
    We first establish property (ii) using backward induction on the
    dimension of the recursive instance.  For each dimension
    $j=d,d-1,\ldots,1$, we prove that when we are at dimension $j$, we
    already have a set $D_j$ of (at most) $2(d-j)$ original defining
    hyperplanes (namely, original hyperplanes defining the walls
    erected so far, in the manner asserted in part (ii)), and that
    each (lower-dimensional) hyperplane in the current collection
    $H_j$ of $(j-1)$-hyperplanes is obtained by an interleaved
    sequence of intersections and projections, which are expressed in
    terms of some subset of the defining hyperplanes and (at most)
    \emph{one additional original hyperplane}. This holds trivially
    initially, for $j=d$. For $j=d-1$ we have two defining hyperplanes
    $h_1^-$, $h_1^+$ in $D_{d-1}$, which contain the floor and ceiling
    of the prism, respectively.  The collection $H_{d-1}$ is obtained
    by intersecting $h_1^-$ and $h_1^+$ with the remaining hyperplanes
    of $H$ (including the intersection $h_1^-\cap h_1^+$), and by
    projecting all these intersections onto the $(d-1)$-hyperplane
    $x_d=0$. Then a new pair of hyperplanes $h_2^-$ and $h_2^+$ (with shortest
    vertical distance) are chosen, and thus the floor (resp., ceiling) of the projection
    of $\cell$ is defined is defined by a sequence of intersections and projections that
    involve $h_1^-, h_1^+, h_2^-, h_2^+$. 
    So the inductive property holds for $j=d-1$. In general,
    when we move from dimension $j$ to dimension $j-1$ we choose a new
    floor $g_{j+1}^-$ and a new ceiling $g_{j+1}^+$ from among the hyperplanes in $H_j$,
    gaining two new (original) defining hyperplanes $h_{j+1}^-$ and $h_{j+1}^+$.
    We add these new defining hyperplanes 
    to $D_j$ to form $D_{j-1}$,
    and intersect each of the floor $g_{j+1}^-$ and ceiling $g_{j+1}^+$ with the other
    hyperplanes in $H_j$, and project all the resulting
    $(j-2)$-intersections onto the $(j-1)$-hyperplane $x_j=0$, to
    obtain a new collection $H_{j-1}$ of $(j-2)$-hyperplanes. Clearly,
    the inductive properties that we assume carry over to the new sets
    $D_{j-1}$ and $H_{j-1}$, so this holds for the final setup in
    $d=1$ dimensions. Since each step adds at most two new defining
    hyperplanes, one for defining the floor and one for defining the
    ceiling, the claim in (ii) follows.  Property (i) follows too
    because the above construction will produce $\cell$ when the input
    consists of just the hyperplanes of $\DefSet{\cell}$.
\end{proof}

An analogous lemma holds for the vertical decomposition of lower
dimensional cells of $\Arr(H)$, except that if the cell $C'$ lies in a
$(d-k)$-flat, we have to replace the original hyperplanes by the
intersection of the $k$ hyperplanes defining the flat with all other
hyperplanes, and start the recursion in dimension $d-k$. The following
corollary is an immediate consequence.

\begin{corollary}
    \corlab{prisms:comb:dim}%
    The combinatorial dimension of the vertical prisms in the vertical
    decomposition of hyperplanes in $\Re^d$ is $b=2d$.
\end{corollary}

\begin{remark} \remlab{cd:of:vd} Although this is marginal in our
    analysis, it is instructive to note that, even though it only
    takes at most $2d$ hyperplanes to define a prism, expressing how
    the prism is constructed in terms of these hyperplanes is
    considerably more space consuming, as it has to reflect the
    sequence of intersections and projections that create each
    hyperplane that bounds the prism (each of the $2d$ bounding
    hyperplanes carries a ``history'' of the way it was formed, of
    size $O(d)$). A naive representation of this kind would require
    $O(d^2)$ storage per prism. We will bypass this issue when
    handling point-location queries in \secref{sec:meiser}.
\end{remark}

\subsection{The complexity of vertical decomposition}
\seclab{vd-comp}

Our first step is to obtain an upper bound on the complexity of the
vertical decomposition, that is, on the number of its prisms.  This
will also determine the local growth function in the associated
Clarkson--Shor framework. Following the presentation in the
introduction, we analyze the complexity within a single cell of
$\Arr(H)$, and then derive a global bound for the entire
arrangement. As it turns out, this is crucial to obtain a good
dependence on $d$.  In contrast, the traditional global analysis, as
applied in \cite{cegs-sessr-91,sa-dsstg-95}, yields a significantly
larger coefficient, which in fact is too large for the purpose of the
analysis of the point-location mechanism introduced in this paper.

\subsubsection{The analysis of a single cell}

Let $C$ be a fixed $d$-dimensional cell of $\Arr(H)$. With a slight
abuse of notation, denote by $n$ the number of its facets (that is,
the number of hyperplanes of $H$ that support these facets), and
consider the procedure of constructing its vertical decomposition, as
described in~\secref{v:d}.  As we recall, the first stage produces
vertical prisms, each having a fixed floor and a fixed ceiling. We
take each such prism, whose ceiling and floor are contained in two
respective hyperplanes $h^+$, $h^-$ of $H$, project it onto the
hyperplane $x_d=0$, and decompose the projection $C_{d-1}$
recursively.

The $(d-2)$-hyperplanes that bound $C_{d-1}$ are projections of
intersections of the form $h \cap h^+$, $h \cap h^-$, for
$h \in H\setminus \{h^+,h^-\}$, including also $h^+\cap h^-$, if it
arises.  In principle, the number of such hyperplanes is at most
$2n-1$, but, as shown in the following lemma, the actual number
is smaller.


\begin{lemma}
  \lemlab{unique_charge}
  Let $\cell$ be a first-stage prism,
  whose ceiling and floor are contained in two respective
  hyperplanes $h^+$, $h^-$.  Then, for each hyperplane $h \in H$,
  $h \neq h^+, h^-$, then only one of $g^+ := h^+ \cap h$ or
  $g^- := h^- \cap h$ can appear on $\bd\cell$.  It is $g^+$ if $C$
  lies below $h$, and $g^-$ if $C$ lies above $h$.
  As a result, the projection $C_{d-1}$ of $\cell$ onto
  $x_d = 0$ has at most $n-1$ facets.
\end{lemma}

\begin{figure}[h]
    \centering%
    \includegraphics{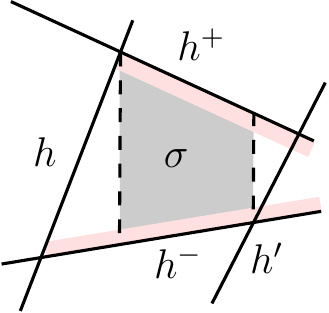}
    \figlab{unique_charge}
\end{figure}

\begin{proof}
    Let $\sigma^+$ (resp., $\sigma^-$) denote the ceiling (resp.,
    floor) of $\sigma$.  Let $f^+$ (resp., $f^-$) denote the facet of
    $C$ that contains $\sigma^+$ (resp., $\sigma^-$).  By
    construction, the vertical projection of $\sigma$ (onto $x_d=0$)
    is the intersection of the projections of $f^+$ and $f^-$. By
    construction, for each $(d-2)$-face of $f^+$, the $(d-2)$-flat
    that supports it is either an intersection of $h^+$ with another
    hyperplane that lies above $C$,
    or the intersection $h^+\cap h^-$.  Symmetric properties hold for
    the floor $f^-$. See Figure~\figref{unique_charge} for an illustration.
    This observation is easily seen to imply the
    assertions of the first part of the lemma.

    Regarding the second part of the lemma, since for a hyperplane $h$,
    only one of $g^+ := h^+ \cap h$ or $g^- := h^- \cap h$, $h \neq h^+, h^-$,
    can appear on $\bd\cell$,
    this contributes at most $n-2$ facets to $C_{d-1}$. 
    Together with the possible facet obtained from $h^+ \cap h^-$, we get the
    asserted bound $n-1$.
%
%
\end{proof}



Using \lemref{unique_charge}, 
we derive a recurrence relation for bounding the
complexity of the vertical decomposition of a single full-dimensional
cell $C$, as follows.

\begin{lemma}
    \lemlab{v:d:cell}%
    Let $K(d,n)$ denote the maximum number of (final) prisms in the
    vertical decomposition of a convex polyhedron in $\Re^d$ with at
    most $n$ facets. 
    Then
    \begin{equation}
        \eqlab{k:fact}%
        K(d,n) \leq \frac{1}{4^{d-2}}
        \pth{\frac{n!}{(n-d+2)!}}^2 (n-d+2)
        \leq \frac{n^{2d-3}}{4^{d-2}} .
    \end{equation}
\end{lemma}

\begin{proof}
    We first note that the number of pairs $(h^-,h^+)$ that can
    generate the floor and ceiling of a first-stage prism is at most
    $n^-n^+$, where $n^-$ (resp., $n^+$) is the number of facets on
    the lower (resp., upper) boundary of $C$. Since the maximum value
    of $n^-n^+$ is $n^2/4$, we obtain, by \lemref{unique_charge},
    \begin{equation}
        \eqlab{krec}%
        K(d,n) \leq \frac{n^2}{4} K(d-1,n-1) ,
    \end{equation}
    and (as is easily checked) $K(2,n) \leq n$. Solving the recurrence,
    we get
\begin{align*}
{\displaystyle K(d,n) \leq \frac{1}{4^{d-2}}
   \pth{\frac{n!}{(n-d+2)!}}^2 (n-d+2)},
\end{align*}
which is bounded by ${\displaystyle \frac{n^{2d-3}}{4^{d-2}}}$.
\end{proof}

\subsubsection{Vertical decomposition of the entire arrangement}

Next we prove the following theorem using \lemref{v:d:cell}. Note that
here we consider prisms of the vertical decomposition of every cell of
$\Arr(H)$, of any dimension.



\begin{theorem}
    \thmlab{vdsize}%
    For $n \ge 2d$, the complexity of the vertical decomposition of an arrangement of
    $n$ hyperplanes in $\Re^d$ is
    \begin{math}
        {\displaystyle O\pth{ \frac{4^d}{d^{7/2}} n^{2d} }},
    \end{math}
    with an absolute, $d$-independent constant of proportionality.
\end{theorem}

\begin{proof}
    Let $H$ be a set of $n$ hyperplanes in $\Re^d$, and let $\VD(H)$
    denote the vertical decomposition of the entire arrangement
    $\Arr(H)$.  The explicit construction of $\VD(H)$ is in fact
    equivalent to taking each cell $C$ of $\Arr(H)$ and applying to it
    the vertical decomposition procedure outlined above. A naive
    implementation of this reasoning gives a somewhat inferior bound
    on the overall number of prisms (in terms of its dependence on
    $d$; see the remark following the proof), so we use instead the
    following somewhat indirect argument.

    We first count the number of $d$-dimensional prisms in $\VD(H)$.
    Following~\corref{prisms:comb:dim}, each such prism in $\VD(H)$ is
    defined in terms of $b \leq 2d$ hyperplanes of $H$. In addition,
    it is easily verified that the vertical decomposition scheme
    complies with the Clarkson--Shor framework (see
    \secref{e:d:l:setup}). In particular, a prism $\cell$ that is
    defined by a subset $H_0$ of $b \leq 2d$ hyperplanes will
    appear as a prism of $\VD(H_0)$. By \eqqref{k:fact}, the overall
    number of $d$-dimensional prisms in the vertical decomposition of
    a single $d$-dimensional cell of $\Arr(H_0)$ is at most
    \begin{align*}
      O\pth{ \frac{d+2}{4^d} \pth{ \frac{(2d)!}{(d+2)!} }^2} ,
    \end{align*}
    with an absolute constant of proportionality, independent of $d$.
    Multiplying this by the number of $d$-dimensional cells of
    $\Arr(H_0)$, which is
    \begin{math}
        \sum_{j=0}^d \binom{b}{j} \leq 2^{b}
    \end{math}
    (see \lemref{num:cells}(a)), we get a total of
    \begin{align*}
      O\pth{ 2^b \pth{ \frac{d+2}{4^d}} \pth{ \frac{(2d)!}{(d+2)!} }^2}
    \end{align*}
    prisms. Finally, multiplying this bound by the number
    $\binom{n}{b}$ of subsets of $H$ of size $b$, and summing over
    $b=0,\ldots,2d$, we get a total of
  \begin{align*}
  O\pth{ \pth{ \sum_{b=0}^{2d} \binom{n}{b}2^b }\cdot
     \frac{d+2}{4^d} \pth{ \frac{(2d)!}{(d+2)!} }^2 }
  \end{align*}
  prisms. As is easily checked, for $n\ge 4d$, the sum $\sum_{b=0}^{2d} \binom{n}{b}2^b$
  is proportional to its last element, i.e., it is
  \begin{align*}
    O\pth{ \binom{n}{2d} 4^d} =
    O\pth{ \frac{n^{2d}}{(2d)!} 4^d} =
    O\pth{ \sqrt{d}\pth{ \frac{2e}{2d} }^{2d} n^{2d} } \ ,
  \end{align*}
  using Stirling's approximation.
  When $2d \le n < 4n$ we proceed as follows.
  We first observe that by the Multinomial Theorem it follows that
  \begin{align*}
    \sum_{b=0}^{2d} {n \choose b} 2^b = \sum_{b=0}^{2d} {n \choose b} \left(\sum_{i=1}^b {b \choose i} \right) = 3^n .
  \end{align*}
%
  We claim that
  ${\displaystyle 3^n \leq \sqrt{d}\pth{ \frac{2e}{2d} }^{2d} n^{2d}\approx \frac{n^{2d}}{(2d)!}}$
  for $2d \le n\le 4d$. Indeed, putting $x=n/(2d)$, this is equivalent to asserting that
  $\frac{3^x}{x} \leq d^{1/(4d) } 2e$ for $1 \le x\le 2$, which does indeed hold.
  Hence, the overall number of prisms is
  \begin{align*}
  O\pth{ \frac{n^{2d}}{(2d)!} 4^d\cdot \frac{d+2}{4^d} \pth{
        \frac{(2d)!}{(d+2)!} }^2} = O\pth{ \frac{n^{2d}}{d^3} \cdot
     \binom{2d}{d} } = O\pth{ \frac{4^d}{d^{7/2}} n^{2d} } .
  \end{align*}
  So far we have accounted only for full-dimensional prisms, in the
  decomposition of full-dimensional cells of $\Arr(H)$. The number of
  prisms of all dimensions in $\VD(H)$, where prisms of dimension $j$
  arise in the vertical decomposition within some $j$-flat formed by
  the intersection of a corresponding subset of $d-j$ hyperplanes of
  $H$,
  is easily seen to be bounded by:
  
  \begin{align*}
    O\pth{ \sum_{j=0}^d \binom{n}{d-j} \cdot \frac{4^j}{j^{7/2}} n^{2j} }
    = O\pth{ \sum_{j=0}^d \frac{n^{d-j}}{(d-j)!} \cdot \frac{4^j}{j^{7/2}} n^{2j} }
  \end{align*}
  \begin{align*}
    = O\pth{ \sum_{j=0}^d \frac{n^{d+j}}{(d-j)!} \cdot
      \frac{4^j}{j^{7/2}} } = O\pth{ \frac{4^d}{d^{7/2}} n^{2d} } .
  \end{align*}

\end{proof}



As a matter of fact, the preceding proof also implies the following
stronger statement.
\begin{corollary} \corlab{vd:all:p}
  Let $H$ be a set of $n \ge 2d$ hyperplanes in $\Re^d$.
  The overall number of prisms that can arise in the
  vertical decomposition of the arrangement of any subset of $H$ is
  ${\displaystyle O\pth{ \frac{4^d}{d^{7/2}} n^{2d} }}$, with an
  absolute, $d$-independent constant of proportionality.
\end{corollary}

\begin{remark}
    The bound in \thmref{vdsize} is a significant improvement over the
    previous upper bound of \cite{cegs-sessr-91} in terms of the
    dependence of its coefficient on $d$, from $2^{O(d^2)}$ to less
    than $4^d$.  We emphasize, however, that we pay a small price in
    terms of the dependence on $n$, which is $n^{2d}$ in the new
    bound, instead of $n^{2d-4}$ in~\cite{koltun} (and only $O(n^d)$
    if one uses instead bottom-vertex triangulation, by
    \lemref{simplices:grow}). The previous analysis of the complexity
    of the vertical decomposition in~\cite{cegs-sessr-91} is global
    (unlike ours, which is local in each
    cell), 
    and yields a recurrence relation for the entire number of prisms
    in the vertical decomposition of $n$ hyperplanes in $d$
    dimensions.  Denoting this function by $V(n,d)$, it is shown
    in~\cite{cegs-sessr-91} that $V(n,d) \leq n^2 V(2n,d-1)$, which has
    the solution $V(n,d) \leq 2\cdot 2^{d(d-1)} n^{2d-1}$ (when we stop
    at $d=1$, in which case $V(n,d) \leq n+1 \leq 2n$).  The dependence
    of this bound on $n$ can be reduced to $n^{2d-4}$ (by a refined,
    and rather complicated, analysis in four dimensions, and using it
    as the base case for this recurrence, see~\cite{koltun}) but not
    the dependence on $d$.  In contrast, the coefficient in our bound
    is only singly exponential in $d$. As will follow later, this is a
    crucial property for improving the query cost in Meiser's point
    location algorithm.
    We pay a small cost (it is small when $d$ is large) in that the power of $n$ is
    larger by a small constant. It would be interesting to establish a
    bound that is both singly exponential in $d$ and has the smaller
    power $2d-4$ of $n$.

    We also note that it is an open problem whether the complexity of
    $\VD(H)$ is really asymptotically larger than $O(n^d)$. If it were
    not, the preceding discussion would become vacuous.
\end{remark}

\thmref{vdsize} implies that the local growth function of $\VD$
satisfies $u(m) \leq 2^\alpha m^\beta$, with $\alpha = 2d$ and
$\beta=2d$. That is, we have $\alpha, \beta, b = 2d$. Substituting
these values into \lemref{weak:cutting}, we obtain the following
result.

\begin{corollary}
    Let $H$ be a set of $n$ hyperplanes in $\Re^d$ and let
    $\BadProb \in (0,1)$ be some confidence parameter.  Then, for
    $\rho \geq cr(d \log (rd)+\ln \frac{1}{\BadProb})$, for some
    suitable absolute constant $c$, the vertical decomposition of a
    $\rho$-sample is a $(1/r)$-cutting with probability
    $\geq 1-\BadProb$.
\end{corollary}

\subsection{The shatter dimension and VC-dimension of prisms}


So far we have only considered ``definable'' prisms that arise in the
vertical decompositions of samples from $H$. In this subsection we
extend them to arbitrary ``similarly looking'' vertical prisms, a
notion to be formally defined momentarily, use these prisms to define
a natural range space on $H$, and analyze the VC-dimension and the
primal shatter dimension of this space.

\paragraph{Parameterizing a prism.}

Let $\cell$ be a vertical prism that arises in the vertical
decomposition of a set of hyperplanes in $\Re^d$, such that it is
defined by exactly $2d$ of these hyperplanes; the case of fewer
defining hyperplanes is discussed later on. It easily follows from the
construction that $\cell$ can be represented as the intersection
region of $2d$ halfspaces of the form
\begin{align}
  \eqlab{para:prism}
  b^-_1
  &
    \quad\leq\quad x_1\quad  \leq\quad b^+_1 \nonumber \\
  a^-_{2,1}x_1 + b^-_2
  &
    \quad\leq\quad x_2 \quad \leq\quad%
    a^+_{2,1}x_1 + b^+_2 \nonumber \\
  a^-_{3,1}x_1 + a^-_{3,2}x_2 + b^-_3
  &
    \quad\leq\quad x_3\quad  \leq\quad
    a^+_{3,1}x_1 +  a^+_{3,2}x_2 + b^+_3 \\
  &
    \quad\qquad\vdots
  & \nonumber \\
  a^-_{d,1}x_1 + \cdots + a^-_{d,d-1}x_{d-1} + b^-_d
  &
    \quad\leq%
    \quad x_d \quad
    \leq\quad%
    a^+_{d,1}x_1 + \cdots + a^+_{d,d-1}x_{d-1} + b^+_d
    ,
    \nonumber
\end{align}
for suitable parameter $a^\pm_{i,j}$ and $b^\pm_j$. The construction
produces these inequalities in reverse order: the last pair of
inequalities define the floor and ceiling of $\cell$, and each
preceding pair of inequalities define the floor and ceiling in the
corresponding lower-dimensional projection of $\cell$. When the number
of hyperplanes defining $\cell$ is smaller than $2d$, some of these
inequalities are absent.

Let $\Sigma$ denote the set of all prisms of the form
\eqqref{para:prism} (including those prisms defined by fewer than $2d$
inequalities). We define the range space $(H,\Sigma)$ on $H$, so that,
for each $\cell\in \Sigma$, the range associated with $\cell$ is the
set of hyperplanes of $H$ that cross (the interior of) $\cell$.

Note that the (maximum) overall number of parameters that define a
prism $\cell\in \Sigma$ is $D=d(d+1)$.  We can therefore represent a
vertical prism as a point in $\Re^D$ (or, rather, in the portion of
$\Re^D$ consisting of points that represent nonempty prisms).  Let $h$
be a fixed hyperplane in $\Re^d$. Let $K_h$ denote the region in
$\Re^D$ consisting of all (valid) points that represent prisms that
are crossed by $h$. The boundary $S_h$ of $K_h$ is the locus of all
(points representing) prisms for which $h$ is a supporting hyperplane.

Consider a vertical prism $\cell\in\Sigma$. In general, $\cell$ has
$2^d$ vertices, each of which is obtained by choosing one inequality
(the left or the right) out of each of the $d$ pairs in
\eqqref{para:prism}, turning each chosen inequality into an equality,
and solving the resulting linear system. We \emph{label} each vertex
$q$ of $\cell$ by the sign sequence $(\eps_1,\ldots,\eps_d)$, where,
for each $i$, $\eps_i$ is $-1$ (resp., $+1$) if the left (resp.,
right) inequality of the $i$\th pair is the one involved in the
construction of $q$.

The following analysis fixes $h$, with its normal vector $v$ (actually
two oppositely directed vectors), and constructs $K_h$ in two
stages. It first constructs a partition of $\Re^D$ into cells, so
that, for a fixed cell, and for all the prisms $\cell$ whose
representing points lie in the cell, the two vertices of $\cell$
supported by hyperplanes parallel to $h$ have fixed labels. These
cells are independent of the coordinates that represent the free
coefficients $b_j^{\pm}$.  In the second stage, we take each of these
``cylindrical'' cells, and extract from it the portion that represents
prisms that are crossed by $h$ (these are prisms for which $h$ lies in
between the two hyperplanes that are parallel to $h$ and support
$\cell$ at the two vertices just mentioned). The union of these
regions is the desired $K_h$.

\paragraph{Partitioning by the labels of the supported vertices.}
As already said, in what follows, the hyperplane $h$, and its pair of
normal directions $v$ and $-v$, are fixed, and the prism $\cell$ is
regarded as a parameter.

Let $q$ be the vertex of $\cell$ with label $\lambda$. We express
algebraically the property that $q$ is the contact vertex of a
supporting hyperplane of $\cell$ with outward normal direction $v$, as
follows. Let $f_1,\ldots,f_d$ be the $d$ facets of $\cell$ incident to
$v$. The corresponding $d$ (not necessarily normalized) outward normal
directions $w_1,\ldots,w_d$ of the hyperplanes of \eqqref{para:prism}
that support these facets are given by
\begin{align*}
  w_1 & = \pm(1,0,\ldots,0) \\
  w_2 & = \pm(-a^\pm_{2,1}, 1, 0,\ldots,0) \\
  w_3 & = \pm(-a^\pm_{3,1}, -a^\pm_{3,2}, 1,0,\ldots,0) \\
      & \vdots \\
  w_d & = \pm(-a^\pm_{d,1}, \ldots, -a^\pm_{d,d-1}, 1) ,
\end{align*}
where, for each $j$, the sign of $w_j$ (and the corresponding set of
coefficients $a_{j,i}^{\pm}$) is the sign of the corresponding
component of $\lambda$.

%

The vertex $q$ is the contact vertex of a supporting hyperplane of
$\cell$ with outward normal direction $v$ if and only if $v$ is in the
convex cone generated by $w_1,\ldots,w_d$. The algebraic
characterization of the latter condition is that $v$
can be expressed as a \emph{nonnegative} linear combination
$v = \sum_{j=1}^d \beta_j w_j$ of $w_1,\ldots,w_d$ (i.e., such that
$\beta_j \geq 0$ for each $j$).
The triangular structure of the coefficients of the $w_j$'s makes it
easy to express the $\beta_j$'s in terms of the parameters
$a^\pm_{j,i}$. It is simplest to reason by using Cramer's rule for
solving the resulting system, noting that the determinant in the
denominator is $\pm 1$ (since the resulting matrix is triangular with
$\pm 1$ on the diagonal).  The condition that all the $\beta_j$'s be
nonnegative is therefore a system of $d$ algebraic inequalities in the
parameters $a^\pm_{i,j}$, each of which involves a polynomial (namely,
the determinant in the corresponding numerator) of degree strictly
smaller than $d$. Let $T_q(v)$ denote the ``feasible region'', which
is the solution of these inequalities in $\Re^{D'}$, for
$D' = d(d-1)/2$, which is the subspace of the parameter space $\Re^D$
that is spanned by the coordinates $a_{j,i}^\pm$ that participate in
the (in)equalities that define $q$.

We repeat the above analysis to each of the $2^d$ possible labels of
vertices $q$ of a prism, with the fixed direction $v$.  Each label
$\lambda$ uses a different set of $D'$ coordinates from among the
$a_{j,i}^\pm$.  Let $D_0 = 2D' = d(d-1)$ be the total number of
coordinates $a_{j,i}^\pm$.  Extend each region $T_q(v)$ in all the
complementary, unused, $D'$ coordinates into a suitable cylindrical
region within $\Re^{D_0}$.


\begin{remark}
    Note that so far we only handle bounded prisms, namely those
    defined by exactly $2d$ inequalities (as in \eqqref{para:prism}).
    It is easy to extend the analysis to prisms
    that use fewer inequalities: There are $3^{d}$ subsets of
    inequalities, each of which corresponds to a different representation of $\sigma$
    (for each $j$, we take from the $j$\th pair the left inequality, the right inequality, or neither
    of them), and we repeat the ongoing analysis to each of them, more
    or less verbatim, except that the dimension of the parametric
    space that represent prisms is smaller.
    Multiplying the bounds
    that we are going to get by $3^{d}$ will not affect the asymptotic
    nature of the analysis---see below.
\end{remark}

We claim that the regions $T_q(v)$, over all possible vertices $q$,
form a decomposition of $\Re^{D_0}$ into $2^d$ pairwise openly
disjoint regions. It is indeed a cover because each point
$\zeta\in \Re^{D_0}$ is the projection (in which the coordinates
$b_j^\pm$ are ignored) of infinitely many valid prisms $\sigma$; this
holds because there is always a solution to the system
\eqqref{para:prism} if we choose $b_j^-$ and $b_j^+$ far apart from
one another, for each $j$.  Any such prism $\sigma$ has at least one
vertex $q$ supported by a hyperplane with an outward normal direction
$v$, and therefore $\zeta$ must belong to $T_q(v)$.  Moreover, if
these prisms have more than one such vertex, then it is easily seen
that $\zeta$ cannot lie in the interior of any region $T_q(v)$.

Denote the resulting subdivision of $\Re^{D_0}$ as $\M_h^+$.

To complete this part of the construction, we apply the same
construction to the opposite normal vector $-v$, and obtain a second
subdivision of $\Re^{D_0}$, which we denote by $\M_h^-$. We then apply
both steps to each hyperplane $h\in H$, and obtain a collection of
$2n$ subdivisions $\{\M_h^+, \M_h^- \mid h\in H\}$ of $\Re^{D_0}$.

\paragraph{Partitioning by the positions of the supporting hyperplanes.}
So far the analysis has only focused on the parameters $a_{j,i}^\pm$,
and ignored the parameters $b_i^\pm$. This was sufficient in order to
classify the family of prisms into subsets depending on the discrete
nature of the two vertices supported by each input hyperplane. Now we
want to bring the parameters $b_i^\pm$ into the game, thereby further
classifying the prisms by distinguishing, for each input hyperplane
$h$, between those prisms crossed by $h$ and those that $h$
misses. This proceeds as follows.

Fix again a hyperplane $h\in H$, and its two directed normal vectors
$v^+$, $v^-=-v^+$, which give rise to the respective subdivisions
$\M_h^+$, $\M_h^-$. In the second stage of the construction, we
construct the desired region $K_h$ in $\Re^D$, by ``lifting'' (and
then overlaying) $\M_h^+$ and $\M_h^-$ to $\Re^D$, as follows.

Fix $v\in\{v^+,v^-\}$, and assume, without loss of generality, that
$v=v^+$.  We take each cell $T_q(v)$ of $\M_h^+$, consider the
cylinder $T_q^*(v)$ over $T_q(v)$ in the remaining $2d$ dimensions
that encode the coordinates $b_j^\pm$ (i.e.,
$T_q^*(v) = T_q(v)\times\Re^{2d}$), and partition $T_q^*(v)$ by the
surface $T_q^0(v)$, which is the locus of all (points encoding valid)
prisms for which $h$ passes through the vertex $q$.  Concretely,
recall that the label of $q$ determines the $d$ inequalities, one out
of each pair in \eqqref{para:prism}, which we turn into equalities,
and solve the resulting linear system (a trivial step, using backward
substitution), to obtain $q$ itself (in terms of the $D$ parameters
defining $\cell$).  The equation defining $T_q^0(v)$ is then obtained
by substituting this $q$ into the linear equation defining $h$.

Write $T_q^0(v)$ as the zero set $F_{q,v}=0$ of a suitable function
$F_{q,v}$, which, by the triangular structure of \eqqref{para:prism},
is a polynomial (in $D/2$ out of the $D$ coordinates) of degree
$d$. The surface $T_q^0(v)$ partitions $T_q^*(v)$ into two regions. In
one of them, $h$ passes above its parallel copy that supports the
corresponding prism and has outward normal vector $v$, and in the
other region $h$ passes below that hyperplane.

We apply a similar construction for $h$ and its opposite normal vector
$v^-=-v$, and then repeat the whole procedure for all $h\in H$ (and
for all labels of vertices $q$ of the prism).

\paragraph{The resulting $D$-dimensional arrangement.}
In the final step of the construction, we form the \emph{overlay} $\O$
of the subdivisions $\M_h^+$, $\M_h^-$, over all $h\in H$, and obtain
a subdivision of $\Re^{D_0}$ into regions, so that each region $Q$ has
a fixed sequence of labels,
$(\lambda_1^+,\lambda_1^-,\ldots,\lambda_n^+,\lambda_n^-)$, so that,
for each prism represented by a point in $Q$, and for each $j$, the
two hyperplanes that support the prism and are parallel to the $j$\th
hyperplane $h_j$ of $H$ touch it at the two vertices with labels
$\lambda_j^+$, $\lambda_j^-$.

For each region $Q$, we draw the $2n$ surfaces
$T_{\lambda_j^+}^0(v_j)$, $T_{\lambda_j^-}^0(-v_j)$, where $v_j$ is
the positive normal vector of $h_j$, and form their arrangement within
the corresponding cylinder $Q^*=Q\times R^{2d}$. We take the union of
these arrangements, over all the regions $Q$, and obtain a final
subdivision $\Xi$ of $\Re^D$ into cells, so that in each cell, the
prisms represented by its points have a fixed conflict list of
hyperplanes that cross them.

\paragraph{The complexity of the subdivision.}
We next bound the complexity of $\Xi$. Let us first bound the number
of vertices of $\Xi$. Each such vertex is formed by the intersection
of $D$ surfaces, each of which is either a vertical cylindrical
surface obtained by lifting a boundary surface between two regions
$T_q(v)$, $T_{q'}(v)$ (such a surface is obtained by turning an
inequality $\beta_j\ge 0$ in the definition of $T_q(v)$ into an
equality), or a surface of the form $T_q^0(v)$. That is, to obtain a
vertex of $\Xi$, we need to choose (with repetitions) $D$ hyperplanes
from $H$, say $h_1,h_2,\ldots,h_D$, choose for each of them the
corresponding outward normal direction $v_1,v_2,\ldots,v_D$ (there are
two choices for each $v_i$), associate with each of them a respective
label of a vertex $q_1,q_2,\ldots,q_D$ of the prism (again, possibly
with repetitions), and, for each $j$, choose either one of the $d$
surfaces defining $T^*_{q_j}(v_j)$ or the surface
$T_{q_j}^0(v_j)$. The number of choices is thus
\begin{equation}
    \eqlab{m:h:comp}%
    X = n^D\cdot 2^D \cdot (2^d)^D \cdot
    (d+1)^D = 2^{(d+1+\log(d+1))D} n^D .
\end{equation}

Finally, by \Bezout's theorem, each such $D$-tuple of surfaces
intersect in at most $d^D$ vertices (recalling that the degree of each
surface is at most $d$).  Altogether, the overall number of vertices
of $\Xi$ is at most
\begin{align*}
  X d^D = X 2^{D\log d}
  =  2^{(d+1+2\log(d+1))D} n^D
  = 2^{d(d+1)^2+2d(d+1)\log(d+1)} n^{d(d+1)} = 2^{O(d^3)}n^{d(d+1)} .
\end{align*}
It is easy to see that this bound on the number of vertices also
controls the number of cells of $\Xi$; that is, the number of cells is
also $2^{O(d^3)}n^{d(d+1)}$.
This is therefore an upper bound on the number of subsets of $H$ that
can be crossed by a prism. That is, the global growth function of our
range space $(H,\Sigma)$ satisfies
\begin{align*}
  g_\Sigma(n) = 2^{O(d^3)}n^{d(d+1)} .
\end{align*}
Hence,
the primal shatter dimension is then at most $d(d+1)$.

However, the only upper bound that we can get on the VC-dimension is
$O(d^3)$. This follows from the standard reasoning, where we denote by
$x$ the largest cardinality of a set of hyperplanes that can be
shattered by vertical prisms of $\Sigma$, and then use the inequality
\begin{align*}
2^x \leq 2^{O(d^3)} x^{d(d+1)} , \qquad\text{or}\qquad x = O(d^3) +
d(d+1)\log x = O(d^3) .
\end{align*}

\medskip We remark that the factor $2^{O(d^3)}$ that appears in the
bound for $g_\Sigma(n)$ arises from the factor $\pth{2^d}^D$ in
\eqqref{m:h:comp}. This raises a challenging open question, of whether
all $D$-tuples of vertices of a prism can be the contact vertices of
$D$ given hyperplanes. This can also be stated as a question of
whether the complexity of the overlay $\O$ can really be so high. At
the moment, we do not have any solid conjecture to offer.

\paragraph{Lower bound.}

We next show that the VC-dimension of hyperplanes and vertical prisms,
that is, of the range space $(H,\Sigma)$, is at least $1+d(d+1)/2$.
That is, we construct a set $H$ of $1+d(d+1)/2$ hyperplanes in
$\Re^d$, and a set $\C\subset\Sigma$ of $2^{1+d(d+1)/2}$ vertical
prisms in $\Re^d$, each having a representation as in
\eqqref{para:prism}, such that the prisms in $\C$ shatter $H$. That
is, for every subset $S\subseteq H$ there is a prism $\tau(S)\in \C$
for which $S$ is the subset of the hyperplanes of $H$ that intersect
(the interior of) $\tau(S)$.

All the prisms in $\C$ will be sufficiently small perturbations of the
unit cube $U = [0,1]^d$, and each hyperplane in $H$ will be a
supporting hyperplane of $U$, touching it at a unique corresponding
vertex. To simplify and unify the construction, we will only specify
the vertex $v$ at which the hyperplane supports $U$, and the
hyperplane itself, denoted as $h_v$, will be the hyperplane that
passes through $v$ and is perpendicular to the segment $vo$,
connecting $v$ with the center $o = (1/2,1/2,\ldots,1/2)$ of $U$.

\paragraph{Constructing $H$.}
We begin with the construction of $H$, or rather of the set $P$ of
vertices of $U$ that are supported by the hyperplanes in $H$, as just
defined. Let $e_i$ denote the unit vector in the $x_i$-direction, for
$i=1,\ldots,d$.  The set $P$ is the union of disjoint sets
$P_1,\ldots,P_d$, where $P_1 = \{{\bf 0},e_1\}$, where ${\bf 0}$ is
the origin, and for $i\ge 2$,
$P_i = \{e_i\} \cup \{e_j + e_i \mid j < i \}$, that is, $P_i$ consists of $e_i$
and $i-1$ vectors, which are sums of pairs of the form $e_j + e_i$.
Intuitively, except for $P_1$, which is somewhat special, $P_i$ lies in the
$(i-1)$-dimensional flat $x_i=1$,
$x_{i+1} = x_{i+2} = \cdots = x_d = 0$, and consists of the (suitably
shifted) origin and of the $i-1$ unit vectors within that subspace.
Since $|P_1|=2$ and $|P_i| = i$ for $i\ge 2$, it follows that
$|P| = 1+d(d+1)/2$.  Note that $P$ consists exactly of all the
vertices of $U$ with at most two nonzero coordinates. We denote by
$H_i$ the subset of $H$ that corresponds to the vertices in $P_i$, for
$i=1,\ldots,d$.

\paragraph{Constructing $\C$.}
We define a sequence of perturbation parameters
$\eps_1,\eps_2,\ldots,\eps_d$, such that
$\eps_i = \sqrt{2(d+1)} \eps_{i-1}$, for $i=2,3,\ldots,d$, and
$\eps_1$ is chosen sufficiently small so as to guarantee that $\eps_d$
is also sufficiently small, say smaller than some prescribed
$\eps_0\ll 1$. By definition, we have
$\eps_d < \pth{2(d+1)}^{d/2}\eps_1$, so we choose
$\eps_1 = \frac{\eps_0}{(2(d+1))^{d/2}}$, to obtain the desired
property.

Now let $S$ be an arbitrary subset of $H$. We will construct a
vertical prism $\tau=\tau(S)$, such that
$S = \{h\in H \mid h\cap{\rm int}(\tau)\ne\emptyset\}$. The
construction of $\tau$ proceeds inductively on the dimension, and
produces a sequence of prisms $\tau_1,\tau_2,\ldots,\tau_d=\tau$, such
that, for each $i$, $\tau_i$ is of dimension $i$, in the sense that it
is contained in the flat
$\{ x\in\Re^d \mid x_j = 0 \text{ for all } i < j \leq d \}$.
Specifically, $\tau_i$ is obtained from $\tau_{i-1}$ by
\begin{align*}
\tau_i = \Bigl\{(x_1,x_2,\ldots,x_i,0,\ldots,0) \mid
(x_1,\ldots,x_{i-1},0,\ldots,0) \in\tau_{i-1}, \text{ and } 0\le x_i
\le \sum_{j<i} a^+_{i,j}x_j + b^+_i \Bigr\} ,
\end{align*}
for suitable coefficients $a^+_{i,j}$ and $b^+_i$. Clearly, the final
prism $\tau_i$ has (a special form of) the representation given in
\eqqref{para:prism}.

Let $S_j = S \cap H_j$, for $j=1,\ldots,d$. We will enforce, by
induction, that for each $i$, the prism $\tau_i$ satisfies the
following two conditions for every $j\le i$.
\begin{enumerate}
    \item \itemlab{co1} Every vertex $p\in P_j$ is $\eps_j$-close to
    some vertex of $\tau_i$.
    \item \itemlab{co2}
    $S_j = \{h\in H_j \mid h\cap{\rm int}(\tau_i)\ne\emptyset\}$,
    where ${\rm int}(\tau_i)$ denotes the relative interior of the
    $i$-dimensional prism $\tau_i$.
\end{enumerate}
We begin the construction of $\tau$ by constructing $\tau_1$, which is
simply $\pm\eps_1\le x_1\le 1\pm\eps_1$, where the signs are chosen so
that Condition \itemref{co2} holds for $j=i=1$ (Condition
\itemref{co1} obviously holds).

Suppose now that $i>1$ and that we have already constructed
$\tau_{i-1}$, which satisfies Conditions 1 and 2.  We construct
$\tau_{i}$ as follows.

%

For each $q\in P_i$ let $q'$ be
the vertex obtained from $q$ by replacing the $1$ in its $i$\th
coordinate by $0$.  Clearly, $q'\in P_j$ for some $j < i$. Since
$\tau_{i-1}$ satisfies Condition \itemref{co1}, there exists a
(unique) vertex $w$ of $\tau_{i-1}$ that lies at distance at most
$\eps_{j}$ from $q'$. Let $z$ denote the point at which the
$x_i$-parallel line through $w$ meets $h_q$. Write $z$ as
$w+(1+t)e_i$, where, as above, $e_i$ is the unit vector in the
$x_i$-direction and $t$ is a real parameter, and note that the
coordinates $x_{i+1},\ldots,x_d$ of $z$ are all $0$. Since $z\in h_q$,
we have $(z-q)\cdot (q-o) = 0$.  That is, we have
\begin{align*}
  (w+(1+t)e_i - q'-e_i)\cdot (q-o) & = ((w-q') + te_i)\cdot (q-o) = 0 , \quad\text{or} \\
  te_i\cdot (q-o) & = (q'-w) \cdot (q-o) .
\end{align*}
We have $q-o = (\pm 1/2,\pm 1/2,\ldots,\pm 1/2)$, so
$te_i\cdot (q-o) = \pm t/2$. We thus get
\begin{align*}
|t|/2 = \left| (q'-w) \cdot (q-o)\right| \leq \|q'-w\|\cdot \|q-o\| \le
\frac{\sqrt{d}}{2} \eps_{i-1} .
\end{align*}
That is, $|t| \leq \sqrt{d}\eps_{i-1}$. It thus follows that
\begin{align*}
\|z-q\| = \|w+(1+t)e_i - q'-e_i\| \leq \|w-q'\| + |t| \leq \pth{1 +
    \sqrt{d} }\eps_{i-1} .
\end{align*}
We now replace $z$ by $z_q:=z\pm se_i$, where
\begin{align*}
s = \pth{ \sqrt{2(d+1)}-(1+\sqrt{d})}\eps_{i-1} ;
\end{align*}
it is easily checked that $s>0$ for $d\ge 2$. The sign of $se_i$ is
positive (resp., negative) if $h_q\in S_i$ (resp., $h_q\not\in S_i$).
Applying this procedure to each vertex in $P_i$, we get $i$ points
$z_q$, for $q\in P_i$, such that we have
\begin{align*}
  \|z_q-q\| & \leq \pth{1 + \sqrt{d} }\eps_{i-1} + s , \quad\text{or} \\
  \|z_q-q\| & \leq \sqrt{2(d+1)} \eps_{i-1} = \eps_i ,
\end{align*}
for each $q$.

Let $\pi_i$ denote the hyperplane that passes through the $i$ points
$z_q$, for $q\in P_i$, and is parallel to all coordinates
$x_{i+1},\ldots,x_d$. Write its equation as
$x_i = \sum_{j<i} a^+_{i,j}x_j + b^+_i$ (it is easily checked that
$\pi_i$ is not $x_i$-parallel, assuming that the $\eps_j$'s are
sufficiently small).  We then define $\tau_i$ to be the prism
\begin{align*}
\tau_i = \Bigl\{(x_1,x_2,\ldots,x_i,0,\ldots,0) \mid
(x_1,\ldots,x_{i-1},0,\ldots,0) \in\tau_{i-1}, \text{ and } 0\le x_i
\le \sum_{j<i} a^+_{i,j}x_j + b^+_i \Bigr\} ,
\end{align*}
It is easy to verify that Conditions \itemref{co1} and \itemref{co2}
hold for $\tau_i$, assuming that the $\eps_j$'s are sufficiently
small. In conclusion, we obtain:

\begin{theorem}
    \thmlab{v:c:dim:prisms}%
    The VC-dimension of the range space $(H,\Sigma)$ of hyperplanes
    and vertical prisms in $\Re^d$ is at least $1+d(d+1)/2$ and at
    most $O(d^3)$.
\end{theorem}

An interesting open question is to tighten this gap. We conjecture
that the VC-dimension is quadratic, or nearly quadratic in $d$.




\section{Point location in an arrangement of hyperplanes}%
\seclab{sec:meiser}

Let $H$ be a set of $n$ hyperplanes in $\Re^d$. We wish to preprocess
$H$ into a data structure for point location in $\Arr(H)$. There are
several variants that depend on how we want the output to a query to be.
We  consider here two variants. In the first one we return, for a query point $q$,
the (possibly lower-dimensional) cell of $\Arr(H)$ that contains it, which we
represent by the sign pattern of $q$ with respect to the entire set $H$.
In the second variant, referred to as \emph{vertical ray shooting}, we return
the first hyperplane that is hit by the upward-directed ray emanating from $q$
(in the positive $x_d$-direction), including the case where $q$ lies on one or several
hyperplanes of $H$, in which case we return one of them.
The goal is to make the query as efficient as possible, in its dependence on both
$n$ and $d$, at the cost of large (but not too large) storage and preprocessing time.

We begin in \secref{meiser-bvt} by reviewing Meiser's algorithm~\cite{m-plah-93} (see also the introduction) in
some detail, and by providing a rigorous analysis of the tradeoff
between the query time and storage; the analysis of this tradeoff, as presented in
\cite{m-plah-93}, is sketchy and suffers from several technical difficulties.
We also discuss the improvement of Meiser's algorithm, due to Liu~\cite{l-nplah-04},
which reduces and tightens the dependence of the storage size on $n$ (but not on $d$---this
dependence is actually worse than what we obtain here). In \secref{meiser-bvt} we
only consider the variant where we return the (sign pattern of the)  cell of
$\Arr(H)$ containing the query point.

We then improve the query time, in \secref{meiser-vd}, by replacing bottom-vertex triangulation
by vertical decomposition, using the Clarkson--Shor random sampling analysis
(which is based on the combinatorial dimension), as presented in \secref{sect2}.
This saves about a factor of $d$ in the query time, and constitutes, in our opinion, one of the main
contributions of this paper.

A second, smaller improvement, by a $\log d$ factor, is obtained by using the optimistic
sampling technique, presented in \secref{sec:optim}, which makes do with a slightly
smaller sample size. This leads to a slightly faster processing of a single recursive
step (on the size of $H$) of the query. This refined approach is presented in \secref{s:opt}.

Finally, in \secref{sec:low}, we consider the case of low-complexity hyperplanes, and present a simpler variant of our data structures for this case. It has comparable bounds on the storage and preprocessing, but the query is faster by roughly a factor of $d$.

\subsection{Meiser's algorithm, with an enhancement}%
\seclab{meiser-bvt}


Essentially, the only known algorithm for point location, in which the query time is \emph{polynomial}
in $d$ and in $\log n$, is due to Meiser~\cite{m-plah-93}
(it has been improved in a follow-up study by Liu~\cite{l-nplah-04}).



Meiser's preprocessing stage constructs the following data structure. It draws a
$\rho$-sample $R\subset H$, of actual size at most $\rho$, where the specific choices of $\rho$ are
discussed later, and constructs the bottom-vertex triangulation $\BT(R)$, as
described in \secref{b:v:t:construction}. By its recursive nature, this construction
also produces the bottom-vertex triangulation within every flat, of any dimension $1\le j\le d-1$,
that is the intersection of some $d-j$ hyperplanes of $R$.

The construction  identifies each cell
of $\Arr(R)$, of any dimension, by its sign pattern with respect to $R$.  We store the sign patterns,
of all the cells of all dimensions, in a ternary trie, denoted by $T=T_R$, whose leaves correspond
to the cells. We represent each cell $C$ by a structure containing, among other features, its bottom
vertex $w_C$, which is represented by the $d$ hyperplanes that define it (or, in case of degeneracies,
the $d$ hyperplanes of smallest indices that define it).
Each leaf $\xi$ of $T$ stores a pointer to the structure of its cell $C_\xi$.


Let $C= C_\xi\in \Arr(R)$ be the cell associated with a leaf $\xi$ of the trie $T$. We store
at $\xi$ (or rather at the structure associated with $C_\xi$, which is accessible from $\xi$) a
secondary tree $Q_C$ that ``encodes'' the bottom vertex triangulation $\BT(C)$ of the cell $C$,
as follows. The root of $Q_C$ corresponds to the cell $C$ itself, and each internal node of $Q_C$
corresponds to some (lower dimensional) cell on (i.e., a face of) $\bd C$.
(There could be several nodes that correspond to the same subcell in $Q_C$.)
A node $v$ in $Q_C$, which corresponds to some $j$-dimensional face $C'_v$ of $C$,
has a child for each $(j-1)$-dimensional face on $\bd C'_v$ that does not contain $w_{C'_v}$.
We index each child of $v$ by the hyperplane that supports the corresponding $(j-1)$-dimensional
face of $C'_v$ (but does not support $C'_v$ itself). In this manner, each leaf $v$ of $Q_C$ is
associated with a vertex of $\bd C$. We store at $v$ the simplex of $\BT(C)$, which is spanned
by the lowest vertices $w_{C'_u}$ of the cells associated with the nodes $u$ on the path to $v$,
including the vertex associated with $v$. It is easily verified that all 
simplices of $\BT(C)$ (whose dimension is equal to that of $C$)
are obtained in this manner (and only these simplices).
Repeating this construction for every cell $C$ of $\Arr(R)$, we obtain all the simplices of $\BT(R)$.

The algorithm then constructs the conflict list $\KillSet{\cell}$ for each simplex $\cell\in\BT(R)$,
in brute force, by checking for each hyperplane $h\in H$ whether it separates the vertices of $\cell$.
The actual procedure that implements this step is presented later.

We process each of these conflict lists recursively. That is, for each simplex $\cell$, we draw
a $\rho$-sample $R'$ from $\KillSet{\cell}$, compute $\BT(R')$, and construct the corresponding
trie $T_{R'}$ and the trees $Q_C$ for each cell $C\in \Arr(R')$. We store $R'$ and the data structures
associated with it, at the leaf $v_\cell$ representing $\cell$ in the tree $Q_C$ of the cell
$C \in \Arr(R)$ containing $\cell$. (We do not need to keep the conflict lists themselves in
the final data structure, because the queries do not access them.)

At each recursive step, when we construct, for each simplex $\cell\in\BT(R)$, the conflict list
$\KillSet{\cell}$, we also get the (fixed) sign pattern of all the points in $\cell$ with respect
to all the hyperplanes that do not cross $\cell$, and we temporarily store this subpattern at $v_\cell$.

The recursion bottoms out when the conflict list of the current simplex $\cell$ is of size smaller
than $\rho$. We call such a simplex $\cell$ a {\em leaf-simplex} (of the entire structure).
A leaf-simplex does not have a recursive structure associated with it. We only construct the
arrangement $\Arr(\KillSet{\cell})$ of the conflict list $\KillSet{\cell}$ of $\cell$ (whose
size is at most $\rho$), and store at $v_\cell$ a trie structure $T_{\KillSet{\cell}}$ over
the cells of $\Arr(\KillSet{\cell})$.  We attach to each cell $C$ (of any dimension) of
$\Arr(\KillSet{\cell})$ its fixed sign pattern with respect to $\KillSet{\cell}$.

Finally, we extend the sign pattern of each cell $C$ at the bottom of recursion
to a sign pattern with respect to all hyperplanes in $H$, as follows. We iterate over all
simplices along the path to the corresponding leaf-simplex $\cell$ in our overall hierarchical
structure, retrieve the subpatterns stored at them, and merge them into a single sign pattern
with respect to all of $H$. It is easy to verify that each hyperplane $h\in H$ appears
exactly once in this collection of subpatterns (for a fixed ``leaf-cell'' $C$)---it either happens
at the (unique) node at which $h$ has stopped belonging to the corresponding conflict list (it could
be one of the hyperplanes that define, and in particular touch, the present simplex), or else $h$ is in $\KillSet{\cell}$.
We store the merged sign pattern at the leaf of $T_{\KillSet{\cell}}$ corresponding to $C$. Alternatively, depending on the application
at hand, we might only store a link to some data structure (or just data) associated
with $C$. This storage is permanent, and is part of the
output structure, whereas the sign patterns of intermediate simplices
are discarded after their recursive processing terminates.

We note that an alternative approach, which saves on preprocessing, is to keep the intermediate
partial sign patterns at their respective nodes, and concatenate the patterns along the search
path of a query point $q$, during the query processing step, to obtain the complete sign pattern
at $q$. In this approach the resulting sign pattern is given as a list rather than a vector, which is provided by the previous approach.

\subsubsection{Answering a query}

A point-location query with a point $q$ is processed as follows. We
retrieve the initial random sample
$R$, and explicitly construct the simplex $\cell=\cell_q$ of $\BT(R)$
that contains $q$.  To do this, we first compute the sign pattern of
the cell $C_q$ of $\Arr(R)$ that contains $q$, namely, the side of each hyperplane
$h\in H$ that contains $q$, including the case where $q$ lies on $h$.
(If $q$ does lie on one or several hyperplanes of $R$ then $C_q$ is lower-dimensional.)
It is straightforward to do this in $O(d\rho)$ time. We then search the  trie $T_R$
with this sign pattern, and identify $C_q$ and  its bottom vertex $w=w_{C_q}$.


\begin{figure}[h]
  \centering
  \includegraphics{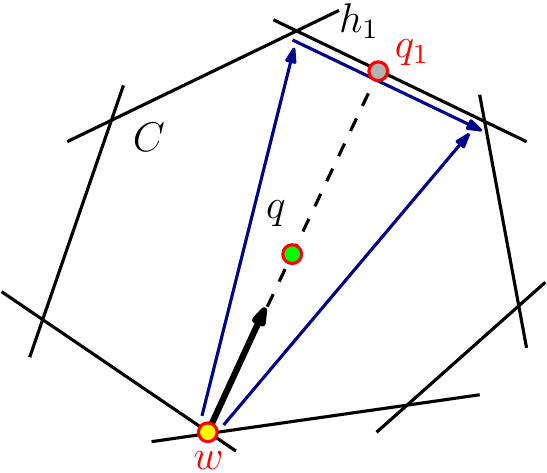}%
  \caption{The recursive construction of the simplicial cell containing $q$.}
  \figlab{BVT}
\end{figure}

We next perform ray shooting along the line $wq$, from $q$ in the direction away
from $w$, and find the first hyperplane $h_1$ of $R$ hit by this ray.
In case $C_q$ is lower-dimensional, the hyperplane $h_1$ is the first \emph{new} hyperplane
that the ray hits, namely a hyperplane not containing $q$. We assume, for simplicity of
presentation, that there are no ties, that is, there is a unique (new) hyperplanes first
hit by $\vec{wq}$. This can be achieved, for example, by perturbing $q$ slightly within $C_q$.
Let $q_1$ denote the point at which $h_1$ is hit. Then we compute the sign pattern of $q_1$
with respect to $R$, and search $T_R$ again with this pattern, to identify the cell $C_{q_1}$
containing $q_1$, and its bottom vertex $w_{C_{q_1}}$. We keep collecting the vertices of $\cell$
in this manner, one vertex per dimension. After $d$ recursive steps (on the dimension),
we obtain all the vertices of $\cell$ (and the cells of progressively decreasing dimension
of which they are bottom vertices), in overall time $O(d\cdot d \rho) = O(d^2 \rho)$.
See~\figref{BVT}.

Note that the sequence of ray shootings, as just described, identifies the path in $Q_{C_q}$
that leads to the leaf $v_\cell$ associated with $\cell$ in $Q_{C_q}$. Indeed, the shootings
identify the sequence of faces of $C_q$, of progressively decreasing dimensions
$d-1$, $d-2$, $\ldots$, $0$, which support the corresponding faces of $\sigma$, and
this sequence of faces identifies the path in $Q_{C_q}$ that we need to follow---at each
visited node we go to the child associated with the new hyperplane that supports the
next face in the sequence. At $v_\cell$ we find the recursive structure associated with
$\cell$ and continue the query recursively in this structure. The recursion terminates when
we reach a leaf-simplex $\cell$. We then locate the cell of $\Arr(\KillSet{\cell})$
containing $q$, from its sign pattern with respect to $\KillSet{\cell}$, which we compute
in brute force in $O(d\rho)$ time, and return the overall sign pattern that is stored at this leaf.

\subsubsection{The cost of a query}
Each step of the main recursion (on the size of the input) takes $O(d^2\rho)$ time.
This bound has already been noted for the cost of the ray shootings, and it also dominates
the cost of the search for the leaf representing the simplex of $\BT(R)$ containing $q$.
Thus the overall cost of a query is $O(d^2\rho)$ times the number of recursive steps
in the main hierarchy.

\begin{figure}[t]

    \centerline{%
       \begin{tabular}{|c|c|c|c|}
         \hline
         Parameter / method & Value $\Bigl.$ & Sample size & Using \\
         \hline\hline%
         VC-dim.  &
           \MCY{0.15}{$O(d^2 \log d)$\newline
           \lemref{v:c:dim}%
           }
         &
           \MCY{0.2}{%
           $O( r d^2 \log r \log d)$%
           }
         &
           \MCY{0.15}{%
           $\eps$-net theorem\newline
           (\thmref{epsilon:net})%
           }
         \\%
         \hline%
         Shatter dim.  &
           \MCY{0.15}{%
           \smallskip%
           $O(d^2)$\newline
           \lemref{h:s:growth}%
           }
         &
           \MCY{0.15}{%
           $O( r d^2 \log( rd ) )$%
           }
         &
           \remref{epsilon:net:s}%
         \\
         \hline
         \hline
         Combinatorial dim.
         &%
           \MCY{0.15}{%
           $b=d(d+3)/2$
           \newline%
           \remref{simp:comb:dim}%
           \newline%
           $\alpha = O(1)$, $\beta=d$
           \newline%
           \lemref{simplices:grow}%
           }%
         & $O(rd^2)$ & \lemref{weak:cutting} \\
         \hline
         Optimistic sampling
         & As above & $\Bigl.O(rd^2)$ & \corref{optimistic} \\
         \hline
       \end{tabular}%
    }
    \caption{The different sizes of the sample $R$, needed to ensure that the corresponding $\BVT(R)$
       in $\Re^d$ is, with high probability, a $(1/r)$-cutting, according to the different tools at our disposal.
       The probability depends on the (absolute) constant parameters in the bounds.}
    \figlab{simplices:cutting}

\end{figure}

Ignoring for the moment the issue of the storage and preprocessing costs (namely, their
dependence on $\rho$), the main issue is how large should $\rho$ be to make the query efficient.
On one hand, we would like
to make $\rho$ as small as possible, so that the cost $O(d^2\rho)$ of each recursive step is small,
but we cannot take $\rho$ to be too small, for then we lose control over the recursion depth,
as the problem size might then not decrease sufficiently, or not at all.
The different sample sizes that are needed to ensure (with high probability) that the resulting
bottom-vertex triangulation is a $(1/r)$-cutting for $\Arr(H)$, according to the
different sampling theories at our disposal, as reviewed and developed in \secref{sect2},
are summarized in the table in \figref{simplices:cutting}.
As the table shows, the various techniques for ensuring the sample quality
do not really differ that much from one another. Still,
the smallest sample size that ensures the $(1/r)$-cutting property is $\rho=O(rd^2)$,
using \lemref{weak:cutting} (or its optimistic sampling variant \corref{optimistic}
which, in the present context, does not make  a difference).
If the sample does not yield a $(1/r)$-cutting, we simply take another sample
and repeat the whole step. In an expected small number of steps we
get a sample that satisfies this property.

Assuming that the $(1/r)$-cutting property does indeed hold throughout the hierarchical structure,
the number of recursive steps used by the query is $\lceil \log_r (n/\rho) \rceil = \lceil \frac{\log (n/\rho)}{\log r} \rceil$, making the query cost
\begin{equation} \eqlab{q:bvt}
  Q(n) = O\pth{\frac{d^2\rho \log (n/\rho)}{\log r}}
  = O\pth{\frac{d^{4}r  \log n}{\log r}}
  = O\pth{ d^{4}  \log n},
\end{equation}
if we take $r=2$ (the best choice for making the query fast).

\begin{remark}
    In the original formulation of the algorithm, Meiser's analysis uses the
    VC-dimension, which is $O(d^2\log d)$, instead of the combinatorial dimension,
    and thereby incurs an (unnecessary) extra factor of $\log d$ in the query cost.
    Another issue is that using the $\eps$-net theorem (\thmref{epsilon:net}),
    instead of the combinatorial
    dimension approach, also requires an extra $\log r$ factor (this however does not
    arise if we use the primal shatter dimension instead). This is not an issue if we
    choose $r=2$, as we just did, but it becomes an issue for larger
    values of $r$, which are needed when we want to reduce the storage required
    by the algorithm; see below for details.
\end{remark}

\subsubsection{Storage}

Consider a sample $R$ and the storage required for the trie $T_R$ and the trees
$Q_C$ for each cell $C$ in $\Arr(R)$, but excluding the storage required for the recursive (or leaf)
substructures associated with the simplices of $\BT(R)$.
The space taken by $T_R$ is
$O\pth{\rho K_d(\rho)}$, where $K_d(\rho)$ is the maximum number of cells, of all
dimensions, in an arrangement of $\rho$ hyperplanes in $\Re^d$, which is
 $O(\rho^d)$ by \lemref{num:cells}. That is $T_R$ requires $O(\rho^{d+1})$ storage.
The space taken by the trees $Q_C$ is $O\pth{dS_d(\rho)}$, where
$S_d(\rho)$ is the maximum number of simplices, of all
dimensions, in the bottom vertex triangulation of  an arrangement of $\rho$ hyperplanes in $\Re^d$, which is
 $O(\rho^d)$ by \lemref{simplices:grow}.
It follows that the total space required for the data structures associated with $R$ is
$O(\rho^{d+1})$.

In addition, we do construct the conflict list $\KillSet{\cell}$ of each simplex $\cell$ of $\BT(R)$,
and the sign pattern of $\cell$ with respect to the hyperplanes of $H\setminus \KillSet{\cell}$,
but we keep them only temporarily, while the recursive processing of $\cell$ takes place.
Once this processing  is over, we discard $\clX{\cell}$, as the queries that reach $\cell$ make no use of
it---they only access the random sample chosen from $\clX{\cell}$ and its associated
recursive data structure, which we do store at the leaf corresponding to $\cell$.
The sign pattern of $\cell$ with respect to the hyperplanes of $H\setminus \KillSet{\cell}$
is integrated in the sign patterns of the cells of the arrangements of  the leaf-simplices generated by the
recursive recursive processing of $\cell$.

If $\cell$ is a leaf-simplex (of the whole structure) we store with it the trie $T_{\KillSet{\cell}}$ which we use to identify the cells in
the arrangement $\Arr(\KillSet{\cell})$.
At each leaf of $T_{\KillSet{\cell}}$
 we store
the sign pattern, with
respect to $H$, of the cell of $\Arr(H)$ that contains all the query points
in $\cell$ that reach this leaf.
(As discussed earlier, this storage can be saved in applications that do not need explicit access
to these sign patterns.)

Let $S_d(n)$ denote the maximum overall storage used by the algorithm
for an instance involving $n$ hyperplanes in $d$ dimensions, ignoring the storage used for
the sign patterns of the cells of the arrangements of the conflict lists of leaf-simplices.
We take, as above, $\rho = crd^2$, for some suitable absolute constant $c$, and redo the sampling until
the size of each recursive subproblem is at most $n/r$, and we get the following recurrence for $S_d(n)$.

\begin{align*}
  S_d(n) &\leq
  \begin{cases}
    a \rho^{d+1}  + e \rho^d S_d(n/r) , & \text{for $n > \rho $} \\
    a \rho^{d+1} & \text{for $n \le \rho$}  ,
  \end{cases}
\end{align*}
where $a$ is a suitable absolute constant (independent of $d$). Unfolding the recursion, we get
$$
S_d(n) \le a\rho^{d+1} \pth{ 1 + e\rho^d + \cdots + e^{j} \rho^{jd} } \le a'\rho^{d+1}(e\rho^d)^j,
$$
where $a'$ is another absolute constant, and
$j = \lceil \log_r (n/\rho) \rceil = \lceil \log (n/\rho)/\log r  \rceil$ is the depth of the recurrence.
Substituting this value of $j$ (and neglecting the rounding), we get that
\begin{equation} \label{storage-bdt}
S_d(n) = O\pth{\rho^{d+1} (e\rho^d)^{(\log (n/\rho))/\log r}   } = O\pth{ \rho^{d+1} (n/\rho)^{(d\log \rho+ \log e)/\log r} } =
O\pth{n^{(d\log \rho+ \log e)/\log r} } .
\end{equation}
In particular, for $r=2$ (so that $\rho=2cd^2$), we get the storage bound $O\pth{ n^{2d \log d + O(d)} }$.
The different bounds for storage and query time achievable according to this scheme are
depicted in the table in \figref{meiser:b:v:t}.
Note that, to get the near-ultimate bound $O(n^{d+\eps})$ for the storage
(for $n\gg d$, ignoring the coefficient that depends on $d$), which is
slightly larger but close to the bound $O(n^d)$ asserted (apparently wrongly) in Meiser~\cite{m-plah-93},
and established in Liu~\cite{l-nplah-04},
the query time bound becomes super-exponential in $d$. As a matter of fact, even to get
bounds like $n^{d+\kappa}$, for a controlled (constant) value of $\kappa$, we already get
super-polynomial bounds for the query cost.

\begin{figure}
    \centerline{%
       \begin{tabular}{|c|c||c|}
         \hline
         & Query time & Storage \\
         \hline
         $\Bigl. r$ & $O\bigl( {\frac{d^{4}r  }{\log r} \log n} \bigr)$%
         & $O\bigl( n^{d \log (crd^2)/ \log r + \log e/\log r} \bigr)$ \\%
         \hline
         \hline
         $\Bigl.2$ & $O\pth{d^{4} \log n}$ & $O\pth{ n^{2d \log d + O(d)} }$  \\%
         \hline
         $\Bigl. d$ & $O\bigl( \frac{d^{5}}{\log d}\log n \bigr) $ & $O\bigl(  n^{(3 + o_d(1))d} \bigr)$  \\%
         \hline%
         $\Bigl.d^{1/\eps}$ & $O\bigl({\frac{\eps d^{4+1/\eps}}{\log d}  \log n}\bigr)$%
         & $O\pth{n^{d + 2\eps d + o_d(1)}} $  \\%
         \hline
         $\Bigl.d^{d}$ & $O\bigl({\frac{ d^{d+3}}{\log d}  \log n}\bigr)$ &
           $O\bigl( n^{d + 2 +o_d(1)} \bigr)$  \\%
         \hline
         $\Bigl.d^{2d/\eps}$ & $O\bigl({\frac{\eps d^{2d/\eps+3}}{\log d} \log n }\bigr)$%
         & $O\bigl(n^{d + \eps + o_d(1)}\bigr)$  \\%
         \hline
       \end{tabular}
    }%
    \caption{Query time and storage for Meiser's data structure (using \BVT). The notation $o_d(1)$ in the various exponents means a term that depends on $d$ and tends   to $0$ as $d$ increases.}
    \figlab{meiser:b:v:t}%
\end{figure}

In addition, when the answer to a query is the sign pattern of the corresponding cell, we need
to add $O(n)$ storage at each leaf of the trie $T_{\KillSet{\cell}}$ of each leaf-simplex $\cell$,
for storing the sign
pattern (with respect to $H$) that serves as the output to all queries that reach that leaf.
This will increase the space by an additional factor of $n$.
Technically, unless $r$ is huge, this does not affect the asymptotic form of
the bound in Equation (\ref{storage-bdt}).

We also mention the follow-up study by Liu~\cite{l-nplah-04}, which has improved
the storage cost in terms of its dependence on $n$, using Chazelle's hierarchical
cutting technique (see~\cite{c-c-05}),
to the optimal value $O(n^d)$, except that the constant of proportionality is still
super-exponential in terms of $d$. Specifically, catering only to the case of the
best possible query time, Liu achieves query time\footnote{%
  Here the notation $\tilde{O}$ hides a polylogarithmic factor in $d$.}
of $\tilde{O}(d^5(\log{n} + d))$, but the constant hidden in the storage bound $O(n^d)$ is $d^{O(d^3)}$.
Although our storage bounds are weaker than Liu's, in terms of their dependence on $n$, the
dependence on $d$ in the coefficient of proportionality is significantly smaller---it actually
\emph{decreases} to $0$ as $d$ increases.

\subsubsection{Preprocessing} \seclab{bvt-preproc}

The nonrecursive preprocessing of a subproblem in the main recursion tree, with an associated $\rho$-sample $R$,
consists of (i) computing $\Arr(R)$, storing the sign patterns of its cells in the trie $T_R$, and computing
the bottom vertex of each cell, (ii) constructing $\BT(R)$ and the corresponding trees $Q_C$ for each cell
$C$ of $\Arr(R)$, and (iii) constructing the conflict list of each simplex (and the partial sign pattern
for hyperplanes not crossing the simplex).

We perform step (i) using the following simple vertex-based approach. We iterate over the vertices of $\Arr(R)$.
Fix a vertex $v$, and denote by $\delta(v)\ge d$ its \emph{degree}, namely the number of hyperplanes of $R$
incident to $v$. In general position, we have $\delta(v)=d$ and the procedure about to be described becomes much simpler.
To handle the general case, we intersect the $\delta(v)$ hyperplanes incident to $v$ with a hyperplane
$h_v$ parallel to the hyperplane $x_d = 0$ and passing slightly above $v$. We recursively compute
the cells of the $(d-1)$-dimensional arrangement, within $h_v$, of these intersections.
For a cell $C$ incident to $v$, $v$ is the bottom vertex of $C$ if and only if
$C\cap h_v$ is bounded. We thus collect all the bounded cells of the arrangement within $h_v$, of any
dimension, and associate each such cell $C'$ with the corresponding cell $C$ of $\Arr(R)$
(which is one dimension higher). The portion of $C$ between $v$ and $h_v$ is the pyramid
${\rm conv}(C'\cup\{v\})$. The sign patterns of $C'$ and of $C$, with respect to $R$, are identical.
In fact, all the cells within $h_v$ have the same sign pattern with respect to all the hyperplanes
not incident to $v$, and they can differ only in their signs with respect to the incident hyperplanes.

By construction, each cell of $\Arr(R)$, of dimension at least $1$, has a bottom vertex, either a real vertex
of the original arrangement, or an artificial one, on the artificial plane $\pi_d^-$ (see \secref{b:v:t}), so it will
be detected exactly once by the procedure just presented. The only cells of $\Arr(R)$ that do not have an associated
cell on some auxiliary hyperplane $h_v$, are the vertices themselves. We add each vertex $v$ to the output,
together with its sign pattern (which is $0$ at each incident hyperplane).

Let $Z_d(\rho)$ be the maximum time it takes to perform this computation on an arrangement of $\rho$
hyperplanes in $d$ dimensions. It follows from the preceding discussion that the full problem is reduced
to a collection of $(d-1)$-dimensional subproblems, one for each vertex of $\Arr(R)$. To prepare for these
subproblems, we need to construct the vertices, find the hyperplanes incident to each vertex $v$, and
prepare them for the recursive subproblem at $v$ by intersecting them with the hyperplane $h_v$.
To perform all these preparations, we iterate over the $\binom{\rho}{d}$ choices of $d$-tuples of hyperplanes,
compute the intersection vertex of each such tuple, in $O(d^3)$ time, using Gaussian elimination, identify
vertices that arise multiple times (and thereby obtain their degree), and then, for each vertex $v$,
intersect each incident hyperplane with $h_v$, in $O(d)$ time.
With a suitable implementation, the pre-recursive overhead takes
$$
O\left( d^3 \binom{\rho}{d} + d \sum_v \delta(v) \right)
$$
time. We have the following identity:
\begin{equation} \label{rho:d}
\sum_v \binom{\delta(v)}{d} \le \binom{\rho}{d} ,
\end{equation}
as the left-hand side counts the number of $d$-tuples of hyperplanes that have a singleton intersection
(a vertex), while the right-hand side counts all $d$-tuples of hyperplanes.
Using (\ref{rho:d}), the cost of the pre-recursive overhead is easily seen to be
$$
O\left( d^3 \binom{\rho}{d} \right) .
$$
In a post-recursive step, we need to compute the sign pattern of each cell with respect to the entire
set of $\rho$ hyperplanes. We do it by computing, for each vertex $v$, its sign pattern with respect to
all the hyperplanes not incident to $v$, in $O(\rho d)$ time. We then pad up this sequence with the local
sign pattern of each (bounded) cell constructed in the recursive call at $v$. To save on time (and storage),
we do not copy the global sign pattern (involving the non-incident hyperplanes) into the pattern of each
local cell. Instead we keep the global sign pattern as a separate entity, shared by all the local cells,
and just form, and store separately, the local sign pattern for each cell. Only at the bottom of recursion we will construct the
full sign pattern of each cell, as the union of the subpatterns, global and local, from each node along
the path to the corresponding leaf.

With this approach, the post-recursive overhead at $v$, which can also be applied before the recursive call,
only involves the computation of the global sign pattern. Hence the non-recursive overhead at $v$ is
$$
O\left( (d^3 + \rho d) \binom{\rho}{d} \right) .
$$
Note that each cell of $\Arr(R)$ arises exactly once as a local cell above a vertex $v$, namely at the
unique bottom vertex $v$ of the cell.

We thus obtain the following recurrence for $Z_d(\rho)$.
\begin{equation} \label{eq:arr-com}
Z_d(\rho) \le c_0 \left( (d^3 + \rho d) \binom{\rho}{d} \right) +
\max_\delta \sum_v Z_{d-1}(\delta(v)) \ ,
\end{equation}
where $c_0$ is some absolute constant, and where the maximum is taken over all possible assignments
of degrees to vertices (each such assignment must satisfy (\ref{rho:d})).

We claim that $Z_d(\rho) \le cd!\rho \binom{\rho}{d}$, for some absolute constant $c$. We will establish,
by induction on $d$, the refined bound $Z_d(\rho) \le c_d d!\rho \binom{\rho}{d}$, where the coefficients
$c_d$ form an increasing convergent sequence, from which the claim follows.

The case $d=1$ is easy, since there is no further recursion, and we only need to handle $O(\rho)$ vertices and edges,
each taking $O(\rho)$ time (mainly to compute its sign pattern). For $d>1$, the induction hypothesis implies that
\begin{align*}
Z_d(\rho) & \le c_0 \left( (d^3 + \rho d) \binom{\rho}{d} \right) +
\max_\delta \left( c_{d-1}(d-1)! \sum_v \delta(v) \binom{\delta(v)}{d-1} \right) \\
& = c_0 \left( (d^3 + \rho d) \binom{\rho}{d} \right) +
\max_\delta \left( c_{d-1}(d-1)! \sum_v \delta(v) \binom{\delta(v)}{d} \frac{d}{\delta(v) -d + 1} \right) \\
& \le c_0 \left( (d^3 + \rho d) \binom{\rho}{d} \right) +
\max_\delta \left(  c_{d-1}d! \rho \sum_v \binom{\delta(v)}{d} \right) \\
& = c_0 \left( (d^3 + \rho d) \binom{\rho}{d} \right) +
c_{d-1}d! \rho \binom{\rho}{d} \ .
\end{align*}
Hence, by putting
$$
c_d := c_{d-1} + \frac{c_0 (d^3+d)}{d!} \ge
c_{d-1} + \frac{c_0 (d^3+\rho d)}{d!\rho} ,
$$
for all $\rho \ge 1$, we establish the induction step. Clearly,
the coefficients $c_d$ form an increasing convergent sequence, as claimed.

Once we have computed all the cells of $\Arr(R)$ and their sign patterns,
it is straightforward to construct $T_R$ in $O(\rho^{d+1})$ time.

We perform step (ii) by locating the children of each cell $C'$ in brute force, as follows.
(Note that $C'$ may appear multiple times in each and all the trees $Q_C$, but we apply
the following procedure only once for each such cell.)
For each non-zero entry $b$ in the sign pattern of $C'$, we check whether setting $b$ to $0$ yields
a valid cell $C''$ on the boundary of $C'$. More precisely, in case of degeneracies, it might
be necessary to set more signs to $0$ for the resulting cell $C''$. To handle this issue, we
iterate over all choices of $j$ hyperplanes of $R$, for $j=2,\ldots,d$, form the intersection of these
$j$ hyperplanes, and collect all other hyperplanes that vanish identically on the resulting flat.
Hence, after setting $b=0$, we take all zero entries in the sign pattern of $C'$ (including $b$)
and find all the other hyperplanes whose sign should also be set to $0$, along with that of $b$.
It thus takes $O(\rho)$ time
to form the children of a cell $C'$ in any of the trees $Q_C$, which, by \lemref{num:cells},
takes a total of $O(\rho^{d+1})$ time for all cells. Once we know the children of every cell we
can assemble the trees $Q_C$ in time proportional to their size, which is\footnote{%
  The number of distinct cells is only $O(\rho^d)$, but because of their possible repetitions
  in the trees $Q_C$, we simply multiply by $d$ the number of leaves.}
$O(d\rho^{d})$.

In step (iii), we compute, for each vertex $y$ of $\Arr(R)$ and for each hyperplane $h$ of
$H$, the sign of $y$ with respect to $h$. These signs then allow us to decide whether $h$
is in $\KillSet{\cell}$ for each $\cell$ in $\BT(R)$ by checking whether
the sign of $h$ is positive for some of the $d+1$ vertices of $\cell$ and negative for others.
This costs $O(dn\rho^d)$ time and, as can be easily checked, constitutes the dominant part of the preprocessing.
We separate between the hyperplanes in $K(\sigma)$ and those that do nor cross $\sigma$. We (temporarily) store the sign pattern of $\sigma$ with respect to the second subset.

Let $T_d(n)$ denote the maximum expected overall preprocessing time of the algorithm
for an instance involving $n$ hyperplanes in $d$ dimensions. If we take, as above, $\rho = crd^2$,
for some suitable absolute constant $c$, then, with high probability, the size of each recursive
subproblem is at most $n/r$. If this is not the case, we discard the structure (that we have constructed
locally for $R$), take a new random sample, and construct the structure anew from scratch.

All these considerations lead to the following recurrence for $T_d(n)$.
\begin{align*}
  T_d(n) &\leq
  \begin{cases}
    a dn\rho^d  + e \rho^d T_d(n/r)  & \text{for $n > \rho $} \\
    a \rho^{d+1} & \text{for $n \le \rho$}  ,
  \end{cases}
\end{align*}
where $a$ is an suitable absolute constant (independent of $d$).
Unfolding the recurrence, we get
$$
T_d(n) \le adn\rho^d\pth{1 + \frac{e\rho^d}{r}+ \pth{\frac{e\rho^d}{r}}^2 +  \cdots + \pth{\frac{e\rho^d}{r}}^{j-1} }
+ a\left( e\rho^d\right)^j \rho^{d+1} \ ,
$$
for $j= \lceil \log (n/\rho) /\log r \rceil$. Ignoring the rounding, as before, it follows that
$$
T_d(n) = O\pth{ d\rho^{d+1} (n/\rho)^{(d\log \rho + \log e)/\log r} } = O\pth{ n^{(d\log \rho + \log e)/\log r} } \ ,
$$
which is the same as our asymptotic bound on $S_d(n)$. Note again that the coefficient of proportionality
is independent of $d$, and in fact tends to $0$ as $d$ increases.

\subsection{Improved algorithm via vertical decomposition}%
\seclab{meiser-vd}

In this subsection we improve upon Meiser's algorithm by using
vertical decomposition instead of bottom-vertex triangulation. This
allows us to use a smaller sample size, exploiting the smaller
combinatorial dimension of vertical prisms, thereby making the query run
faster.  We pay for it (potentially)\footnote{%
  Since the bounds on the complexity of vertical decomposition are not known to be tight,
  it is conceivable that we do not pay anything extra for storage using this technique.}
in the storage size, due to the discrepancy between the upper bounds on the complexities of
bottom-vertex triangulation (given in \lemref{simplices:grow}) and of vertical decomposition
(given in \thmref{vdsize}). Nevertheless, observing the data in \figref{meiser:b:v:t}, the bounds on the
storage size for \BVT, at least for the case of reasonably fast queries, are already rather
large, making the (possible additional) increase in the storage bound when using vertical
decomposition relatively less significant.
We will also use the enhancement suggested in the previous subsection, which is based
on optimistic sampling (unlike the case for \BVT, here using optimistic sampling does
make a difference), to improve the query time further, by another factor of $\log d$.
This further enhancement faces certain additional technical issues, so we  first
present the standard approach, and only then discuss the improved one.

\subsubsection{The data structure} \label{sec:ds}

The general strategy of the algorithm is similar to Meiser's: We draw a $\rho$-sample
$R\subset H$, construct the arrangement $\Arr(R)$, compute the sign pattern of each cell
(of any dimension) in $\Arr(R)$, and store the cells in a ternary trie $T=T_R$, indexed
by their sign patterns, exactly as in the previous subsection.
We then construct the vertical decomposition within
each cell $C$ of $\Arr(R)$ separately, a decomposition that we denote as $\VD(C)$,
by recursing on the dimension, following the constructive definition of vertical decomposition
given in \secref{v:de}.

Let $C= C_\xi\in \Arr(R)$ be the cell associated with the leaf $\xi$ of the trie $T$. We store
at $\xi$ (or at the structure associated with $C_\xi$, which is accessible from $\xi$) a
secondary tree $Q_C$ that ``encodes'' the  vertical decomposition $\VD(C)$ of the cell $C$ (described in more detail below).
Each leaf of $Q_C$ corresponds to a prism $\sigma\in \VD(C)$.

Once all these structures have been constructed, we take each prism $\cell\in\VD(C)$,
for every cell $C\in \Arr(R)$, compute its conflict list $\clX{\cell}$ (in brute force,
see below), and continue the main recursion (on the size of the input) on $\clX{\cell}$.
During the construction of $\clX{\cell}$, we also obtain the partial sign pattern of all
points in $\cell$ with respect to all the hyperplanes that do not cross $\cell$, and store the resulting partial sign pattern
temporarily at $\sigma$.

The recursion bottoms out when the conflict list is of size smaller than $\rho$. Exactly as
in the structure in \secref{meiser-bvt}, at each such leaf-prism we compute the
arrangement of the at most $\rho$  remaining hyperplanes, and store with each cell $C$ of
this arrangement a pointer to the unique cell of $\Arr(H)$ whose sign pattern is the
union of all partial sign patterns stored along the path to the leaf, including the sign
pattern of $C$ with respect to the hyperplanes of $\clX{\cell}$. This sign pattern is the answer to
all the queries that end up in $C$. The argument that justifies the last property is
identical to the one given for bottom vertex triangulation.

We next describe in more detail the construction of $VD(C)$ and the tree $Q_C$ encoding
it.

The sign pattern of $C$ identifies the set $R_0$ of the hyperplanes of $R$ that contain $C$ (if any),
the set $R^+$ of the hyperplanes of $R$ that pass above $C$ (and can only appear on its upper boundary), and
the set $R^-$ of the hyperplanes of $R$ that pass below $C$ (and can only appear on its lower boundary).
For simplicity, we only consider the full-dimensional case (where $R_0=\emptyset$); the lower-dimensional
instances are handled in  the same way, restricting everything to the flat $\bigcap R_0$ that
supports $C$. We iterate over all pairs $(h^-, h^+)$ of hyperplanes, such that $h^-\in R^-$ and
$h^+ \in R^+$, and for each such pair, we check whether $\bd C$ contains facets $f^- \subseteq h^-$
and $f^+ \subseteq h^+$ that are vertically visible within $C$. To do so, we construct
the following set of halfspaces within the hyperplane $\pi_0:\;x_d=0$.
For each $h\in R^+$, take the halfspace in $h^+$, bounded by $h^+\cap h$,
in which $h^+$ is below $h$ (in the $x_d$-direction), and project it onto $\pi_0$.
Similarly, for each $h\in H^-$, take the projection onto $\pi_0$ of
the halfspace in $h^-$, bounded by $h^-\cap h$, in which $h^-$ is above $h$.
Finally, take the projection of the halfspace in $h^+$ (or in $h^-$), bounded
by $h^+\cap h^-$, in which $h^-$ is lower than $h^+$. Denote the resulting set of halfspaces by $G$.
See Figure~\figref{vertical_visibility}.

The proof of the following lemma is straightforward and hence omitted.

\begin{figure}[h]
  \centering
  \begin{tabular}{cc}
    {\includegraphics{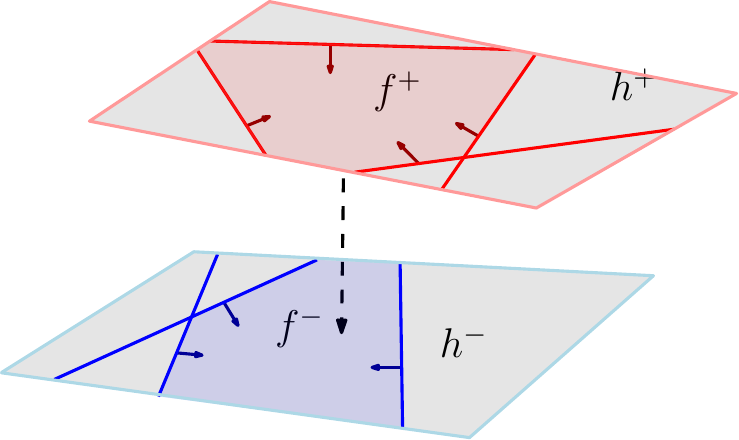} }%
    &\quad%
      {\includegraphics{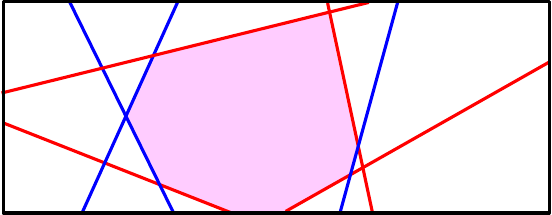} }  \\
    {\small (a)} & \quad\quad
                   {\small (b)}
  \end{tabular}
  \caption{
    (a) The halfspaces on $h^+$ (resp, $h^{-}$) are depicted by the straight lines and the arrows (indicating their direction).
    (b) The resulting set of halfspaces projected onto the hyperplane $\pi_0:\;x_d=0$. Their intersection is depicted by the
    shaded polygon.
  }
  \figlab{vertical_visibility}
\end{figure}

\begin{lemma} \lemlab{lem:feasible}
There are facets $f^-\subseteq h^-$ and $f^+ \subseteq h^+$ on $\bd C$ that are vertically
visible within $C$ if and only if the intersection of the halfspaces in $G$ is nonempty.
\end{lemma}

To apply this lemma, we construct the halfspaces within $\pi_0$, as prescribed in the lemma.
Denote by $R'$ the set of the $(d-2)$-hyperplanes (within $\pi_0$) that bound these halfspaces,
and note that $|R'|\le \rho-1$. We then apply linear programming (LP for short) to determine
whether the intersection of these halfspaces, denoted as $C'_{h^-,h^+}$, is nonempty.
Assume that $h^-$ and $h^+$ are indeed vertically visible within $C$. We then recurse on
the convex polyhedron $C'=C'_{h^-,h^+}$. For this we need the sign pattern of $C'$ with
respect to hyperplanes in $R'$ in the $x_{d-1}$-direction. We compute these signs
in brute force, in $O(d\rho)$ time, with respect to a witness point within $C'$, that
the LP procedure does provide.

We have omitted in this description details concerning the handling of unbounded cells that
do not have a lower or an upper boundary. Handling such cells is in fact simpler, because there is no need to pair up
a facet from the top boundary with a facet on the bottom boundary. Assuming that the top boundary
does not exist, we simply take each facet on the bottom boundary, intersect its hyperplane with all
other hyperplanes of $R$ (all lying below the cell), project, and recurse. We omit the further straightforward details.

We now construct the  tree $Q_C$. Its root represents $C$ and all the hyperplanes in $R$.
The root has a child for each pair of hyperplanes $h^-$ and $h^+$ that are vertically visible inside $C$.
Each child $v$, with a corresponding pair $h_v^-$, $h_v^+$ of vertically visible hyperplanes,
represents the polyhedron $C'_v=C'_{h_v^-,h_v^+}$ and the set $R'_v$ of the $(d-2)$-hyperplanes
in $\pi_0$ that surround $C'_v$, as specified in \lemref{lem:feasible}. The construction
of $Q_C$ proceeds recursively at each node $v$, with $C'_v$ and $R'_v$ as input.
At level $i$, each node $v$ represents a $(d-i)$-dimensional polyhedron $C'_v$ within the subspace
$x_d=0,\ldots,x_{d-i+1}=0$, and a corresponding set $R'_v$ of fewer than $\rho$ hyperplanes in this subspace.
Each node $v$ in $Q_C$ has at most $\frac14 \rho^2$ children, where each child $w$ is associated with
a pair $(h_w^-,h_w^+)$ of vertically visible hyperplanes of $R'_v$ inside $C'_v$.

By \lemref{def:hyp}, if $v$ is a child of a node $u$ that arises at dimension $j$,
each of $h_v^-$, $h_v^+$ is defined by a sequence of intersections and projections of
$2(d-j)$ original hyperplanes of $R$, and one additional respective original hyperplane
$\hat{h}_v^-$, $\hat{h}_v^+$. Since the other previous $2(d-j)$ hyperplanes are those that define $u$,
$v$ can be uniquely indexed from $u$ by the pair $(\hat{h}_v^-,\hat{h}_v^+)$ of original hyperplanes.
We therefore use these pairs to index the children of $u$. Using appropriate pointers,
this takes only $O(1)$ storage per child. The query will navigate through $Q_C$ using
these indexed links.

The recursion bottoms out at $d=1$, where the relevant cell $C'$, which is equal to its
trivial vertical decomposition, is just a (possibly unbounded) interval.
Each leaf $v_\sigma$ of $Q_C$ corresponds to a prism $\sigma$ in $\VD(C)$
which is defined by the pairs of hyperplanes associated with the nodes along the path
of $Q_C$ from the root to $v$.

It is noteworthy to compare this technique with the corresponding one used in \secref{meiser-bvt}
for bottom-vertex triangulation. The only significant difference is that here the parent-child
links are indexed by a pair of hyperplanes, whereas there a single hyperplane was used.

\paragraph{Constructing the conflict lists.}
The construction continues recursively at each prism $\cell$ in the vertical decomposition of $\VD(R)$,
with the conflict list of $\cell$ as input. To proceed, we construct, for each $\cell$,
the conflict list $\clX{\cell}$
of $\cell$ in brute force, by determining for each hyperplane $h\in H$, whether $h$ crosses $\cell$.
This is done using linear programming, regarding $\cell$ as the feasible region defined by (at most)
$2d$ linear inequalities, and regarding the normal direction $\nn$ of $h$ as an objective
function $\xx\cdot\nn$, for $\xx\in\cell$. By computing the minimum and the maximum values of this
function over $\cell$, we get two hyperplanes that are parallel to $h$ and support $\cell$. Then
$h$ crosses $\cell$ if and only if it lies (strictly) between these two supporting hyperplanes.

Note that constructing the conflict lists in the case of bottom-vertex triangulation was much simpler,
since we only had to deal with the original vertices of $\Arr(R)$. Here, in contrast, the prisms have
too many vertices, so constructing the conflict lists by checking which hyperplane separates the
vertices of which prisms is too expensive, and one has to resort to the LP-based technique (which is
different than the approach taken in~\cite{ES17}).

We do not store $\clX{\cell}$ explicitly after the preprocessing is over, but only maintain it
temporarily, as long as the recursive construction at $\cell$ is still going on. We only store
(permanently) the random sample $R_\cell$ drawn from $\clX{\cell}$ and
its associated data structures (defined recursively).
Similarly, we store
the partial sign pattern with respect to the complement of the conflict list only temporarily and discard it
once we have computed the complete sign patterns at the leaves.

\subsubsection{Answering a query}

The query with a point $q$ follows a path in the main hierarchical tree structure, where at each step
we have access to a random sample $R$, drawn from the conflict list of a parent prism, and we identify
the prism of $\VD(R)$ that contains $q$. We recursively search through the structure in this manner
until we reach a leaf, from which we retrieve, as in the case of $\BT$, the sign pattern of the cell of $\Arr(H)$ that contains $q$.

Consider a single step of this procedure. We compute the sign pattern of the cell $C$ of $\Arr(R)$
containing $q$ (with respect to the hyperplanes of the current random sample $R$); as before, this
is straightforward to do in $O(d\rho)$ time. We then locate the leaf $\xi$ of the top level trie
that is associated with  $C$ (and the tree $Q_C$ which encodes $\VD(C)$). This sign pattern also identifies
(i) the set of hyperplanes that pass above $C$ (and can contribute facets on its upper boundary),
(ii) the set of hyperplanes that pass below $C$ (and can contribute facets on its lower boundary),
and (iii) the set of hyperplanes that contain $C$ (if $C$ is lower-dimensional).

We then compute the hyperplane $h^-$ (resp., $h^+$) that is first hit by the
downward (resp., upward) $x_d$-vertical ray emanating from $q$;  this
too takes $O(d\rho)$ time.
The pair $(h^-,h^+)$ identifies the child
$v$ of the root of the search tree $Q_C$ in which our search has to continue.

To continue the search from $v$, we apply the filtering procedure of \lemref{unique_charge}.
That is, we compute the intersection hyperplanes $h\cap h^-$, for the hyperplanes
$h\in R\setminus \{h^-,h^+\}$ that pass below $C$, the intersection hyperplanes $h\cap h^+$,
for the hyperplanes $h\in R\setminus \{h^-,h^+\}$ that pass above $C$, and also include the
intersection $h^-\cap h^+$. We then project all these $(d-2)$-flats onto $x_d=0$.
This yields the set $R'_v$ of at most $\rho-1$ hyperplanes within $x_d=0$. We also take
the projection $q'$ of $q$ onto $x_d=0$, and compute the sign pattern of $q'$ with respect
to the set $R'_v$ and the $x_{d-1}$-direction.\footnote{%
  This $(d-1)$-dimensional sign pattern can also be computed and stored during preprocessing,
  but computing it on the fly, as we do here, does not affect the asymptotic cost of the query,
  and somewhat simplifies the preprocessing and storage.}
We continue the search with $q'$ within the cell $C'_v$, within $x_d=0$, that contains $q'$
(that is, the cell that has the sign pattern that we have just computed).

In general, the search reaches
 node $v$ of level $i$ in $Q_C$ when the query $q$ projected into
the $(d-i)$-dimensional subspace
$x_d=0,\ldots,x_{d-i+1}=0$ is contained the polyhedron $C'_v$ (also in this subspace) which
is represented by $v$.
To continue the search we find the
 hyperplanes $h^-$ (resp., $h^+$) that is first hit by the
downward (resp., upward) $x_j$-vertical ray emanating from $q'$ in
$C'_v$.
The hyperplanes $h^-$ and $h^+$ correspond to a pair of respective original hyperplanes $(\hat{h}^-, \hat{h}^+)$,
that were intersected with previous floors and ceilings along the path to $v$ so far
and projected into the subspace $x_d=0,\ldots,x_{d-i+1}=0$.
We continue the search with the child $w$ of $v$ in $Q_C$ that is indexed by $(\hat{h}^-, \hat{h}^+)$;
$w$ must exist by construction.

 We trace the path in $Q_C$ in this manner for the sole purpose of reaching
its leaf $w$, which represents the prism $\cell$ in $\VD(R)$ containing $q$.
This leaf  stores the next random sample $R_w$
and its associated data structures in which we continue the search.

When we reach a leaf-prism $\cell$ of the overall tree hierarchy (a prism whose conflict list is of size smaller than $\rho$), we compute the sign pattern of the query
with respect to the at most $\rho$ remaining hyperplanes stored at that leaf to identify the cell
of the arrangement $\Arr(\clX{\cell})$ that contains $q$. We locate the leaf
corresponding to $\cell$ in the trie associated with it and   return the sign pattern (with respect to all of $H$) stored there.

\begin{figure}[t]
    \centerline{%
       \begin{tabular}{|c|c|c|c|}
         \hline
         Parameter & Value $\Bigl.$ & Sample size & Using \\
         \hline\hline%
         VC-dim. &
           \MCY{0.15}{%
           $O(d^3)$\newline
           \thmref{v:c:dim:prisms}%
           }
           & $O( r d^3 \log r )$ &
           \MCY{0.15}{%
           $\eps$-net theorem\newline
           (\thmref{epsilon:net})%
           }
         \\
         \hline
         Shatter dim. &
           \MCY{0.15}{%
           $d(d+1)$\newline
           \thmref{v:c:dim:prisms}%
           }
         & $O( r d^2 \log (dr) )$ & \remref{epsilon:net:s} \\
         \hline \hline
         Combinatorial dim.
         &%
           \MCY{0.2}{%
           $b=2d$
           \newline%
           \corref{prisms:comb:dim}%
           \newline%
           $\alpha, \beta \leq 2d$
           \newline%
           \thmref{vdsize}%
           }
         & $O(rd \log (rd))$ & \lemref{weak:cutting} \\
         \hline
         Optimistic sampling &\MCY{0.2}{%
           $b=2d$
           \newline%
          $\alpha, \beta \leq 2d$
           \newline%
            Same references
           } & $\Bigl.O(rd)$ & \corref{optimistic} \\
         \hline
       \end{tabular}%
    }
    \caption{The different sizes of a sample $R$ needed to ensure, with high probability, that $\VD(R)$
       is a $(1/r)$-cutting for $R$, in $\Re^d$, according to the different tools at our disposal.
       The probability increases with the constant of proportionality.}
    \figlab{prism:cutting}
\end{figure}

\begin{remark}
Instead of answering point-location queries by returning the (sign pattern of the) cell of $\Arr(H)$
containing the query point $q$, as described above, and as in Meiser's algorithm, here it is somewhat
more natural to return the lowest hyperplane of $H$ (in the $x_d$-direction) that lies above the query
point $q$ (or a hyperplane that contains $q$). The ceiling of each prism containing $q$, along the
search path in the main tree hierarchy, is a candidate for the answer, and we return the vertically
closest hyperplane among these ceilings (and the hyperplanes in the leaf-subproblem).
The correctness of this procedure is easy to establish, and we omit its details.
\end{remark}

\subsubsection{The cost of a query}
As in Meiser's algorithm, each step of the main recursion (on the set of hyperplanes, passing
from some parent subset $H'$ to the conflict list of the vertical prism $\cell_q$ containing $q$
in the vertical decomposition of a $\rho$-sample $R$ from $H'$)
takes $O(d^2\rho)$ time. Indeed, the main operations performed in each dimension-reducing recursive step in the
construction of $\cell_q$ are (i) computing the sign pattern of the query point $q$ with respect to $R$,
(ii) searching the trie $T$ with the sign pattern to a leaf $\xi$ associated with the cell $C$ that
contains $q$, and (iii) searching $Q_C$, by identifying, at each node $v$ along this search path,
the pair $(\hat{h}^-,\hat{h}^+)$ of the original hyperplanes involved in the definition of the
floor and ceiling of the corresponding projected cell.
Steps (i) and (ii) take $O(d\rho)$ time, and step (iii) takes $O(d\rho)$ time at each node $v$,
using vertical ray shootings, for a total of $O(d^2\rho)$ time.
That is, the overall cost of a query is $O(d^2\rho)$ times the number of steps in the main recursion (on the size of the input).

It remains to determine how large should $\rho$ be to make the decomposition efficient,
that is, to make it be a $(1/r)$-cutting for a suitable parameter $r$.
Here we do not want to use the VC-dimension or the primal shatter dimension, since they
are both at least $\Omega(d^2)$. Instead, we use the smaller combinatorial dimension,
via the Clarkson--Shor analysis technique. See the table in \figref{prism:cutting} for a summary of
the various bounds, as derived earlier in this paper.

The best option in \figref{prism:cutting} (ignoring the issue of optimistic sampling, which will be
discussed later in this section is the one based on the combinatorial dimension $b=2d$
(see \corref{prisms:comb:dim}), so we take
\begin{equation}
    \eqlab{rcd}%
    \rho=cb r \log (b r) = 2cdr \log (2dr),
\end{equation}
where $c$ is some small absolute constant.  This choice guarantees that, with constant
probability (which can be made larger by increasing $c$), the conflict list of each
prism in the vertical decomposition of $\Arr(R)$ is of size at most $n/r$.
If we discover, during preprocessing, that this is not the case, we take,
as in the case of bottom-vertex triangulation, another sample, and repeat the whole step.
In an expected small number of steps we will get a sample that satisfies this property.
By applying this approach at each recursive step where we draw a random sample, we may
assume that this $(1/r)$-cutting property holds  at all sampling steps, over
the entire structure.

As before, this implies that the number of recursive steps is
$\log_r (n/\rho) = \frac{\log (n/\rho)}{\log r}$, making the query time
\begin{equation} \label{vd:qt}
  Q(n) = O\pth{\frac{d^2\rho \log( n/\rho)}{\log r}} =
  O\pth{\frac{d^3 r \log (dr) \log n}{\log r}} =
  O\pth{ d^3\log n \cdot \frac{r \log (dr)}{\log r}} .
\end{equation}
This replaces one factor of $d$ from the  bound using bottom  vertex triangulation in \Eqref{q:bvt} by the factor $\log(dr)$.
When we discuss optimistic sampling (see below) we show how to slightly improve this bound.


\subsubsection{Storage}
The storage required for the vertical decomposition of a $\rho$-sample $R$ of some
subset of the hyperplanes is estimated as follows. The trie $T$ has a leaf for each cell of $\Arr(R)$,
and the number of cells (of any dimension) is at most $e\rho^{d}$, by \lemref{num:cells}.
Each cell has a sign pattern of length $\rho$, so the total size of the trie $T$ is $O(\rho^{d+1})$.

By \thmref{vdsize}, the overall number of prisms of all dimensions (where prisms of dimension $j$
arise in the vertical decomposition within some $j$-flat formed by the intersection of a corresponding
subset of $d-j$ hyperplanes of $R$) is at most $c'\frac{4^d}{d^{7/2}} \rho^{2d}$ for some absolute constant $c'$.
Hence this also bounds the total number of leaves in the trees $Q_C$ and the branching factor of our global
hierarchical structure.

The depth of each leaf $v$ of any tree $Q_C$ that represents a $j$-dimensional prism is at most $j$,
for $j=0,\ldots,d$. Moreover, we store, at each internal node of each such tree, only the identifiers of
two (original) hyperplanes. Finally, there are no unary nodes in $Q_C$, because no cell (or any dimension
larger than $1$) can have a single facet on its top boundary and a single facet on the bottom boundary.
It therefore follows that the total size of the trees $Q_C$ is at most $O(\frac{4^d}{d^{7/2}} \rho^{2d})$.

We conclude that the total storage required for $\VD(R)$ (the trie $T$ and the associated
trees $Q_C$ of the cells $C$ in $\Arr(R)$) is
$$
O\left(\rho^{d+1} + \frac{4^d \rho^{2d}}{d^{7/2}} \right) \le c_0 \frac{4^d \rho^{2d}}{d^{7/2}} ,
$$
for some absolute constant $c_0$.

Let $S_d(n)$ denote the maximum overall storage used by the data structure
for an instance involving $n$ hyperplanes in $d$ dimensions, ignoring the storage used for
the sign patterns (with respect to the entire $H$) stored at the cells of the arrangements of the conflict lists of
the bottommost leaf-prisms. We take, as above, $\rho = 2cdr \log (2dr)$, for some suitable absolute constant $c$,
and repeat the sampling until the size of each recursive subproblem is $n/r$. This leads to
the following recurrence for $S_d(n)$.
\begin{align*}
  S_d(n) &\leq
  \begin{cases}
    c_0 \frac{4^d \rho^{2d}}{d^{7/2}}  + c'\frac{4^d\rho^{2d}}{d^{7/2}} S_d(n/r) , & \text{for $n > \rho $} \\
    a \rho^{d+1} & \text{for $n \le \rho$}  ,
  \end{cases}
\end{align*}
where $c'$ is another absolute constant, and the bound at the bottom of recursion is the same as for
bottom-vertex triangulation, with a suitable absolute constant $a$.
Unfolding the recurrence, and upper bounding both factors $c_0/d^{7/2}$ and $c/d^{7/2}$ by $1$,
for simplicity (this holds when $d$ is at least some sufficiently large constant), we get
$$
S_d(n) \le 4^d \rho^{2d} \pth{ 1 + 4^d\rho^{2d} + \cdots + (4^d\rho^{2d})^j } \le c''4^d \rho^{2d} (4^d\rho^{2d})^j ,
$$
where $c''$ is another absolute constant (very close to $1$), and
$j = \lceil \log_r (n/\rho) \rceil = \lceil \log (n/\rho)/\log r \rceil$ is the depth of the recurrence.
Substituting this value of $j$ (and neglecting the rounding) we get that
\begin{equation} \label{eq:sdn}
S_d(n) = O\pth{4^d \rho^{2d} (4^d\rho^{2d})^{\log (n/\rho)/\log r} } =
O\pth{4^d \rho^{2d} (n/\rho)^{\frac{2d + 2d\log \rho }{\log r}}} =
O\pth{ n^{\frac{2d + 2d\log \rho }{\log r}}} ,
\end{equation}
where $\rho = 2cdr \log (2dr)$. As is easily checked, the coefficient of proportionality is independent of $d$.

\subsubsection{Preprocessing}
\seclab{vd-prep-lp}

We construct $\Arr(R)$, as in the case of bottom-vertex triangulation.
That is, we compute the sign pattern of each vertex, and then of each cell of $\Arr(R)$
(with respect to $R$), and store these sign patterns in a trie $T=T_R$.
This takes $O( \rho^{d+1} )$ time by the procedure described in \secref{bvt-preproc}.

We next construct the trees $Q_C$. For each leaf $\xi$ of $T$, we take the cell $C=C_\xi$ of $\Arr(R)$
associated with $\xi$,  construct $\VD(C)$ and, in parallel, the tree $Q_C$,
by the dimension-recursive construction described above. Specifically, for each projected subcell $C'$,
at any dimension $j\le d$, with its associated set $R'$ of at most $\rho$ $(j-1)$-hyperplanes, we have
$O(\rho^2)$ potential floor-ceiling pairs of hyperplanes in $R'$. For each such pair $(h^-,h^+)$,
we determine whether $h^-$ and $h^+$ are vertically visible within $C'$, using the LP-based procedure described earlier.
%
%
%
Using the best known randomized sub-exponential algorithm for linear programming, as presented in
G\"artner and Welzl~\cite{gw-lp}, we can solve a linear program with $n$ constraints in $d$ dimensions in
$$
O\left( d^2n + e^{O(\sqrt{d\log d})} \right)
$$
expected time. It follows that the total expected time to perform the computation described above,
for all $O(\rho^2)$ pairs $(h^-,h^+)$, is
$$
O\pth{\rho^2 \cdot \left( d\rho+\left(d^2\rho + e^{O(\sqrt{d\log d})} \right) \right) } =
O\left(d^2\rho^3 + \rho^2e^{O(\sqrt{d\log d})} \right) .
$$
As we already argued, the number of nodes in all trees $Q_C$ is $O\pth{\frac{4^d \rho^{2d}}{d^{7/2}} }$.
Hence the overall cost of constructing the trees $Q_C$ is
\begin{equation} \label{eq:Qc1}
O\left( \frac{4^d \rho^{2d}}{d^{7/2}}\cdot \left(  d^2\rho^3 + \rho^2e^{O(\sqrt{d\log d})} \right) \right) .
\end{equation}
Once all the leaves of $Q_C$ (that is, prisms of $\VD(C)$), over all cells $C$ of $\Arr(R)$,
have been constructed, we proceed to construct the conflict list of each of these prisms.
Using linear programming once again, as described above, this takes
$O\left( n\left(d^2\rho + e^{O(\sqrt{d\log d})} \right) \right)$ expected time per prism,
for a total time of
\begin{equation} \label{eq:Qc2}
O\left( \frac{4^d \rho^{2d}}{d^{7/2}} n\left(d^2\rho + e^{O(\sqrt{d\log d})} \right)  \right) .
\end{equation}
This step also yields the partial sign pattern of each prism, with respect to the hyperplanes not corssing it, so the bound in
Equation \ref{eq:Qc2} also bounds the cost of the preparation of these partial patterns.

If any of the conflict lists is of size larger than $n/r$, we repeat the whole construction with a new sample.
Taking $\rho$ as in \Eqref{rcd}, the probability of such a failure is small, so this resampling approach
increases the expected running time by at most some small constant factor.

Let $T_d(n)$ denote the maximum overall preprocessing time of the algorithm
for an instance involving $n$ hyperplanes in $d$ dimensions.
By Equations (\ref{eq:Qc1}) and (\ref{eq:Qc2}), we get the following recurrence for $T_d(n)$.

\begin{align*}
  T_d(n) &\leq
  \begin{cases}
    a \frac{4^d \rho^{2d}}{d^{7/2}} (n+\rho^2)\left(d^2\rho + e^{O(\sqrt{d\log d})} \right) +
    b \frac{4^d \rho^{2d}}{d^{7/2}} T_d(n/r) & \text{for $n > \rho $} \\
    a \rho^{d+1} & \text{for $n \le \rho$}  ,
  \end{cases}
\end{align*}
where $a$ and $b$  are suitable absolute constants (independent of $d$).
Unfolding the recurrence, and upper bounding $a/d^{1/2}$ and $b/d^{7/2}$ by $1$, for simplicity, we get
$$
T_d(n) \le a' 4^d \rho^{2d}  n\left(\rho + e^{O(\sqrt{d\log d})} \right)
\left( \frac{4^d \rho^{2d}}{r} \right)^j + a'4^d \rho^{2d+2}\left(\rho + e^{O(\sqrt{d\log d})} \right) (4^d \rho^{2d})^j,
$$
where $a'$ is  another absolute constant and $j= \lceil \log (n/\rho) /\log r \rceil$.
Substituting this value of $j$ (and neglecting the rounding) we get
\begin{equation} \label{eq:tdn}
T_d(n) = O\pth{ 4^d \rho^{2d+2}  \left(    \rho + e^{O(\sqrt{d\log d})}\right)
   (n/\rho)^{ \frac{2d+2d\log \rho }{\log r}} } =  O\pth{n^{ \frac{2d+2d\log \rho}{\log r}} } ,
\end{equation}
where $\rho = 2cdr \log (2dr)$.

We can therefore conclude with the following main result of this section.

\begin{theorem}
    \thmlab{meiser:v:d}%
    Given a set $H$ of $n$ hyperplanes  in $\Re^d$, and a
    parameter $r>1$, one can construct a data-structure for point location
    (or vertical ray-shooting) in $\Arr(H)$ that answers a query in time
    $O\pth{ d^3\log n \cdot \frac{r \log (dr)}{\log r}}$.
    The bounds on the storage and expected preprocessing costs of the structure are given
    in (\ref{eq:sdn}) and (\ref{eq:tdn}), respectively, where $\rho = 2cdr \log (2dr)$.
\end{theorem}

The query and storage bounds of \thmref{meiser:v:d}, according to the chosen value of
$r$, are depicted in the table in \figref{meiser:v:d:st}.

{\scriptsize
\begin{figure}[htbp]
    \centerline{%
       \begin{tabular}{|c||c||c|}%
         \hline
         & {Query time} & {Storage} \\
         \hline
         $\Bigl. r$ & $O\pth{ d^3 \log n \cdot \frac{r\log(dr)}{\log r}}$
         & $O\left(n^{(2d + 2d\log\rho)/\log r} \right)$ \\%
         \hline
         \hline
         $\Bigl.2$ & $O\pth{d^{3} \log d \log n}$ &
         $O\left(n^{2d(\log d + \log\log 2d + 2+\log c)} \right)$ \\
         \hline
         $\Bigl.d$ & $O\bigl( d^{4} \log n \bigr)$ &
         $O(n^{4d(1+o_d(1))})$ \\
         \hline%
         $\Bigl.d^{1/\eps}$ & $O\bigl(d^{3+1/\eps} \log n\bigr)$ &
         $O(n^{2d + 2d\eps(1+o_d(1)})$ \\
         \hline
         $\Bigl.d^{d}$ & $O\bigl(d^{d+3} \log n \bigr)$ &
         $O(n^{2d + 4 + o_d(1)})$ \\
         \hline
         $\Bigl.d^{4d/\eps}$ & $O\bigl(d^{4d/\eps+3} \log n \bigr)$ &
         $O(n^{2d + \eps(1+o_d(1)})$ \\
         \hline
       \end{tabular}
    }%
    \caption{Query and storage costs for our variant of Meiser's data structure,
      which uses vertical decomposition, with sample size $\rho=2cdr\log(2dr)$.}
    \figlab{meiser:v:d:st}%
\end{figure}
}

\begin{figure}
    \centerline{%
       \begin{tabular}{|c||c||c|}%
         \hline
         & {Query time} & {Storage} \\
         \hline
         $\Bigl. r$ & $O\pth{ d^3 \log n \cdot \frac{r}{\log r}}$
         & $O\left(n^{(2d + 2d\log\rho)/\log r} \right)$ \\%
         \hline
         \hline
         $\Bigl.2$ & $O\pth{d^{3} \log n}$
         & $O(n^{2d(\log d + 2 + \log c)} )$ \\%
         \hline
         $\Bigl.d$ & $O\bigl( \frac{d^{4}}{\log d} \log n \bigr)$
         & $O(n^{4d(1+o_d(1))})$ \\%
         \hline%
         $\Bigl.d^{1/\eps}$ & $O\bigl({\frac{\eps d^{3+1/\eps}}{\log d} \log n}\bigr)$
         & $O(n^{2d + 2d\eps(1+o_d(1))})$ \\%
         \hline $\Bigl.d^{d}$ & $O\bigl({\frac{ d^{d+2}}{\log d} \log n}\bigr)$
         & $O(n^{2d + 2 + o_d(1)})$ \\%
         \hline
         $\Bigl.d^{2d/\eps}$ & $O\bigl({\frac{\eps d^{2d/\eps+2}}{\log d} \log n }\bigr)$
         & $O(n^{2d + \eps + o_d(1)})$ \\%
         \hline
       \end{tabular}
    }%
    \caption{Query and storage costs for the optimistic version of the structure,
      with $\rho = O(dr + \log\log n)$. The table only depicts the resulting bounds when
      $dr$ dominates $\log\log n$.}
    \figlab{meiser:v:d}%
\end{figure}


\subsubsection{Optimistic sampling} \seclab{s:opt}

We can slightly improve the query time further, by using a slightly smaller random sample.
Set $\BadProb = 1/(2h)$, where $h = \frac{\log (n/\rho)}{\log r}$ is the maximum recursion depth.
The strategy is to take a sample of size
$$
\rho = O\left(r\left(b + \ln\frac{1}{\BadProb}\right)\right) = O(dr + \log \log n) .
$$
For a fixed query point $q$, \corref{optimistic} implies that, with probability at
least $\geq 1- \BadProb$, the prism containing $q$ in the vertical decomposition
of the arrangement of the random sample has at most $n/r$ elements in its conflict list.

We modify the construction of the data structure, so that the recursion continues into
a prism only if its conflict list is sufficiently small (i.e., of size at most $n/r$),
in which case we call the prism {\em light}. That is, we continue the construction only
with the light prisms. In order to be able to answer queries that get stuck at a
\emph{heavy} prism, we build $u := \beta d^2\log (2n)$ independent copies of the data-structure,
for a suitable sufficiently large constant $\beta$.
Clearly, for a query point $q$, there are at most $h$ nodes in the search path of the
main hierarchical structure in a single copy, and the probability that at least one of
the prisms associated with these nodes is heavy is at most $h \cdot \BadProb \leq 1/2$.
If this happens, we say that the current copy of the data structure fails for $q$, and
we move to the next copy, starting the query process from scratch at this copy.
Clearly, the expected number of nodes that the query process visits, over all copies,
is $O(h)$, and the probability that the query, with a specific point $q$, succeeds
in at least one copy is at least $1-1/2^u \geq 1- 1/(2n)^{\beta d^2}$.

To continue the analysis, we need the following lemma. In its statement, two points $q$, $q'$
are said to be \emph{combinatorially equivalent} if, for any choice of random samples, at all
nodes of the main recursion, if $q$ and $q'$ land at the same prism, at each node along their
(common) search path.
\begin{lemma} \lemlab{num:opt:q}
The number of combinatorially distinct queries, i.e., the number of classes in the
combinatorial equivalence relation, is at most $(2n)^{2d^2}$.
\end{lemma}
\begin{proof}
To establish the lemma, we construct the collection of all possible hyperplanes that can
bound a prism, in the decomposition of the arrangement of any subset of $H$. Since each
prism is bounded by at most $2d$ hyperplanes, the desired number is at most $2d$ times
the number of all possible such prisms. By \corref{vd:all:p}, this number is
$$
N := O\left( d\cdot O\pth{ \frac{4^d}{d^{7/2}} n^{2d} } \right) =
O\left( \frac{4^d}{d^{5/2}} n^{2d} \right) .
$$
We now form the overlay of all these hyperplanes. It is clear from the construction
that for each cell of the overlay, all its points are combinatorially equivalent, so the
number of desired equivalence classes is at most the complexity of the overlay of $N$
hyperplanes in $\Re^d$, which, by \lemref{num:cells}, is
$$
\sum_{i=0}^d \binom{N}{i}2^i \leq 2 \pth{\frac{2Ne}{d}}^d \le (2n)^{2d^2} ,
$$
as is easily checked.
\end{proof}

\lemref{num:opt:q} implies that, choosing $\beta > 2$, the resulting structure
will answer correctly, with probability at least $1- 1/(2n)^{(\beta-2)d^2}$, all possible queries.

The new expected query time is
\begin{align*}
  Q(n) = O\pth{\frac{d^2\rho \log n}{\log r}}
  = O\pth{\frac{d^2(dr + \log \log n) \log n}{\log r}}
  = O\pth{ d^2\log n \cdot \frac{dr + \log \log n}{\log r}} .
\end{align*}
This expected query time is $O(d^3r \log n / \log r)$ when $\log \log n = O(dr)$, which is an
improvement by a $O(\log (dr))$ factor over the previous (deterministic) bound.

\smallskip

\paragraph{Storage and preprocessing.}
The preprocessing proceeds exactly as in the general treatment of vertical decomposition,
in \secref{meiser-vd}, except that (i) the sample size is smaller, by a logarithmic factor,
and (ii) we need to construct $u=O(d^2\log n)$ independent copies of the structure. Otherwise,
the algorithm and its analysis are identical to those presented above. That is, the storage
used by the structure is
$$
S_d(n) = O\left( d^2\log n \cdot
 n^{\frac{2d + 2d\log \rho }{\log r}}
\right) ,
$$
and the preprocessing cost is
$$
T_d(n) = O\left( d^2\log n \cdot
 n^{\frac{2d + 2d\log \rho }{\log r}}
\right) ,
$$
for $\rho = O(dr + \log \log n)$.

\paragraph{Discussion.}
The improvement achieved by optimistic sampling comes at a cost. First, the algorithm is Monte Carlo,
that is, with some small probability, queries may fail. We do not know how to turn it into a Las Vegas
algorithm; the resampling approach that we have used earlier in this section is too inefficient, because
it requires that we test all possible combinatorially distinct queries for success,, and there are too
many such points, by \lemref{num:opt:q}.

Second, we only have an expected time bound for a query. We do not know how to turn it into a
high-probability bound without repeating the query $\Theta(\log (1/\varphi))$ times, if we want to ensure
the bound with probability at least $1-\varphi$,
which kills the improvement that we have for small values of $\varphi$.

Still, optimistic sampling is an interesting alternative to consider for the problems at hand.


\subsection{Point location in arrangements of low-complexity hyperplanes}
\seclab{sec:low}

Let $H$ be a set of $n$ \emph{low complexity} hyperplanes in $\reals^d$.
By this we mean that each $h\in H$ has integer coefficients, and the $L_1$-norm
of the sequence of its coefficients is at most some (small) parameter $w$;
we refer to the $L_1$-norm of $h$ as its \emph{complexity}.
We also assume that all the hyperplanes
of $H$ pass through the origin. (This involves no real loss of generality, as
we can identify $\reals^d$ with the hyperplane $x_{d+1}=1$ in $\reals^{d+1}$,
and replace each input hyperplane $h$ by the affine hull of $h\cup\{o\}$, where
$o$ is the origin in $\reals^{d+1}$.)

Low-complexity hyperplanes have been studied in the recent groundbreaking work of
Kane et al.~\cite{KLM17}. They have shown that the \emph{inference dimension}
of the set of all such hyperplanes (with integer coefficients and complexity
at most $w$) is $\delta = O(d\log w)$. Without getting into the details of this
somewhat technical notion (full details of which can be found in \cite{KLM17}),
this implies that a random sample $R$ of $2\delta$ hyperplanes of $H$ has the
following property.

Regard each $h\in H$ as a vector in $\reals^d$ (it should actually be a vector in
projective $d$-space, but we stick to one concrete real representation).\footnote{%
  These arbitrary choices of affine representations of projective quantities
  should give us some flexibility in the algorithm that follows. However, we
  do not see how to exploit this flexibility; neither does the machinery in \cite{KLM17}.}
Let $R-R$ denote the set $\{h-h' \mid h,h'\in R\}$.
Let $x$ be some point in $\reals^d$, and let $C(x)$ denote the relatively open
cell (actually, a cone with apex at the origin), of the appropriate dimension,
that contains $x$ in the arrangement $\A(R\cup (R-R))$. Then the expected number
of hyperplanes of $H$ that cross $C(x)$ is smaller than $|H|/2$.

Actually, Kane et al.~\cite{KLM17} also establish the stronger property, that
a random sample of $\rho:= O(\delta + d\log\delta) = O(d\log(dw))$ hyperplanes of $H$
is such that, with constant probability, \emph{every} $x\in\reals^d$ has the property
that $C(x)$ is crossed by at most $\frac78 |H|$ hyperplanes of $H$.

Note that $C(x)$ is uniquely determined by the subset $R^0(x)$ of
hyperplanes of $R$ that vanish at $x$, and by the sequences $R^-(x)$,
$R^+(x)$, where $R^-(x)$ (resp., $R^+(x)$) consists of the hyperplanes of $R$
that are negative (resp., positive) at $x$, so that each of these sequences
is sorted by the values $\langle h,x \rangle$, for the hyperplanes $h$ in the
respective sequence. Note also that if $h\in H$ does not cross $C(x)$ then it
has a fixed sign with respect to all the points in the cell.

Paraphrasing what has just been discussed, a random $\rho$-sample $R$ from $H$
has the property that the cell decomposition formed by $\A(R\cup (R-R))$ is a
$(7/8)$-cutting of $\A(H)$, with constant probability.

\paragraph{Preprocessing.}
We now apply the point-location machinery developed so far in \secref{sec:meiser}
using $\A(R\cup (R-R))$ as the cell decomposition. We  go briefly over this
machinery, highlighting mainly the new aspects that arise when dealing with this
kind of cell decomposition.

We construct a hierarchical tree structure similar to the one  in
\secref{meiser-bvt} and \secref{meiser-vd}.
Here the hierarchy consists only of tries, where  each trie indexes  cells of an arrangement
$\A(R\cup (R-R))$ for some random sample $R$.
The top trie  is associated with a random sample $R$ from
 the entire $H$.
Each leaf $\xi$ of this trie corresponds to a cell $C_\xi$ of $\A(R\cup (R-R))$, and points to
 a trie
 associated with a random sample $R_\xi$ from the corresponding conflict list $K(C_\xi)$.

Each random sample $R$ (we abuse the notation slightly and use $H$ to denote some conflict list in the hierarchical structure and $R$
the random sample  from $H$) is of size
$\rho = O(\delta + d\log\delta) = O(d\log(dw))$.
Each cell $C$ of $\A(R\cup (R-R))$ is identified by a \emph{compact sign pattern},
which, for an arbitrary point $x\in C$, consists of the set $R^0(x)$ (sorted, say, by the indices of its hyperplanes) and of the
two sequences $R^+(x)$, $R^-(x)$, sorted by the values $\langle h,x \rangle$ of
their hyperplanes. Clearly, this compact sign pattern is independent of the choice
of $x\in C$.
The trie $T_R$, at the present node of the structure, stores the compact sign patterns of the cells in $\A(R\cup (R-R))$.
That is, each parent-child link in $T_R$ is associated with some hyperplane of $R$, and the hyperplanes associated with the  edges on the path from the root of $T_R$ to a leaf $\xi$, associated with  cell
$C_\xi$,
appear in their order $R^0$,$R^-$, $R^+$,  appropriately delimited,
 in the compact sign pattern of $C_\xi$.
It follows that each
 node of $T_R$ may have up to $\rho$ children, each corresponding to a different hypeplane in $R$.\footnote{To traverse a compact sign pattern
 in $T_R$ in  constant time per hyperplane,  we  store the children of each
node in $T_R$ in a hash table.}

The construction of $\A(R\cup (R-R))$ is performed using the same vertex-based
scheme presented in  \secref{bvt-preproc}. (Here too we expect the arrangement to be degenerate,
and we handle vertices of high degree exactly as before.) We  store the
compact sign patterns of the cells in  $T_R$ and compute the conflict list $K(C)$ of each cell $C\in \A(R\cup (R-R))$, using
linear programming as in \secref{vd-prep-lp}.
Specifically, to
check, for a hyperplane $h\in H$, whether
$h$ crosses $C$, we observe that the compact sign pattern of $C$ defines it as
an intersection of at most $\rho$ halfspaces and hyperplanes, where the hyperplanes
are of the form $h=0$, for $h\in R^0$, and the halfspaces are of the form
$h_{i+1}-h_i \ge 0$, for all the pairs of consecutive elements $h_i$, $h_{i+1}$
in $R^+$ and in $R^-$. We then determine whether $h$ crosses $C$ using the LP-based method of \secref{vd-prep-lp}
 with this  set of linear inequalities.

By the aforementioned properties, as established in \cite{KLM17}, we have the property that,
with constant probability, the size of $K(C)$, for every cell $C$, is at most $\frac78 |H|$.
As before, we can ensure this property with certainty by discarding samples that do not
satisfy this property, and by resampling, at each node of the recursive structure, yielding,
almost surely, a structure that satisfies that the size reduction property holds at each of its nodes.
As before we keep each conflict list and the sign pattern of $C$ with respect to each hypeplane in $R\setminus K(C)$, only temporarily and discard it once
the recursive preprocessing terminates.
We only store
permanently the hierarchy of the tries.

The recursion bottoms out at each cell $C$ for which $K(C)$ is of size smaller than $\rho$.
In each such  leaf-cell $C$, we
construct the arrangement $\A(K(C))$ (there is no need to consider $K(C)-K(C)$ now),
in the same vertex-based manner as before, where each cell is represented by its sign
pattern with respect to $K(C)$, and store at each cell $C'\in \A(K(C))$ a pointer to the
sign pattern of (every point in) $C'$ with respect to the full original set $H$.
This sign pattern is simply the union of the partial sign patterns computed
(and stored) at the nodes of the path from the root to $v$, including the sign
pattern with respect to $K(C)$, as just computed; see \secref{bvt-preproc} for details.

We note that the structure here is simpler than its counterparts considered earlier,
since each of its nodes only has to store a trie; there is no need for the ``cell trees'' $Q_C$
that were attached to the leaves of the tries in the previous versions.

\paragraph{Answering a query.}
A query with a point $x$ is processed by following a path in the hierarchical structure.
At each level we take the sample $R$ and compute
the compact sign pattern of $x$ with respect to the hyperplanes of $R\cup (R-R)$. That
 is, we compute the sign of each $h\in R$ at $x$, separate the hyperplanes into
the sets $R^0(x)$, $R^+(x)$, $R^-(x)$, as defined above, and sort each of
$R^+(x)$, $R^-(x)$, in increasing order of the values $\langle h,x \rangle$ and sort $R^0(x)$ with respect to the indices of the hyperplanes.
We then search the trie $T_R$ with this compact sign pattern  to locate
the leaf $\xi$ representing the cell
$C_\xi$
of $\A(R\cup (R-R))$ that contains $x$.
Then we continue the search recursively at the trie stored
at $\xi$ (which is associated with a random sample out of $K(C_\xi)$).
When we reach a leaf-cell $C$ we search the trie
associated with $\Arr(K(C))$ for the appropriate cell $C'$ of this arrangement containing $x$, and return the sign pattern
associated with $C'$.

We search $O(\log n)$ tries. At each recursive step, computing the compact sign pattern of $x$ with respect to $R\cup (R-R)$ takes
$O(\rho\log \rho)$ linear tests, each taking $O(d)$ time, for a total of $O(\rho d\log \rho)$ time.
Searching the trie $T_R$ takes $O(\rho)$ time,
so the cost of the query at each visited node is $O(\rho d\log \rho)$.
With the choice $\rho=O(d\log(dw))$, the total cost of the query is
$$
O(\rho d\log \rho \log n) = O(d^2\log^2(dw)\log n) .
$$

\paragraph{Storage.}
Each trie
 of the main hierarchical structure indexes a random sample $R$, of size
$\rho = cd\log(dw)$, for a suitable absolute constant $c$. The size of the trie $T_R$
is at most $\rho$ times the number of its leaves, namely, the number of cells, of all
dimensions, in $\A(R\cup (R-R))$
which,  by \lemref{num:cells}, is $O((2e\rho^2/d)^d)$.
 So the storage used by $T_R$ is $O(\rho(2e\rho^2/d)^d)$.

 We thus obtain the following recurrence for the
maximum storage $S_d(n)$ needed for an input set of $n$ hyperplanes in $d$ dimensions.
$$
S_d(n) \le a \rho \left(\frac{2e\rho^2}{d} \right)^d +
b \left(\frac{2e\rho^2}{d} \right)^d \cdot S_d\left(\tfrac78 n\right) ,
$$
where $a$ and $b$ are absolute constants.
The cost of storage at a leaf node of the structure is only $O(\rho^d)$.
The solution of this recurrence is easily seen to be
$$
S_d(n) = O\left( \rho^d  b^j \left(\frac{2e\rho^2}{d} \right)^{dj} \right),
$$
for $j = \left\lceil \log(n/\rho)/\log(8/7)\right\rceil$. Ignoring rounding, we get
\begin{equation} \label{eq:special-sdn}
S_d(n) = O\left(\rho^d  (n/\rho)^{(d\log(2e\rho^2/d)+\log b)/\log(8/7)} \right) =
O\left(  n^{(d\log(2e\rho^2/d)+\log b)/\log(8/7)} \right) ,
\end{equation}
where the coefficient of proportionality tends to $0$ as $d$ increases.

The resulting bound, which is $n^{O(d\log d)}$, falls short off the ideal bound $O(n^d)$,
but is reasonably close to it, and is comparable with the previous off-ideal bounds for $r=2$,
although the constant of proportionality in the exponent is slightly larger,
(see \figref{meiser:v:d:st} and \figref{meiser:v:d}), topped with the fact that the query cost
here is faster (by roughly a factor of $d$) than the best previous query costs.

\paragraph{Preprocessing.}
At each step of the main structure, with its associated random sample $R$ or size $\rho$,
we perform the following steps:
(i) Construct the arrangement $\Arr(R\cup (R-R))$.
(ii) Compute the compact sign pattern of each cell.
(iii) Construct the trie $T_R$.
(iv) Construct the conflict list of each cell.

Step (i) is carried out as in \secref{bvt-preproc}, in
$O(\rho^{2(d+1)})$ time.
Creating the compact sign pattern of each of the $O((2e\rho^2/d)^d)$ cells of $\A(R\cup (R-R))$
takes $O(\rho \log \rho)$ time, for a total of $O((2e/d)^d \rho^{2d+1}\log \rho)$ time. It is also straightforward to create $T_R$
within the same amount of time.
Finally, we compute the conflict list of each cell by solving, for each hyperplane in $H$, a linear program with at most $\rho$ constraints in $d$ dimensions.
This takes, as in \secref{vd-prep-lp}, $O\left( n\left(d^2\rho + e^{O(\sqrt{d\log d})} \right) \right)$ expected time per cell of
$\Arr(R\cup (R-R))$,
for a total expected time of
$$
O\left( \left(\frac{2e\rho^2}{d} \right)^d n\left(d^2\rho + e^{O(\sqrt{d\log d})} \right)  \right) .
$$

Let $T_d(n)$ denote the maximum overall expected  preprocessing time of the algorithm
for an instance involving $n$ hyperplanes in $d$ dimensions.
We get the following recurrence for $T_d(n)$.

\begin{align*}
  T_d(n) &\leq
  \begin{cases}
    a \left(\frac{2e\rho^2}{d} \right)^d n\left(d^2\rho + e^{O(\sqrt{d\log d})}\right)  + a \rho^{2(d+1)}   +
    b \left(\frac{2e\rho^2}{d} \right)^d T_d\left(\tfrac78 n \right) & \text{for $n > \rho $} \\
    a \rho^{2(d+1)} & \text{for $n \le \rho$} ,
  \end{cases}
\end{align*}
for some absolute constants $a$ and $b$.
Unfolding the recurrence, and noting that the overhead is dominated by the $a\rho^{2(d+1)}$ term, we get that
$$
T_d(n) = O\left( \rho^{2(d+1)} b^j \left(\frac{2e\rho^2}{d} \right)^{jd} \right)
$$
where  $j = \left\lceil \log(n/\rho)/\log(8/7)\right\rceil$.
Substituting this value of $j$ (and neglecting the rounding), we get
\begin{equation} \label{eq:special-tdn}
T_d(n) = O\left(\rho^{2(d+1)}  (n/\rho)^{(d\log(2e\rho^2/d)+\log b)/\log(8/7)} \right) =
O\left(  n^{(d\log(2e\rho^2/d)+\log b)/\log(8/7)} \right) ,
\end{equation}
where the coefficient of proportionality tends to $0$ as $d$ increases.
In summary we have the following result.
\begin{theorem}
    \thmlab{special-klm}%
    Given a set $H$ of $n$ hyperplanes  in $\Re^d$ with vectors of coefficient of $L_1$-norm bounded by $w$, one can construct a data-structure for point location
 in $\Arr(H)$ that answers a query in
$O(\rho d\log \rho \log n) = O(d^2\log^2(dw)\log n)$ time.
    The bounds on the storage and expected preprocessing costs of the structure are given
    in (\ref{eq:special-sdn}) and (\ref{eq:special-tdn}), respectively, where  $\rho=O(d\log(dw))$.
\end{theorem}


 \providecommand{\CNFX}[1]{ {\em{\textrm{(#1)}}}}
  \providecommand{\tildegen}{{\protect\raisebox{-0.1cm}{\symbol{'176}\hspace{-0.03cm}}}}
  \providecommand{\SarielWWWPapersAddr}{http://sarielhp.org/p/}
  \providecommand{\SarielWWWPapers}{http://sarielhp.org/p/}
  \providecommand{\urlSarielPaper}[1]{\href{\SarielWWWPapersAddr/#1}{\SarielWWWPapers{}/#1}}
  \providecommand{\Badoiu}{B\u{a}doiu}
  \providecommand{\Barany}{B{\'a}r{\'a}ny}
  \providecommand{\Bronimman}{Br{\"o}nnimann}  \providecommand{\Erdos}{Erd{\H
  o}s}  \providecommand{\Gartner}{G{\"a}rtner}
  \providecommand{\Matousek}{Matou{\v s}ek}
  \providecommand{\Merigot}{M{\'{}e}rigot}
  \providecommand{\Hastad}{H\r{a}stad\xspace}
  \providecommand{\CNFCCCG}{\CNFX{CCCG}}
  \providecommand{\CNFBROADNETS}{\CNFX{BROADNETS}}
  \providecommand{\CNFESA}{\CNFX{ESA}}
  \providecommand{\CNFFSTTCS}{\CNFX{FSTTCS}}
  \providecommand{\CNFIJCAI}{\CNFX{IJCAI}}
  \providecommand{\CNFINFOCOM}{\CNFX{INFOCOM}}
  \providecommand{\CNFIPCO}{\CNFX{IPCO}}
  \providecommand{\CNFISAAC}{\CNFX{ISAAC}}
  \providecommand{\CNFLICS}{\CNFX{LICS}}
  \providecommand{\CNFPODS}{\CNFX{PODS}}
  \providecommand{\CNFSWAT}{\CNFX{SWAT}}
  \providecommand{\CNFWADS}{\CNFX{WADS}}


\end{document}